\documentclass[sigconf,edbt]{acmart-edbt2024}

\usepackage{algorithm}
\usepackage[noend]{algorithmic}
\usepackage{balance}
\usepackage{booktabs}
\usepackage{enumitem}
\usepackage{subcaption}

\newtheorem{problem}{Problem}

\newcommand{\parax}[1]{\noindent\textbf{#1.}}

\DeclareMathOperator*{\argmax}{arg\,max}

\setcopyright{rightsretained}

\begin{document}
\title{Balancing Utility and Fairness in Submodular Maximization}

\author{Yanhao Wang}
\orcid{0000-0002-7661-3917}
\affiliation{%
  \institution{East China Normal University}
  \city{Shanghai}
  \country{China}
}
\email{yhwang@dase.ecnu.edu.cn}

\author{Yuchen Li}
\orcid{0000-0001-9646-291X}
\affiliation{%
  \institution{Singapore Management University}
  \city{Singapore}
  \country{Singapore}
}
\email{yuchenli@smu.edu.sg}

\author{Francesco Bonchi}
\orcid{0000-0001-9464-8315}
\affiliation{%
  \institution{CENTAI Institute \city{Turin} \country{Italy}}
  \institution{Eurecat \city{Barcelona} \country{Spain}}
}
\email{bonchi@centai.eu}

\author{Ying Wang}
\affiliation{%
  \institution{East China Normal University}
  \city{Shanghai}
  \country{China}
}
\email{yingwang007@stu.ecnu.edu.cn}

\begin{abstract}
Submodular function maximization is a fundamental combinatorial optimization problem with plenty of applications -- including data summarization, influence maximization, and recommendation. In many of these problems, the goal is to find a solution that maximizes the average utility over all users, for each of whom the utility is defined by a monotone submodular function. However, when the population of users is composed of several demographic groups, another critical problem is whether the utility is fairly distributed across different groups. Although the \emph{utility} and \emph{fairness} objectives are both desirable, they might contradict each other, and, to the best of our knowledge, little attention has been paid to optimizing them jointly.

To fill this gap, we propose a new problem called \emph{Bicriteria Submodular Maximization} (BSM) to balance utility and fairness. Specifically, it requires finding a fixed-size solution to maximize the utility function, subject to the value of the fairness function not being below a threshold. Since BSM is inapproximable within any constant factor, we focus on designing efficient instance-dependent approximation schemes. Our algorithmic proposal comprises two methods, with different approximation factors, obtained by converting a BSM instance into other submodular optimization problem instances. Using real-world and synthetic datasets, we showcase applications of our proposed methods in three submodular maximization problems: maximum coverage, influence maximization, and facility location.
\end{abstract}

\maketitle

\section{Introduction}
\label{sec:intro}

The awareness that decisions informed by data analytics could have inadvertent discriminatory effects against particular demographic groups has grown substantially over the last few years \cite{JagadishBEGGS19}.
Researchers have been studying how to inject \emph{group fairness}, along the lines of \emph{demographic parity} \cite{Dwork12} or \emph{equal opportunity} \cite{HardtPNS16}, into algorithmic decision-making processes.
The bulk of this algorithmic fairness literature has mainly focused on avoiding discrimination against a sensitive attribute (i.e., a protected social group) in supervised
machine learning \cite{KearnsRW17}, while less attention has been devoted to combinatorial optimization problems such as ranking \cite{AsudehJS019}, selection \cite{Gummadi_selection}, allocation \cite{MashiatGRFD22}, and matching \cite{dgs2020, EsmaeiliDNSD22}.

In this paper, we study how to incorporate group fairness into \emph{submodular maximization} \cite{KrauseG14} (SM), one of the most fundamental combinatorial optimization problems.
Specifically, a set function is submodular if it satisfies the ``\emph{diminishing returns}'' property, whereby the marginal gain of adding any new item decreases as the set of existing items increases. More formally, a set function $f: 2^V \rightarrow \mathbb{R}_{\geq 0}$ defined on a ground set $V$ of items is submodular if $f(S \cup  \{v\}) - f(S) \geq f(T \cup  \{v\}) - f(T)$ for any two sets $S \subseteq T \subseteq V$ and item $v \in V \setminus T$.
This property occurs in numerous data science applications, such as data summarization \cite{BadanidiyuruMKK14}, influence maximization \cite{KempeKT03}, and recommendation \cite{ParambathVC18}.

Submodular maximization with a cardinality constraint is the most widely studied optimization problem on submodular functions.
For a monotone submodular function $f$ and a cardinality constraint $k \in \mathbb{Z}^{+}$, it aims to find a size-$k$ subset $S \subseteq V$ of items such that $f(S)$ is maximized.
Although this problem is generally NP-hard \cite{Feige98}, the classic greedy algorithm \cite{NemhauserWF78} yields a best possible $(1-1/e)$-approximation factor.

In many scenarios, the goal of a submodular maximization problem is to find a good solution to serve a large number of users, for each of whom the utility is measured by a monotone submodular function.
For such problems, the classic greedy algorithm is effective in maximizing the \emph{average utility} over all users (called the \emph{utility objective}) because it is a non-negative linear combination of monotone submodular functions, which is still monotone and submodular \cite{KrauseG14}.
However, when users are composed of different demographic groups, the utilities should be \emph{fairly distributed} across groups to avoid biased outcomes \cite{10.1093/jla/laz001}.
A typical choice for the \emph{fairness objective} is to improve the welfare of the least well-off group, as per Rawls's theory of justice \cite{Rawls71}, which advocates arranging social and financial inequalities so that they are to the greatest benefit of the worst-off.
The problem of maximizing the minimum utility among all groups coincides with \emph{robust submodular maximization} \cite{Krause08, Udwani18, TsangWRTZ19} (RSM), which is inapproximable within any constant factor in general but admits approximate solutions by violating the cardinality constraints (i.e., providing solutions of sizes larger than $k$).

Although the two objectives of \emph{utility} and \emph{fairness}, when taken in isolation, can be treated as SM and RSM instances, respectively, optimizing them jointly results to be a much more challenging task, as maximizing one can come at the expense of significantly reducing the other. To our best knowledge, there has been no prior investigation into optimizing both objectives jointly.

To fill this gap, we propose a new problem called \emph{Bicriteria Submodular Maximization} (BSM) to balance utility and fairness.
As is often the case when facing two opposite objectives \cite{OhsakaM21, GershteinMY21, WeiIWBB15}, the BSM problem requires finding a fixed-size solution to maximize the utility function subject to that the value of the fairness function is at least a $\tau$-fraction of its optimum for a given factor $\tau \in [0,1]$.

Our proposal can find application in many real-world submodular maximization problems, such as:
\begin{itemize}
  \item \textbf{Maximum Coverage (MC)} is a well-known NP-hard problem with broad applications in summarization and recommendation \cite{SahaG09, SerbosQMPT17}. Given a collection of sets over the same ground set, it selects a set of $k$ sets such that their union contains as many elements as possible. When the ground set is a set of individuals divided into different population groups, we also need the solution to cover individuals from every group proportionally \cite{AsudehBD20, BandyapadhyayBB21} for equitable representation. By combining both objectives, we formulate a new problem of maximizing the total coverage of all users and the minimum average coverage for all groups simultaneously.
  \item \textbf{Influence Maximization (IM)} arises naturally in the design of viral marketing campaigns in social networks \cite{KempeKT03}. The IM problem \cite{KempeKT03, LiFWT18} asks to find a set of $k$ ``seed'' users in a social network so that, if they are targeted by the campaign, the expected influence spread in the network is maximized. Recent studies \cite{TsangWRTZ19, BeckerCDG20} introduced group fairness into IM, aimed at balancing the spread of influence among different population groups to reduce information inequality. In such context, BSM corresponds to a new problem of achieving a better trade-off between the influence spread among all users and the influence spread within the least well-off group.
  \item \textbf{Facility Location  (FL)} aims at finding a set of items to maximize the total benefits for all users: it finds application in data clustering and location-based services (LBSs) \cite{BadanidiyuruMKK14, ThejaswiOG21, LindgrenWD16}. In submodular facility location problems, the benefit of an item set to a user is defined as the maximum benefit among all items in the set to the user. As an illustrative example, a set of service points deployed in a district can lead to higher benefits for citizens whose residences are closer to their nearest service points. To achieve group-level fairness, it should ensure that each group of users receives approximately the same benefits on average. As such, BSM finds a set of items that maximizes the benefits for all users and the least well-off group to balance utility and fairness in facility location.
\end{itemize}

\smallskip\parax{Contributions \& Roadmap}
We first formally define the bicriteria submodular maximization (BSM) problem: Given two functions $f, g$ to measure the \emph{average utility} for all users and the \emph{minimum utility} among all groups, respectively, and a balance factor $\tau \in [0, 1]$, it asks for a size-$k$ set $S \subseteq V$ to maximize $f(S)$, subject to that $g(S) \geq \tau \cdot \mathtt{OPT}_g$, where $\mathtt{OPT}_g = \max_{S' \subseteq V, \lvert S' \rvert = k} g(S')$. We then prove that BSM cannot be approximated within any constant factor for all $\tau > 0$ and $k \geq 1$ (Section~\ref{sec:def}).

Due to the inapproximability of BSM, we focus on designing efficient instance-dependent approximation schemes for BSM. Specifically, we first propose a two-stage greedy algorithm combining the two greedy algorithms for \emph{submodular maximization} \cite{NemhauserWF78} and \emph{submodular cover} \cite{Wolsey82}, for maximizing $f$ and $g$ in turn to obtain a BSM solution (Section~\ref{subsec:alg2}).
Then, we devise an improved algorithm with a better approximation ratio at the expense of violating the cardinality constraint by converting a BSM instance into several \emph{submodular cover} \cite{Wolsey82} instances and running the \textsc{Saturate} algorithm \cite{Krause08} for submodular cover to find a BSM solution (Section~\ref{subsec:alg1}). We also analyze the approximation factors and computational complexities of both algorithms.
Moreover, we formulate specific classes of BSM problems, i.e., \emph{maximum coverage} and \emph{facility location}, as integer linear programs (ILPs), thus acquiring their optimal solutions on small instances using any ILP solver (Appendix~\ref{app:ilp}). This can help quantify the gap between the optimal and approximate solutions to BSM in practice.

Finally, we conduct extensive experiments to evaluate the proposed algorithms on three submodular maximization problems, namely \emph{maximum coverage}, \emph{influence maximization}, and \emph{facility location}, using real-world and synthetic datasets in comparison to non-trivial baselines, including approximation algorithms for SM \cite{NemhauserWF78}, RSM \cite{Krause08}, and submodular maximization under submodular cover \cite{OhsakaM21} (SMSC). Our experimental results demonstrate that our algorithms offer better trade-offs between the utility and fairness objectives than the competing algorithms and are scalable against large datasets (Section~\ref{sec:exp}).

Our main technical contributions are summarized as follows:
\begin{enumerate}
  \item The definition of the bicriteria submodular maximization (BSM) problem to balance utility and fairness in submodular maximization.
  \item The proof of the inapproximability of BSM.
  \item Two instance-dependent approximation algorithms for BSM with theoretical analyses of their approximation factors and complexities as well as problem-specific ILP formulations for BSM.
  \item Comprehensive experimental results to verify the effectiveness and efficiency of our proposed BSM framework and algorithms in real-world applications.
\end{enumerate}

\section{Related Work}
\label{sec:literature}

Before presenting our technical results, we first collocate our contribution to the literature by discussing the related work.

\smallskip\parax{Submodular Maximization (SM)}
Maximizing a monotone submodular function with a cardinality constraint $k$ is one of the most widely studied problems in combinatorial optimization.
A simple greedy algorithm \cite{NemhauserWF78} runs in $\mathcal{O}(nk)$ time and provides the best possible $(1-1/e)$-approximate solution for this problem unless $P=NP$ \cite{Feige98}.
Several optimization techniques, e.g., \emph{lazy forward} \cite{LeskovecKGFVG07} and \emph{subsampling} \cite{MirzasoleimanBK15}, were proposed to improve the efficiency of the greedy algorithm.
Any of the above algorithms can be used as a subroutine in our algorithmic frameworks for BSM.
Furthermore, there have been approximation algorithms for submodular maximization problems with more general constraints beyond cardinality, such as \emph{knapsack} \cite{TangTLHLY21}, \emph{matroid} \cite{CalinescuCPV11}, \emph{$k$-system} \cite{CuiHZZWH21}, and more \cite{MirzasoleimanBK16}.
In addition, submodular maximization problems were also studied in various settings, including the \emph{streaming} \cite{BadanidiyuruMKK14}, \emph{distributed} \cite{MirzasoleimanKS16}, \emph{dynamic} \cite{NEURIPS2020_6fbd841e}, and \emph{sliding-window} \cite{EpastoLVZ17, WangLT19, WangFLT17} models.
However, unlike our BSM problem, all the above problems only consider maximizing a single submodular objective function.

Previous work \cite{NEURIPS2020_9d752cb0, WangFM21} has also introduced fairness into submodular maximization problems.
Specifically, they considered that the items to select represent persons and are partitioned into several demographic groups based on sensitive attributes and proposed efficient approximation algorithms to find a set of items for maximizing a submodular function subject to the constraint that the number of items chosen from each group must fall within the pre-defined lower and upper bounds to achieve an equal (or proportional) group representation in the solution.
Nevertheless, such a notion of fairness is distinct from ours. In our BSM problem, the items are assumed to be non-sensitive; thus, no specific restriction is imposed on the attributes of items to pick in the solution. We alternatively focus on picking a fixed-size set of items to distribute utilities fairly across sensitive user groups. Due to the differences in data models and fairness notions, the algorithms in \cite{NEURIPS2020_9d752cb0, WangFM21} are not comparable to those in this paper and thus ignored in the experimental evaluation.

\smallskip\parax{Robust Submodular Maximization (RSM)}
The problem of \emph{robust (or multi-objective) submodular maximization} aims to maximize the minimum of $c > 1$ submodular functions.
The RSM problem is inapproximable within any constant factor in polynomial time unless $P = NP$ \cite{Krause08}.
The algorithms that provide approximate solutions for RSM by violating the cardinality constraints were proposed in \cite{Krause08, AnariHNPST19}.
The \emph{multiplicative weight updates} (MWU) algorithms for RSM, which achieves constant approximation factors when $c = o(k\log^{-3}{k})$, were developed in \cite{Udwani18, FuBBP21}.
In this paper, our second algorithm for BSM adopts a similar scheme to that for RSM in \cite{Krause08}.
But the analyses are different because an additional objective of maximizing the average utility is considered.

The deletion-robust submodular maximization problem \cite{MirzasoleimanKK17, KazemiZK18}, which aims to maintain a ``robust'' solution when a part of the data may be adversarially deleted, is also identified as RSM in the literature. However, this is a different problem, and the algorithms in \cite{MirzasoleimanKK17, KazemiZK18} do not apply to BSM.

\smallskip\parax{Submodular Optimization with More Than One Objective}
There exist several other variants of submodular optimization problems to deal with more than one objective.
Iyer and Bilmes \cite{IyerB12} studied the problem of minimizing the difference between two submodular functions.
Approximation algorithms proposed in \cite{JinYYSHX21, NikolakakiET21} aim to maximize the difference between a submodular utility function and a modular cost function.
Bai \emph{et al.} \cite{BaiIWB16} considered minimizing the ratio of two submodular functions.
Iyer and Bilmes \cite{IyerB13} investigated the problem of maximizing a submodular utility function while minimizing a submodular cost function.

The most related problem to BSM is \emph{submodular maximization under submodular cover} \cite{OhsakaM21} (SMSC), which finds a fixed-size solution to maximize one submodular function under the constraint that the value of the other submodular function is not below a given threshold.
A similar problem to SMSC was also defined by Gershtein \emph{et al.}~\cite{GershteinMY21} in the context of influence maximization.
BSM differs from SMSC as the minimum of $c > 1$ submodular functions is not a submodular function, and thus the theoretical results for SMSC \cite{OhsakaM21, GershteinMY21} do not hold for BSM.
Nonetheless, SMSC can be treated as a special case of BSM when $c = 2$, i.e., there are only two groups of users.
Thus, our experiments empirically compare SMSC with BSM in the case of $c = 2$.
The simultaneous optimization of the average and robust submodular objectives was also considered by Wei \emph{et al.}~\cite{WeiIWBB15}.
This work is different from BSM in two aspects: (1) the objective function is a linear combination of average and robust objective functions, to which an approximate solution might be unbounded for each of them; (2) the algorithms are specific for submodular data partitioning and not directly applicable for submodular maximization.

\smallskip\parax{Balancing Utility and Fairness in ML and Optimization}
The problem of balancing utility and fairness has also been widely considered for many machine learning and optimization problems other than submodular optimization. In supervised learning, the utility is typically measured by the difference in accuracy \cite{Dwork12, MehrabiMSLG21} or distribution of model output \cite{CalmonWVRV17} before and after fairness intervention, and the fairness is defined by \emph{demographic parity} \cite{Dwork12} or \emph{equal opportunity} \cite{HardtPNS16} for different groups. The overall goal is thus to improve the fairness metrics at little utility expense. In unsupervised learning, e.g., center-based clustering \cite{MakarychevV21, GhadiriSV21}, and optimization, e.g., allocation \cite{MashiatGRFD22} and matching \cite{EsmaeiliDNSD22}, similar to BSM, the utility is the same as for the original (fairness-unaware) problem, and the fairness is measured by the utilities for different protected groups. Various schemes have also been proposed to balance both objectives in those problems. However, although the above works are similar to our work in formulation, their algorithmic techniques do not apply to BSM, and vice versa, since none of them is submodular.

\section{Problem Statement}
\label{sec:def}

In this section, we first introduce the basic notions, then give the formal definition of \emph{Bicriteria Submodular Maximization} (BSM), and finally analyze the theoretical hardness of BSM.

For a positive integer $n$, we use $[n]$ to denote the set of integers $\{1,\ldots,n\}$.
Let $V$ be a set of $n$ items indexed by $[n]$ and $U$ be a set of $m$ users indexed by $[m]$.
For each user $u \in [m]$, we define a non-negative set function $f_u: 2^V \rightarrow \mathbb{R}_{\geq 0}$ to measure the utility of any subset $S \subseteq V$ of items for user $u$.
We assume that $f_u$ is normalized, i.e., $f_u(\emptyset) = 0$, monotone, i.e., $f_u(S) \leq f_u(T)$ for any $S \subseteq T \subseteq V$, and submodular, i.e., $f_u(S \cup \{v\}) - f_u(S) \geq f_u(T \cup \{v\}) - f_u(T)$ for any $S \subseteq T \subseteq V$ and $v \in V \setminus T$.
We define the following function to measure the utility of a set $S$ for each user in $U$ on \emph{average}:
\begin{equation}\label{eq:average}
  f(S) := \frac{1}{m} \sum_{u \in [m]} f_u(S)
\end{equation}
Following existing submodular maximization literature \cite{NemhauserWF78, BadanidiyuruMKK14}, we assume that the value of $f_u(S)$ is given by an oracle in $\mathcal{O}(1)$ time. Hence, the time complexity of computing $f(S)$ is $\mathcal{O}(m)$.

We consider the case in which the set $U$ of users is partitioned into $c$ (disjoint) demographic groups according to a certain sensitive attribute, e.g., \emph{gender} or \emph{race}, and $U_i \subseteq U$ is the set of users from the $i$-th group for $i \in [c]$.
We want to distribute the utilities across groups fairly to achieve group-level fairness.
Specifically, the average utility function of the $i$-th group is $f_i(S) := \frac{1}{m_i} \sum_{u \in U_i} f_u(S)$, where $m_i = \lvert U_i \rvert$.
Then, the widely adopted \emph{maximin fairness} \cite{dgs2020, dgs2021, TsangWRTZ19, FuBBP21}, which improves the conditions for the \emph{least well-off} group by maximizing the minimum average utility for any of the $i$ groups, is used as the notion of fairness.
This leads to the following function for measuring how fairly the utilities are distributed among groups:
\begin{equation}\label{eq:maximin}
  g(S) := \min_{i \in [c]} f_i(S) = \min_{i \in [c]} \frac{1}{m_i} \sum_{u \in U_i} f_u(S)
\end{equation}

Towards the balance between \emph{utility} and \emph{fairness}, we should consider optimizing both objective functions jointly: on the one hand, we should maximize $f$ for the effectiveness of the solution; on the other hand, we should maximize $g$ for the group-level fairness.
To achieve both goals, we generalize a common framework for bicriteria optimization \cite{OhsakaM21, GershteinMY21}.
Specifically, the first objective, i.e., maximizing $f$, is considered the primary objective of the new problem. Meanwhile, the second objective, i.e., maximizing $g$, is modeled as the constraint of the new problem, where a factor $\tau \in [0, 1]$ is used to restrict that the function value of $g$ must be at least $\tau$-fraction of its optimum to seek the balance between both objectives.
We define the \emph{Bicriteria Submodular Maximization} (BSM) problem as follows.
\begin{problem}\label{def:bs}[Bicriteria Submodular Maximization]
  Given an item set $V$, a user set $U$ partitioned into $c$ groups $U_1, \ldots, U_c$ with $\bigcup_{i=1}^c U_i = U$ and $U_i \cap U_j = \emptyset$ for any $i, j \in [c]$ ($i \neq j$), two functions $f$ and $g$ in Eqs.~\ref{eq:average} and~\ref{eq:maximin}, the cardinality constraint $k \in \mathbb{Z}^{+}$, and the balance factor $\tau \in [0, 1]$, return a set $S^* \subseteq V$ such that $\lvert S^* \rvert =k$, $f(S^*)$ is maximized, and $g(S^*) \geq \tau \cdot \mathtt{OPT}_g$, where $\mathtt{OPT}_g = \max_{S' \subseteq V, \lvert S' \rvert = k} g(S')$. Formally,
  \begin{align*}
   \max_{S \subseteq V, \lvert S \rvert = k} \;\; & \;\; f(S) \\
   \textnormal{subject to} \;\; & \;\; g(S) \geq \tau \cdot \mathtt{OPT}_g
  \end{align*}
\end{problem}
Here, $S^*$ is called an $\alpha^*$-approximate solution if $f(S^*) = \alpha^* \cdot \mathtt{OPT}_f$ where $\mathtt{OPT}_f = \max_{S \subseteq V, \lvert S \rvert = k} f(S)$.
Accordingly, $\alpha^*$ is called the best achievable factor for the BSM instance.

\begin{figure}[t]
  \centering
  \includegraphics[width=0.6\linewidth]{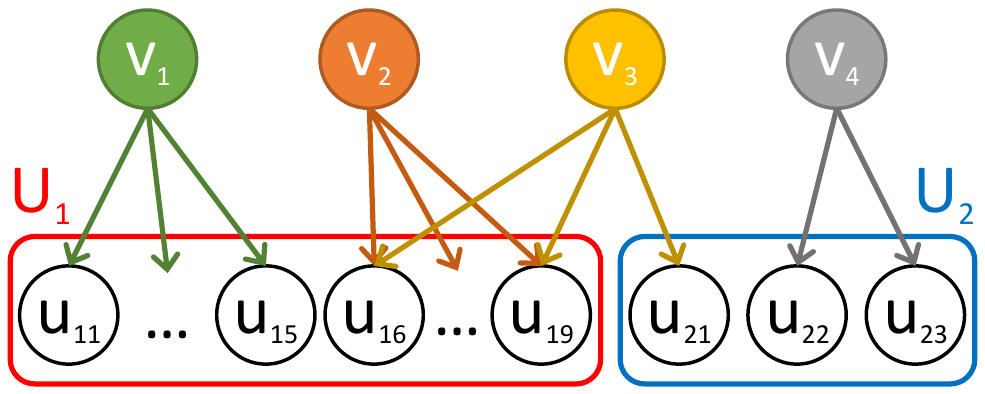}
  \caption{Running example for BSM based on an instance of maximum coverage.}
  \label{fig:example}
\end{figure}

\begin{example}
  \label{example:BSM}
  Let us consider the BSM instance for maximum coverage in Figure~\ref{fig:example}. It consists of four items $V = \{v_1, \ldots, v_4\}$ and a ground set $U$ of twelve users divided into two groups $U_1 = \{u_{11}, \ldots, u_{19}\}$ and $U_2 = \{u_{21}, u_{22}, u_{23}\}$. The sets of users covered by $v_1$, $v_2$, $v_3$, and $v_4$ are $S(v_1) = \{u_{11}, \ldots, u_{15}\}$, $S(v_2) = \{u_{16}, \ldots, u_{19}\}$, $S(v_3) = \{u_{16}, u_{19},$ $u_{21}\}$, and $S(v_4) = \{u_{22}, u_{23}\}$.

  The maximum coverage problem with a cardinality constraint $k = 2$ aims to select two items to cover the most number of users. Intuitively, it returns $S_{12} = \{v_1, v_2\}$ since $f(S_{12}) = \frac{1}{12} \cdot |S(v_1) \cup S(v_2)| = 0.75$ is the largest among all combinations of size-$2$ sets in $V$. The robust maximum coverage problem with the same constraint $k = 2$ aims to maximize the minimum of average coverage between $U_1$ and $U_2$. It alternatively returns $S_{14} = \{v_1, v_4\}$ with $\mathtt{OPT}_g = g(S_{14}) = \min_{i \in \{1, 2\}} f_i(S_{14}) = \min \{\frac{5}{9}, \frac{2}{3}\} \approx 0.556$. A BSM instance with balance factor $\tau \in [0, 1]$ finds a solution $S$ that maximizes $f(S)$ while ensuring that $g(S) \geq \tau \cdot \mathtt{OPT}_g$. When $\tau = 0$ (i.e., no constraint on $g$), the same solution $S_{12}$ as the vanilla maximum coverage problem is returned for BSM; when $0 < \tau \leq 0.6$, $S_{13} = \{v_1, v_3\}$ is returned for BSM since $g(S_{13}) = \frac{1}{3} \geq \tau \cdot \mathtt{OPT}_g$ and $f(S_{13}) = 8$ is the maximum among all size-$2$ sets satisfying the constraint on $g$; when $0.6 < \tau \leq 1$, the solution for BSM becomes $S_{14}$ since it is the size-$2$ set that has the largest coverage on $U$ and can satisfy the constraint on $g$.
\end{example}

We next analyze the hardness of BSM.
First, when $\tau = 0$, BSM is equivalent to maximizing $f$ subject to a cardinality constraint $k$, which is NP-hard and cannot be approximated within a factor better than $1-1/e$ (i.e., $\alpha^* \leq 1-1/e$) unless $P=NP$ \cite{Feige98}.
Furthermore, when $\tau = 1$, we need to compute the optimum $\mathtt{OPT}_g$ for maximizing $g$ with a cardinality constraint $k$ to ensure the satisfaction of the constraint.
This is an instance of \emph{robust submodular maximization}, which is generally NP-hard to approximate within any constant factor \cite{Krause08}.
Due to the intractability of computing $\mathtt{OPT}_f$ and $\mathtt{OPT}_g$, we consider a bicriteria approximation scheme for BSM.
Specifically, a set $S$ is called an $(\alpha, \beta)$-approximate solution to a BSM instance for some factors $\alpha, \beta \in (0, 1)$ if $f(S) \geq \alpha \cdot \mathtt{OPT}_f$ and $g(S) \geq \beta \tau \cdot \mathtt{OPT}_g$.
Finally, we will show by the following lemma that the objectives of maximizing $f$ and $g$ might contradict each other. A BSM instance with any $\tau > 0$ and $k \geq 1$ can be inapproximable within any constant factors $\alpha, \beta > 0$ even when $\mathtt{OPT}_f$ and $\mathtt{OPT}_g$ are known.

\begin{figure}[t]
  \centering
  \includegraphics[width=0.6\linewidth]{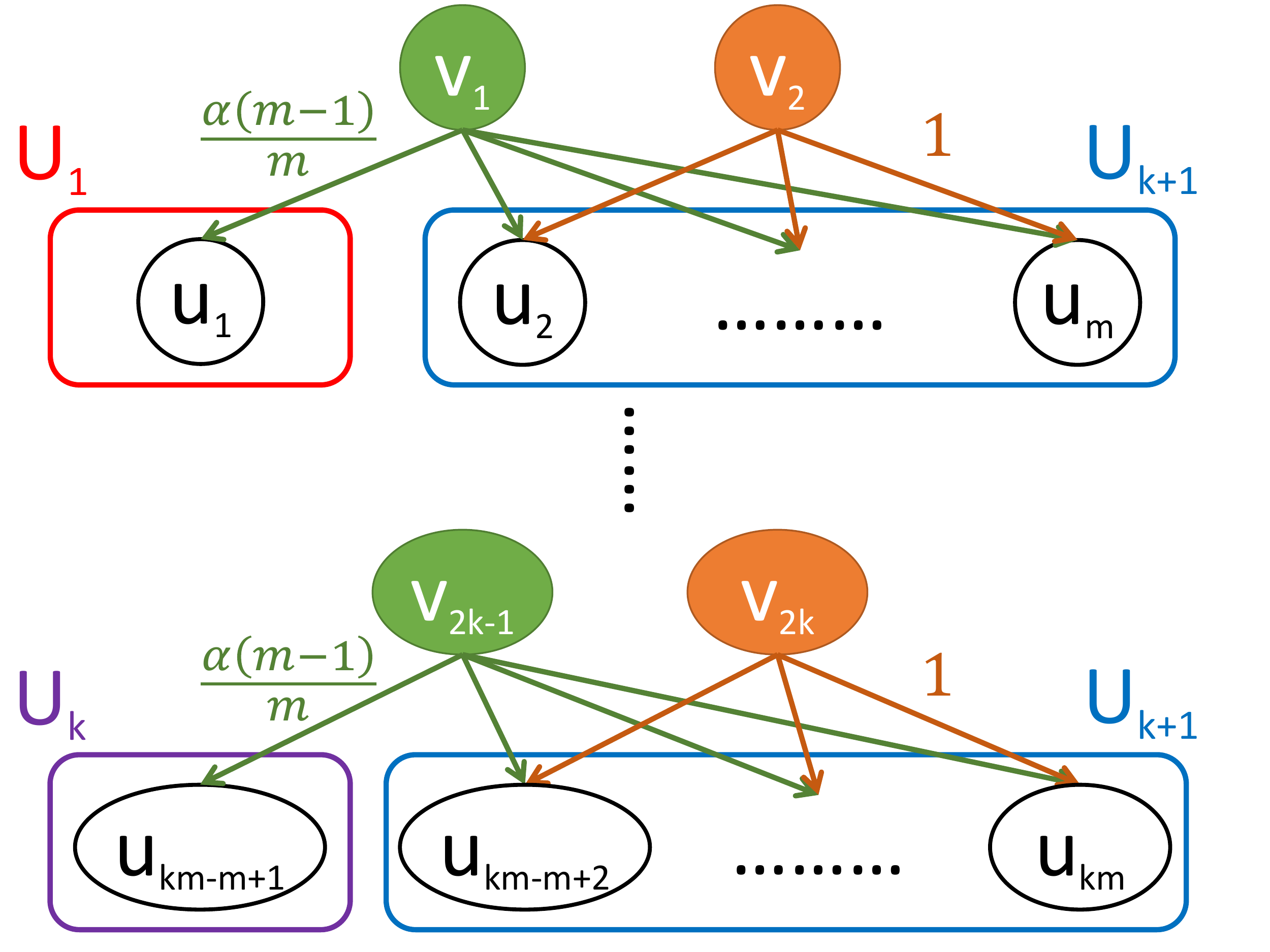}
  \caption{BSM instance for the proof of Lemma~\ref{lm:inapprox}.}
  \label{fig:hardness}
\end{figure}

\begin{lemma}\label{lm:inapprox}
  For any $\tau > 0$ and $k \geq 1$, there exists a BSM instance without any $(\alpha, \beta)$-approximate solution for constants $\alpha, \beta > 0$.
\end{lemma}

\begin{proof}
  We construct a BSM instance $\mathcal{I}$ that does not have any $(\alpha, \beta)$-approximate solution for all constants $\alpha, \beta > 0$ as shown in Figure~\ref{fig:hardness}.
  We first consider the procedure of constructing $\mathcal{I}$ for the case of $k=1$.
  We have an item set $V = \{v_1, v_2\}$ and a user set $U = \{u_1, u_2, \ldots, u_m\}$ with two groups $U_1 = \{u_1\}$ and $U_2 = \{u_2, \ldots, u_m\}$.
  Then, the utility function of $u_1$ is defined as
  \begin{equation*}
    f_{u_1}(S) =
    \begin{cases}
      \frac{\alpha (m - 1)}{m} \text{, if $v_1 \in S$} \\
      0 \text{, otherwise}
    \end{cases}
  \end{equation*}
  and the utility function of $u_j$ for each $j > 1$ is defined as
  \begin{equation*}
    f_{u_j}(S) =
    \begin{cases}
      1 \text{, if $v_2 \in S$} \\
      \frac{\alpha (m - 1)}{m} \text{, if $v_1 \in S$ and $v_2 \notin S$} \\
      0 \text{, otherwise.}
    \end{cases}
  \end{equation*}
  Obviously, $f_u$ is monotone and submodular for each $u \in U$.
  To maximize $f$ on $\mathcal{I}$ with $k=1$, $S_2 = \{v_2\}$ is the optimal solution and $\mathtt{OPT}_f = \frac{1}{m} \sum_{u \in [m]} f_u(S_2) = \frac{m - 1}{m}$.
  To maximize $g$ on $\mathcal{I}$ with $k=1$, $S_1 = \{v_1\}$ is the optimal solution and $\mathtt{OPT}_g = f_1(S_1) = \frac{\alpha (m - 1)}{m}$.
  However, as $g(S_2) = f_1(S_2) = 0$, the approximation factor $\beta$ of $S_2$ on $g$ for any $\tau > 0$ is $0$ and thus $S_2$ is a $(1, 0)$-approximate solution.
  In addition, $f(S_1) = \frac{\alpha (m - 1)}{m} = \alpha \cdot \mathtt{OPT}_f$.
  Therefore, $S_1$ is a $(\alpha, 1)$-approximate solution and the best achievable factor $\alpha^*$ of $\mathcal{I}$ is equal to $\alpha$.
  Since the value of $\alpha$ can be arbitrarily small, e.g., $\alpha \leq \frac{1}{m}$, there does not exist any $(\alpha, \beta)$-approximate solution to $\mathcal{I}$ for constants $\alpha, \beta > 0$ when $k = 1$.

  We further generalize the above construction procedure to the case of $k > 1$.
  Specifically, we replicate the above constructed instance $k$ times, denoted as $\mathcal{I}_1, \ldots, \mathcal{I}_k$.
  For each $\mathcal{I}_i$ where $i \in [k]$, we have an item set $\{v_{2i - 1}, v_{2i}\}$ and a user set $\{u_{i m - m + 1}, \ldots, u_{i m}\}$.
  Then, we assign the first user $u_{i m - m + 1}$ to group $U_i$ and all remaining users to group $U_{k + 1}$.
  The utility functions of all users are defined in the same way as in $\mathcal{I}$.
  For the new instance $\mathcal{I}'$, the sets of items and users are the unions of the sets of items and users for $\mathcal{I}_1, \ldots, \mathcal{I}_k$.
  Then, the optimal solution to maximizing $f$ with a cardinality constraint $k$ on $\mathcal{I}'$ is $S_{2i} = \{v_{2i} : i \in [k]\}$, and the optimal solution to maximizing $g$ with a cardinality constraint $k$ on $\mathcal{I}'$ is $S_{2i - 1} = \{v_{2i - 1} : i \in [k]\}$.
  We can see that, for every subset $S \subseteq V'$ with $|S| \leq k$ except $S_{2i - 1}$, there must exist some $i \in [k]$ such that $g(S) = f_i(S) = 0$.
  Hence, their approximation ratios on $g$ all equal $0$.
  Moreover, the approximation ratio of $S_{2i - 1}$ on $f$ is $\frac{f(S_{2i - 1})}{f(S_{2i})} = \alpha$.
  Therefore, $S_{2i - 1}$ is a $(\alpha, 1)$-approximate solution and the best achievable factor $\alpha^*$ of $\mathcal{I}'$ is also equal to $\alpha$.
  Accordingly, the inapproximability result holds for arbitrary $k \geq 1$.
\end{proof}

Given the above inapproximability result, we will focus on instance-dependent bicriteria approximation schemes for BSM in the subsequent section.

\section{Algorithms}
\label{sec:alg}

In this section, we propose our algorithms for BSM.
Our first algorithm is a two-stage greedy algorithm that combines the two greedy algorithms for \emph{submodular maximization} \cite{NemhauserWF78} and \emph{submodular cover} \cite{Wolsey82} (Section~\ref{subsec:alg2}).
Our second algorithm converts BSM into the \emph{submodular cover} \cite{Wolsey82} problem and adapts the \textsc{Saturate} algorithm \cite{Krause08} to obtain a solution to BSM (Section~\ref{subsec:alg1}).
We also analyze the approximation factors and complexities of both algorithms.

\subsection{The Two-Stage Greedy Algorithm for BSM}
\label{subsec:alg2}

By definition, BSM can be divided into two sub-problems, i.e., computing a solution $S$ that (1) maximizes $f(S)$ and (2) satisfies the constraint $g(S) \geq \tau \cdot \mathtt{OPT}_g$.
Therefore, an intuitive scheme to solve BSM is first to find a solution for each sub-problem independently and then to combine both solutions as the BSM solution.
The above scheme is feasible in practice because the prior and latter problems are instances of \emph{submodular maximization} and \emph{submodular cover}, respectively, on which \textsc{Greedy} algorithms \cite{NemhauserWF78, Wolsey82} can provide approximate solutions, despite of their NP-hardness \cite{Feige98}.
Accordingly, we propose a two-stage greedy algorithm called \textsc{BSM-TSGreedy}, which first runs the greedy submodular cover algorithm \cite{Wolsey82} to obtain an initial solution based on the constraint on $g$ and then uses the greedy submodular maximization algorithm \cite{NemhauserWF78} to add items to the initial solution for maximizing $f$. In this way, the final solution after two stages can not only approximately satisfy the constraint on $g$ due to the items in the first stage but also maximizes $f$ with an approximation factor depending on the number of items added in the second stage, thus achieving an instance-dependent bicriteria approximation as described in Section~\ref{sec:def}.

\begin{algorithm}[tb]
  \small
  \caption{\textsc{BSM-TSGreedy}}
  \label{alg2}
  \begin{algorithmic}[1]
    \REQUIRE Two functions $f, g : 2^V \rightarrow \mathbb{R}_{\geq 0}$, balance factor $\tau \in (0, 1)$, solution size $k \in \mathbb{Z}^{+}$
    \ENSURE Solution $S'$ with $\lvert S' \rvert = k$ for BSM
    \STATE Run \textsc{Greedy} \cite{NemhauserWF78} on $f$ to compute $S_f$ and $\mathtt{OPT}'_f$
    \STATE Run \textsc{Saturate} \cite{Krause08} on $g$ to compute $S_g$ and $\mathtt{OPT}'_g$
    \STATE Initialize $S' \gets \emptyset$
    \STATE Define $g'_{\tau}(S) := \frac{1}{c} \sum_{i \in [c]} \min \big\{ 1, \frac{f_i(S)}{\tau \cdot \mathtt{OPT}'_g} \big\}$
    \WHILE{$g'_{\tau}(S') < 1$ and $\lvert S' \rvert < k$}
      \STATE $ v^* \gets \argmax_{v \in V} g'_{\tau}(S' \cup \{v\}) - g'_{\tau}(S') $
      \STATE $ S' \gets S' \cup \{v^*\}$
    \ENDWHILE
    \IF{$\lvert S' \rvert = k$ and $g'_{\tau}(S') < 1$\label{ln-backup}}
      \STATE $S' \gets S_g$
    \ELSIF{$\lvert S' \rvert < k$}
      \FOR{$l \gets 1, \ldots, k$}
        \STATE Let $v_{f,l}$ be the $l$-th item added to $S_f$ by \textsc{Greedy}
        \STATE $S' \gets S' \cup \{v_{f,l}\}$
        \IF{$\lvert S' \rvert = k$}
          \STATE \textbf{break}
        \ENDIF
      \ENDFOR
    \ENDIF
    \RETURN{$S'$}
  \end{algorithmic}
\end{algorithm}

In particular, the procedure of \textsc{BSM-TSGreedy} is presented in Algorithm~\ref{alg2}.
First of all, it calls \textsc{Greedy} \cite{NemhauserWF78} for SM and \textsc{Saturate} \cite{Krause08} for RSM separately to obtain the approximations $\mathtt{OPT}'_f$ and $\mathtt{OPT}'_g$ for $\mathtt{OPT}_f$ and $\mathtt{OPT}_g$, respectively.
Then, in the first stage, a function $g'_{\tau}(S):= \frac{1}{c} \sum_{i \in [c]} \min \big\{ 1, \frac{f_i(S)}{\tau \cdot \mathtt{OPT}'_g} \big\}$ is defined based on $\tau$ and $\mathtt{OPT}'_g$. It holds that $g'_{\tau}$ is monotone and submodular because (1) $\overline{f}(t, S):= \min\{t, f(S)\}$ is monotone and submodular for any constant $t \geq 0$ and monotone submodular function $f$ \cite{KrauseG14} and (2) any nonnegative linear combination of monotone submodular functions is still monotone and submodular \cite{KrauseG14}.
Starting from $S' = \emptyset$, the item to maximize $g'_{\tau}$ is greedily added to $S'$ until $g'_{\tau}(S') = 1$, i.e., $g(S') \geq \tau \cdot \mathtt{OPT}'_g$ and the constraint on $g$ is (approximately) satisfied, or $\lvert S' \rvert = k$.
After that, if $g'_{\tau}(S') < 1$, the partial solution $S'$ will be replaced with $S_g$.
The purpose of this step is to guarantee that a size-$k$ solution satisfying $g'_{\tau}(S') = 1$ is provided, as $g'_{\tau}(S_g) = 1$ always holds from the definition of $g'_{\tau}$, in the case that the \textsc{Greedy} algorithm fails to find a size-$k$ solution $S'$ with $g(S') \geq \tau \cdot \mathtt{OPT}'_g$.
Then, it executes the second stage by further adding the items in the greedy solution $S_f$ for maximizing $f$ to $S'$ until $\lvert S' \rvert = k$ to increase the value of $f(S')$ as largely as possible.
Finally, the solution $S'$ after two stages is returned for BSM.

\begin{example}
  We consider how \textsc{BSM-TSGreedy} works on the BSM instance with $k = 2$ in Figure~\ref{fig:example}. As shown in Example~\ref{example:BSM}, it first runs \textsc{Greedy} and \textsc{Saturate} to compute $S_f = \{v_1, v_2\}$ with $\mathtt{OPT}'_f = 0.75$ and $S_g = \{v_1, v_4\}$ with $\mathtt{OPT}'_g = 5/9 \approx 0.556$. Accordingly, $g'_{\tau}(S) := \frac{1}{2} \big(\min\{ 1, \frac{9 f_1(S)}{5 \tau} \} + \min\{ 1, \frac{9 f_2(S)}{5 \tau} \} \big)$. When $\tau = 0.2$, it adds $v_3$ to $S'$ in the first stage since $g'_{0.2}(\{v_3\}) = 1$, then includes $v_1$ into $S'$ in the second stage because $v_1$ is the first item in $S_f$, and finally returns $S' = \{v_1, v_3\}$ for BSM. When $\tau = 0.5$, it also first adds $v_3$ to $S'$ since $g'_{0.5}(\{v_3\}) = 0.9$ is the largest among all four items. Then, either $v_1$ or $v_2$ can be added to $S'$ because $g'_{0.5}(\{v_1, v_3\}) = g'_{0.5}(\{v_2, v_3\}) = 1$. Assuming that $v_2$ is added, it will return $S' = \{v_2, v_3\}$ after the first stage, which is inferior to $S_{13}$ since $f(S') = 5 < f(S_{13}) = 8$, though they both satisfy the constraint on $g$. When $\tau = 0.8$, it still first adds $v_3$ to $S'$ with $g'_{0.8}(\{v_3\}) = 0.625$. Then, $g'_{0.8}(S') < 1$ and $|S'| = 2$ after any remaining item is added to $S'$. Thus, it returns $S' = S_g = \{v_1, v_4\}$ according to Line~\ref{ln-backup} of Algorithm~\ref{alg2}.
\end{example}

Subsequently, we analyze the approximation factor and time complexity of \textsc{BSM-TSGreedy} in the following theorem.
\begin{theorem}\label{thm:alg2}
  Algorithm~\ref{alg2} returns a $\big( 1 - \exp(- \frac{k'}{k}), 1 - \varepsilon_g \big)$-approximate solution $S'$ with $\lvert S' \rvert = k$ for an instance of BSM in $\mathcal{O}\big(n m k \log(c m) \big)$ time, where $k'$ is the number of iterations in the second stage, $\varepsilon_g \leq 1 - \frac{1}{m} - \frac{\widehat{\mathtt{OPT}_g}}{\mathtt{OPT}_g}$, and $\widehat{\mathtt{OPT}_g}$ is the optimum of maximizing $g$ with a cardinality constraint $\mathcal{O}\big( k \log^{-1}(cm) \big)$.
\end{theorem}
\begin{proof}
  First of all, Algorithm~\ref{alg2} either finds a solution $S'$ such that $g'_{\tau}(S') = 1$ and $g(S') \geq \tau \cdot \mathtt{OPT}'_g$ after the first stage, or replaces $S'$ with $S_g$, which always satisfies $g(S_g) = \mathtt{OPT}'_g \geq \tau \cdot \mathtt{OPT}'_g$.
  According to the analysis for \textsc{Saturate} in \cite[Thm.~8]{NguyenT21}, we have $\mathtt{OPT}'_g \in [(1 - \theta) \cdot \widehat{\mathtt{OPT}_g}, \mathtt{OPT}_g]$ where $\widehat{\mathtt{OPT}_g}$ is the optimum of maximizing $g$ with a cardinality constraint $\mathcal{O}\big(k \log^{-1}( \frac{c}{\theta}) \big)$.
  By setting $\theta = \frac{1}{m}$, we obtain that $\varepsilon_g \leq 1 - \frac{1}{m} - \frac{\widehat{\mathtt{OPT}_g}}{\mathtt{OPT}_g}$ and the approximation factor holds for a cardinality constraint $\mathcal{O}\big( k \log^{-1}(cm) \big)$.
  In the second stage, since the first $k'$ items of $S_f$ are added to $S'$, by applying the approximation analysis in \cite{NemhauserWF78} for $k'$ instead of $k$, it holds that $f(S') \geq (1 - \exp(- \frac{k'}{k})) \cdot \mathtt{OPT}_f$.
  Therefore, $S'$ is a $\big( 1 - \exp(- \frac{k'}{k}), 1 - \varepsilon_g \big)$-approximate solution of size $k$ for BSM.

  Moreover, \textsc{Greedy} and \textsc{Saturate} run in $\mathcal{O}(n m k)$ and $\mathcal{O}\big( n m k$ $\log(cm) \big)$ time, respectively.
  The time complexity of \textsc{Greedy} for maximizing $g'_{\tau}$ is $\mathcal{O}(n m k)$ since it always terminates after $k$ iterations no matter whether $g'_{\tau}(S') = 1$.
  Then, adding the first $k'$ items of $S_f$ to $S'$ takes $\mathcal{O}(k)$ time.
  Thus, the time complexity of Algorithm~\ref{alg2} is $\mathcal{O}\big(n m k\log(cm)\big)$.
\end{proof}

\subsection{The Saturate Algorithm for BSM}
\label{subsec:alg1}

The \textsc{BSM-TSGreedy} algorithm has two limitations. First, its approximation factor might be arbitrarily bad (e.g., dropping to $0$ when $k' = 0$). Also, how close its approximation factor is to the best achievable factor $\alpha^*$ is unavailable. Second, as will be shown in Section~\ref{sec:exp}, it suffers from significant losses in solution quality when the number of items added to $S'$ in the first stage for ensuring $g'_{\tau}(S') = 1$ is equal or close to $k$.
To achieve better performance for BSM, we alternatively consider applying a Lagrangian-like formulation similar to that for SMSC \cite{OhsakaM21}, which combines the maximization of $f$ and the satisfaction of the constraint on $g$ into a single problem, to convert BSM into \emph{submodular cover} \cite{Wolsey82} instances. Then, we run the \textsc{Saturate} algorithm \cite{Krause08} for submodular cover to provide a BSM solution that not only achieves a bicriteria approximation guarantee theoretically but also strikes a good balance between $f$ and $g$ empirically.

Next, we present the detailed conversion procedure from BSM to submodular cover.
Let us first consider the decision version of BSM defined as follows.
\begin{definition}[BSM Decision]
\label{def:bs:dec}
  For any approximation factor $\alpha \in (0, 1)$, determine if a set $S \subseteq V$ with $\lvert S \rvert = k$ such that $f(S) \geq \alpha \cdot \mathtt{OPT}_f$ and $g(S) \geq \tau \cdot \mathtt{OPT}_g$ exists.
\end{definition}
If the answer to the BSM Decision problem in Definition~\ref{def:bs:dec} is \emph{yes}, then there must exist an $\alpha$-approximate solution to the BSM instance and vice versa.
Assuming that $\mathtt{OPT}_f$ and $\mathtt{OPT}_g$ are known, the BSM decision problem for a given $\alpha \in (0, 1)$ can be divided into two sub-problems:
\emph{(i)} is there a size-$k$ set $S \subseteq V$ such that $f_{\alpha}(S) := \frac{f(S)}{\alpha \cdot \mathtt{OPT}_f} \geq 1$?
and \emph{(ii)} is there a size-$k$ set $S \subseteq V$ such that $g_{\tau}(S) := \frac{g(S)}{\tau \cdot \mathtt{OPT}_g} \geq 1$?
The BSM decision problem is thus transformed to decide whether the objective value of the following problem equals $2$:
\begin{equation}\label{eq:obj}
  \max_{S \subseteq V, \lvert S \rvert = k} \min \{1, f_{\alpha}(S)\} + \min \{1, g_{\tau}(S)\}.
\end{equation}
Then, the problem of Eq.~\ref{eq:obj} is reduced to \emph{submodular maximization} according to the truncation method as used in \cite{Krause08}:
\begin{equation}\label{eq:sm}
  \max_{S \subseteq V, \lvert S \rvert = k} F_{\alpha}(S) := \min\big\{1, \frac{f(S)}{\alpha \mathtt{OPT}_f} \big\} + \frac{1}{c} \sum_{i \in [c]} \min \big\{1, \frac{f_i(S)}{\tau \mathtt{OPT}_g} \big\}.
\end{equation}
Intuitively, $F_{\alpha}$ in Eq.~\ref{eq:sm} is monotone and submodular because it is a nonnegative linear combination of monotone submodular functions.
Since computing the optimums $\mathtt{OPT}_f$ and $\mathtt{OPT}_g$ is NP-hard, we further consider replacing them with approximate values $\mathtt{OPT}'_f$ and $\mathtt{OPT}'_g$, i.e., $\mathtt{OPT}'_f \in [(1-\varepsilon_f) \cdot \mathtt{OPT}_f, \mathtt{OPT}_f]$ and $\mathtt{OPT}'_g \in [(1-\varepsilon_g) \cdot \mathtt{OPT}_g, \mathtt{OPT}_g]$, where $\varepsilon_f$ and $\varepsilon_g$ are the relative errors for approximating $\mathtt{OPT}_f$ and $\mathtt{OPT}_g$, respectively.
The following lemma indicates that the BSM decision problem will still be answered with a theoretical guarantee by solving the approximate version of the problem in Eq.~\ref{eq:sm}.
\begin{lemma}\label{lm:dec}
  Let $F'_{\alpha}(S) := \min \big\{1, \frac{f(S)}{\alpha \cdot \mathtt{OPT}'_f} \big\} + \frac{1}{c} \sum_{i \in [c]} \min \big\{1, $ $\frac{f_i(S)}{\tau \cdot \mathtt{OPT}'_g} \big\}$. On the one hand, any set $\widehat{S}$ with $F'_{\alpha}(\widehat{S}) \geq 2(1 - \frac{\varepsilon}{c})$ is an $(\hat{\alpha}, \hat{\beta})$-approximate solution to BSM, where $\hat{\alpha} =(1 - 2\varepsilon - \varepsilon_f)\alpha$ and $\hat{\beta} = 1 - 2\varepsilon - \varepsilon_g$. On the other hand, there is not any $\alpha$-approximate solution to BSM if $F'_{\alpha}(S) < 2$ for any size-$k$ set $S$.
\end{lemma}

\begin{proof}
  On the one hand, if $F'_{\alpha}(\widehat{S}) \geq 2(1 - \frac{\varepsilon}{c})$, then it will hold that $\frac{f(\widehat{S})}{\alpha \cdot \mathtt{OPT}'_f} \geq 1 - \frac{2\varepsilon}{c}$ and $\frac{1}{c} \sum_{i \in [c]} \frac{f_i(\widehat{S})}{\tau \cdot \mathtt{OPT}'_g} \geq 1 - \frac{2\varepsilon}{c}$. According to the prior inequality, we have
  \begin{align*}
    f(\widehat{S}) & \geq \alpha (1 - \frac{2\varepsilon}{c}) \cdot \mathtt{OPT}'_f \\
                   & \geq \alpha (1 - 2\varepsilon)(1 - \varepsilon_f) \cdot \mathtt{OPT}_f \\
                   & > (1 - 2\varepsilon - \varepsilon_f) \alpha \cdot \mathtt{OPT}_f
  \end{align*}
  According to the latter inequality, for each $i \in [c]$,
  \begin{align*}
    f_i(\widehat{S}) & \geq \tau \big((1 - \frac{2\varepsilon}{c}) c - (c - 1) \big) \cdot \mathtt{OPT}'_g \\
                     & > \tau (1 - 2\varepsilon)(1 - \varepsilon_g) \cdot \mathtt{OPT}_g \\
                     & > (1 - 2\varepsilon - \varepsilon_g) \tau \cdot \mathtt{OPT}_g
  \end{align*}
  Thus, $g(\widehat{S}) = \min_{i \in [c]} f_i(\widehat{S}) > (1 - 2\varepsilon - \varepsilon_g) \tau \cdot \mathtt{OPT}_g$.
  We prove that $\widehat{S}$ is an $(\hat{\alpha}, \hat{\beta})$-approximate solution to BSM, where $\hat{\alpha} =(1 - 2\varepsilon - \varepsilon_f)\alpha$ and $\hat{\beta} = 1 - 2\varepsilon - \varepsilon_g$.
  On the other hand, if $F'_{\alpha}(S) < 2$ for any size-$k$ set $S \subseteq V$, then either $f(S) < \alpha \cdot \mathtt{OPT}'_f$ or there exists some $i \in [c]$ such that $f_i(S) < \tau \cdot \mathtt{OPT}'_g$ and thus $g(S) < \tau \cdot \mathtt{OPT}'_g$. In either case, it must hold that $f(S) < \alpha \cdot \mathtt{OPT}'_f < \alpha \cdot \mathtt{OPT}_f$ or $g(S) < \tau \cdot \mathtt{OPT}'_g < \tau \cdot \mathtt{OPT}_g$. Therefore, $S$ must not be an $\alpha$-approximate solution to BSM.
\end{proof}

\begin{algorithm}[tb]
  \small
  \caption{\textsc{BSM-Saturate}}
  \label{alg1}
  \begin{algorithmic}[1]
    \REQUIRE Two functions $f, g : 2^V \rightarrow \mathbb{R}_{\geq 0}$, balance factor $\tau \in (0, 1)$, solution size $k \in \mathbb{Z}^{+}$, error parameter $\varepsilon \in (0, 1)$
    \ENSURE Solution $\widehat{S}$ with $\lvert \widehat{S} \rvert \leq k \ln{\frac{c}{\varepsilon}}$ for BSM
    \STATE Run \textsc{Greedy} \cite{NemhauserWF78} on $f$ to compute $\mathtt{OPT}'_f$
    \STATE Run \textsc{Saturate} \cite{Krause08} on $g$ to compute $\mathtt{OPT}'_g$
    \STATE $\alpha_{max} \gets 1$ and $\alpha_{min} \gets 0$
    \WHILE{$(1-\varepsilon)\alpha_{max} > \alpha_{min}$\label{ln-stop}}
      \STATE Set $\alpha \gets (\alpha_{max} + \alpha_{min}) / 2$
      \STATE Define $F'_{\alpha} := \min \big\{1, \frac{f(S)}{\alpha \cdot \mathtt{OPT}'_f} \big\} + \frac{1}{c} \sum_{i \in [c]} \min \big\{1, $ $\frac{f_i(S)}{\tau \cdot \mathtt{OPT}'_g} \big\}$
      \STATE Initialize $S \gets \emptyset$
      \FOR{$i \gets 1, \ldots, k \ln{\frac{c}{\varepsilon}}$\label{ln-k-val}}
        \STATE $v^* \gets \argmax_{v \in V} F'_{\alpha}(S \cup \{v\}) - F'_{\alpha}(S)$
        \STATE $S \gets S \cup \{v^*\}$
      \ENDFOR
      \IF{$F'_{\alpha}(S) \geq 2(1 - \frac{\varepsilon}{c})$}
        \STATE $\alpha_{min} \gets \alpha$ and $\widehat{S} \gets S$
      \ELSE
        \STATE $\alpha_{max} \gets \alpha$
      \ENDIF
    \ENDWHILE
    \RETURN{$\widehat{S}$}
  \end{algorithmic}
\end{algorithm}

Based on Lemma~\ref{lm:dec}, we complete the conversion by connecting a BSM decision problem with a submodular cover problem on $F'_{\alpha}$.
Accordingly, we propose \textsc{BSM-Saturate} in Algorithm~\ref{alg1} by solving submodular cover instances on $F'_{\alpha}$ with different $\alpha$'s to find an appropriate $\alpha$ value and thus provide a good BSM solution.
First, it also utilizes \textsc{Greedy} \cite{NemhauserWF78} for SM and \textsc{Saturate} \cite{Krause08} for RSM to compute the approximate values $\mathtt{OPT}'_f$ and $\mathtt{OPT}'_g$ for $\mathtt{OPT}_f$ and $\mathtt{OPT}_g$, respectively.
Next, it performs a bisection search on $\alpha$ within $[0, 1]$.
For each value of $\alpha$, it runs \textsc{Greedy} \cite{NemhauserWF78} to maximize the submodular function $F'_{\alpha}$ defined in Lemma~\ref{lm:dec} with size constraint $k \ln{\frac{c}{\varepsilon}}$.
If the function value $F'_{\alpha}(S)$ of the greedy solution $S$ reaches $2(1 - \frac{\varepsilon}{c})$, it will set $S$ as the current solution $\widehat{S}$ and search on the upper half to find a better solution.
Otherwise, the search will be performed on the lower half to find a feasible solution.
Finally, it terminates the bisection search when the ratio between the upper and lower bounds of $\alpha$ is within $1-\varepsilon$ and returns the solution $\widehat{S}$ obtained for the lower bound of $\alpha$ at the last iteration for BSM.

Subsequently, we analyze the approximation factor and time complexity of \textsc{BSM-Saturate} in the following theorem.
\begin{theorem}\label{thm:alg1}
  Algorithm~\ref{alg1} returns a $\big( (1 - 3\varepsilon - \varepsilon_f) \alpha^*, 1 - 2\varepsilon - \varepsilon_g \big)$-approximate solution $\widehat{S}$ with $\lvert \widehat{S} \rvert \leq k \ln{\frac{c}{\varepsilon}}$ for an instance of BSM in $\mathcal{O}\big(n m k\log^2(\frac{cm}{\alpha^*\varepsilon})\big)$ time, where $\alpha^*$ is the best possible approximation factor of the BSM instance, $\varepsilon_f \leq \frac{1}{e}$, $\varepsilon_g \leq 1 - \frac{1}{m} - \frac{\widehat{\mathtt{OPT}_g}}{\mathtt{OPT}_g}$, and $\widehat{\mathtt{OPT}_g}$ is the optimum of maximizing $g$ with a cardinality constraint $\mathcal{O}\big(k\log^{-1}(cm)\big)$.
\end{theorem}
\begin{proof}
  First, since the approximation factor of \textsc{Greedy} \cite{NemhauserWF78} for SM is $1 - 1/e$, we have $\mathtt{OPT}'_f \in [(1 - 1/e) \cdot \mathtt{OPT}_f, \mathtt{OPT}_f]$ and thus $\varepsilon_f \leq \frac{1}{e}$.
  Second, according to the analysis of \textsc{Saturate} in \cite{NguyenT21}, we have $\mathtt{OPT}'_g \in [(1 - \frac{1}{m}) \cdot \widehat{\mathtt{OPT}_g}, \mathtt{OPT}_g]$ where $\widehat{\mathtt{OPT}_g}$ is the optimum of maximizing $g$ with a cardinality constraint $\mathcal{O}\big(k\log^{-1}(cm)\big)$. And thus, $\varepsilon_g \leq 1 - \frac{1}{m} - \frac{\widehat{\mathtt{OPT}_g}}{\mathtt{OPT}_g}$.
  For the lower bound $\alpha_{min}$ of $\alpha$ when the bisection search in Algorithm~\ref{alg1} is terminated, as it holds that $F'_{\alpha_{min}}(\widehat{S}) \geq 2(1 - \frac{\varepsilon}{c})$, we have $\widehat{S}$ is an $\big((1 - 2\varepsilon - \varepsilon_f)\alpha_{min}, 1 - 2\varepsilon - \varepsilon_g \big)$-approximate solution to BSM according to Lemma~\ref{lm:dec}.
  Furthermore, for the upper bound $\alpha_{max}$ of $\alpha$ when the bisection search in Algorithm~\ref{alg1} is terminated, it holds that $F'_{\alpha_{max}}(S_{\alpha_{max}}) < 2(1 - \frac{\varepsilon}{c}) < 2$, where $S_{\alpha_{max}}$ is the greedy solution of size $k\ln{\frac{c}{\varepsilon}}$ for maximizing $F'_{\alpha_{max}}$.
  By generalizing the analysis of \textsc{Greedy} in \cite{NemhauserWF78}, for any monotone submodular function $F$, the greedy solution $S_l$ after $l$ iterations satisfies that $F(S_l) \geq \big(1 - \exp(-\frac{l}{k})\big) \cdot \mathtt{OPT}_F$, where $\mathtt{OPT}_F$ is the optimum of maximizing $F$ with a size constraint $k$. Taking the above inequality into $F'_{\alpha_{max}}$ and $S_{\alpha_{max}}$, we have
  \begin{equation*}
    F'_{\alpha_{max}}(S_{\alpha_{max}}) \geq \Big( 1 - \exp(-\frac{k \ln{\frac{c}{\varepsilon}}}{k}) \Big) \cdot \mathtt{OPT}_{F'} \geq (1 - \frac{\varepsilon}{c}) \cdot \mathtt{OPT}_{F'}.
  \end{equation*}
  Therefore, we have $\mathtt{OPT}_{F'} < 2$, i.e., $F'_{\alpha_{max}}(S) < 2$ for any size-$k$ set $S \subseteq V$.
  Then, we can safely say there is no $\alpha_{max}$-approximate solution to the BSM instance and $\alpha^* \leq \alpha_{max}$.
  Considering all the above results, we prove that $\widehat{S}$ returned by Algorithm~\ref{alg1} is an $\big( (1 - 3\varepsilon - \varepsilon_f) \alpha^*, 1 - 2\varepsilon - \varepsilon_g \big)$-approximate solution of size at most $k \ln{\frac{c}{\varepsilon}}$ to any BSM instance.

  In terms of complexity, \textsc{Greedy} and \textsc{Saturate} take $\mathcal{O}(n m k)$ and $\mathcal{O}\big( n m k \log(cm) \big)$ time, respectively.
  Moreover, the bisection search attempts $\mathcal{O}\big( \log(\frac{1}{\alpha^* \varepsilon}) \big)$ different $\alpha$'s before termination. For each value of $\alpha$, it takes $\mathcal{O}(n m k\log{\frac{c}{\varepsilon}})$ time to compute a solution $S$. Therefore, the time complexity of Algorithm~\ref{alg1} is $\mathcal{O}\big(n m k \log^2(\frac{cm}{\alpha^*\varepsilon})\big)$.
\end{proof}

Theorem~\ref{thm:alg1} indicates that \textsc{BSM-Saturate} improves upon \textsc{BSM-TSGreedy} with better theoretical guarantees:
its approximation factor is close to the best achievable $\alpha^*$ when the error terms are ignored.
However, such an improvement comes at the expense of providing solutions of sizes greater than $k$.
In practice, to ensure that the solution size is at most $k$, we substitute $k \ln{\frac{c}{\varepsilon}}$ in Line~\ref{ln-k-val} of Algorithm~\ref{alg1} with $k$, while keeping the remaining steps unchanged.

\begin{example}
  We consider running \textsc{BSM-Saturate} on the BSM instance with $k = 2$ in Figure~\ref{fig:example}. Here, we set $\varepsilon = 0.1$ and replace $k \ln{\frac{c}{\varepsilon}} = 2 \ln{20}$ in Line~\ref{ln-k-val} with $k = 2$ to ensure that the size of $\widehat{S}$ is $2$. As shown in Example~\ref{example:BSM}, it first runs \textsc{Greedy} and \textsc{Saturate} to compute $S_f = \{v_1, v_2\}$ with $\mathtt{OPT}'_f = 0.75$ and $S_g = \{v_1, v_4\}$ with $\mathtt{OPT}'_g \approx 0.556$. For a given $\tau \in [0, 1]$, it maximizes $F'_{\alpha}$ in Lemma~\ref{lm:dec} for different $\alpha$'s using the greedy algorithm until $\frac{\alpha_{min}}{\alpha_{max}} > 0.9$ and returns the solution w.r.t.~$\alpha_{min}$ for BSM. When $\tau = 0.2$ and $0.5$, it attempts to maximize $F'_{\alpha}$ for $\alpha = 0.5$, $0.75$, $0.875$, and $0.9375$ one by one. For each $\alpha$ value, it adds $v_3$ and $v_1$ into $S$ and gets $F'_{\alpha}(\{v_1, v_3\}) > 1.9$. Thus, it terminates with $\alpha_{min} = 0.9375$, $\alpha_{max} = 1$, and $F'_{0.9375}(\{v_1, v_3\}) \approx 0.95 + 1 > 1.9$, and returns $\widehat{S} = \{ v_1, v_3 \}$ for BSM. When $\tau = 0.8$, it also attempts to maximize $F'_{\alpha}$ for $\alpha = 0.5$, $0.75$, and $0.875$. However, it obtains $F'_{0.875}(\{v_1, v_3\}) = 1.875 < 1.9$ for $\alpha = 0.875$. Next, it sets $\alpha = 0.8125$ and gets $F'_{0.8125}(\{v_1, v_4\}) \approx 0.96 + 1 > 1.9$. Thus, it finishes with $\alpha_{min} = 0.8125$, $\alpha_{max} = 0.875$, and returns $\widehat{S} = \{ v_1, v_4 \}$ for BSM.
\end{example}

\section{Experimental Evaluation}
\label{sec:exp}

In this section, we evaluate our optimization framework, i.e., BSM, and algorithms, i.e., \textsc{BSM-TSGreedy} and \textsc{BSM-Saturate}, by extensive experiments on three problems, namely \emph{maximum coverage}, \emph{influence maximization}, and \emph{facility location}, using real-world and synthetic datasets.
The goal of the experiments is to answer the following questions:
\begin{description}
  \item[Q1:] How does the factor $\tau$ affect the balance between the values of two objective functions $f$ and $g$?
  \item[Q2:] How far are the solutions produced by the proposed approximation algorithms from optimal in practice?
  \item[Q3:] How effective and efficient are the proposed algorithms?
  \item[Q4:] Are the proposed algorithms scalable to large data?
\end{description}

\parax{Algorithms and Baselines}
The following algorithms and baselines are compared in the experiments.
\begin{itemize}
  \item \textsc{Greedy} \cite{NemhauserWF78}: The classic $(1-1/e)$-approximation greedy algorithm for SM, which is also used as a subroutine to maximize $f$ in our algorithms.
  \item \textsc{Saturate} \cite{Krause08}: The bicriteria approximation algorithm for RSM, which is also used as a subroutine to maximize $g$ in our algorithms.
  \item SMSC \cite{OhsakaM21}: The $(0.16, 0.16)$-approximation algorithm for submodular maximization under submodular cover, which can be used for BSM only when $c=2$ by maximizing two submodular functions $f_1$ and $f_2$ simultaneously.
  \item \textsc{BSM-Optimal}: The exponential-time exact algorithm for BSM. Since the \emph{maximum coverage} and \emph{facility location} problems can be formulated as integer linear programs (ILPs) \cite{ip-book}, we find the optimal solutions of small BSM instances using an ILP solver (see Appendix~\ref{app:ilp}). We aim to measure the gap between the optimal and approximate solutions to BSM.
  \item \textsc{BSM-TSGreedy}: Our first instance-dependent bicriteria approximation algorithm for BSM (Algorithm~\ref{alg2}).
  \item \textsc{BSM-Saturate}: Our second instance-dependent bicriteria approximation algorithm for BSM (Algorithm~\ref{alg1}), which is adapted to provide solutions of size at most $k$ for a fair comparison. In preliminary experiments, we observe that the value of $\varepsilon$ hardly affects the performance of \textsc{BSM-Saturate} unless $\varepsilon \geq 0.5$ (see Appendix~\ref{app:sec:exp}). To guarantee the good performance of \textsc{BSM-Saturate} in all cases, we set $\varepsilon = 0.05$ throughout the experiments in this section.
\end{itemize}

We implemented all the above algorithms in Python 3.
We used the \emph{lazy-forward} strategy \cite{LeskovecKGFVG07} to accelerate all algorithms except \textsc{BSM-Optimal} (for which the \emph{lazy-forward} strategy is not applicable).
The Gurobi optimizer\footnote{\url{https://www.gurobi.com/products/gurobi-optimizer/}} was applied to solve ILPs in \textsc{BSM-Optimal}.
All experiments were run on a server with an Intel Xeon E5-2650v4 2.2GHz processor and 96GB memory running Ubuntu 18.04 LTS.
Our code and data are published on \url{https://github.com/yhwang1990/code-bsm-release}.

\begin{table*}[t]
  \small
  \centering
  \caption{Statistics of datasets in the MC and IM experiments.}
  \label{tbl-stats}
  \begin{tabular}{llll}
  \toprule
  \textbf{Dataset} & $n$ (and $m$) & $\lvert E \rvert$ & \textbf{Percentage of users from each group} \\
  \midrule
  RAND ($c = 2$) & 500 or 100 & 8,946 or 360 & [`$U_0$': 20\%, `$U_1$': 80\%] \\
  RAND ($c = 4$) & 500 or 100 & 6,655 or 257 & [`$U_0$': 8\%, `$U_1$': 12\%, `$U_2$': 20\%, `$U_3$': 60\%] \\
  Facebook (Age, $c = 2$) & 1,216 & 42,443 & [`$< 20$': 8\%, `$\geq 20$': 92\%] \\
  Facebook (Age, $c = 4$) & 1,216 & 42,443 & [`$19$': 8\%, `$20$': 28\%, `$21$': 31\%, `$22$': 33\%] \\
  DBLP (Continent, $c = 5$) & 3,980 & 6,966  & [`Asia': 21\%, `Europe': 23\%, `North America': 52\%, `Oceania': 3\%, `South America': 1\%] \\
  Pokec (Gender, $c = 2$) & 1,632,803 & 30,622,564 & [`Female': 51\%, `Male': 49\%] \\
  Pokec (Age, $c = 6$) & 1,632,803 & 30,622,564 & [`0--20': 17\%, `21--30': 45\%, `31--40': 29\%, `41--50': 6\%, `51--60': 2\%, `60+': 1\%] \\
  \bottomrule
  \end{tabular}
\end{table*}

\begin{figure*}[t]
  \centering
  \includegraphics[width=.8\linewidth]{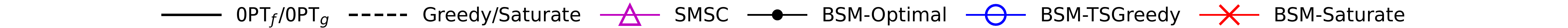}\\
  \smallskip
  \subcaptionbox{RAND ($c = 2, k = 5$)\label{fig:mc:tau:1}}[.33\linewidth]{\includegraphics[width=.48\linewidth]{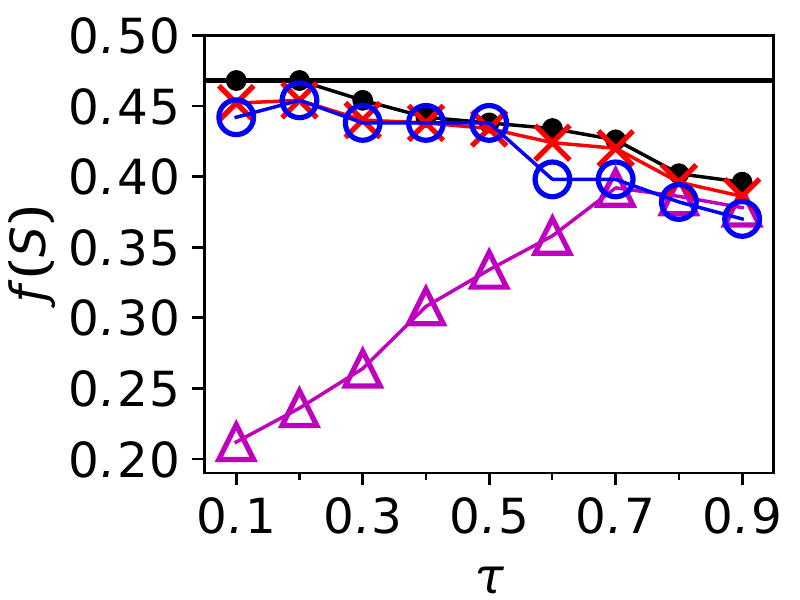} \includegraphics[width=.48\linewidth]{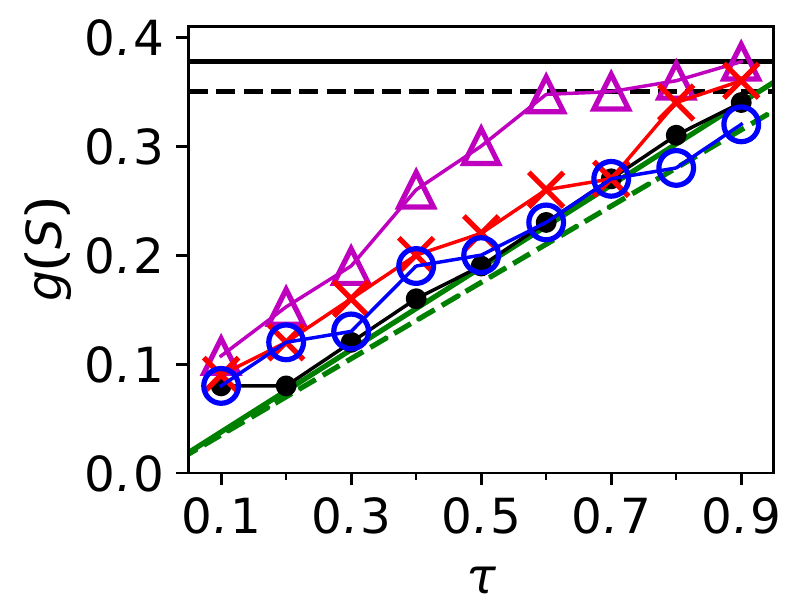}}
  \subcaptionbox{RAND ($c = 4, k = 5$)\label{fig:mc:tau:2}}[.33\linewidth]{\includegraphics[width=.48\linewidth]{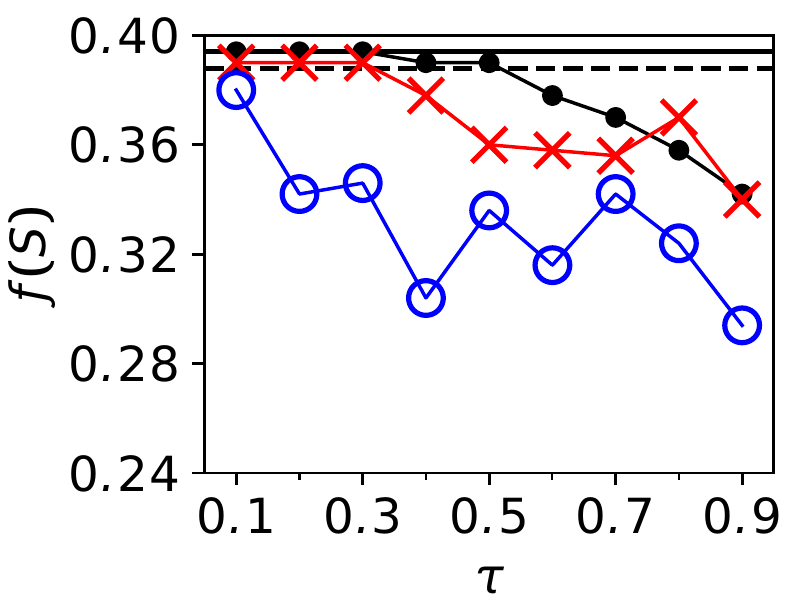} \includegraphics[width=.48\linewidth]{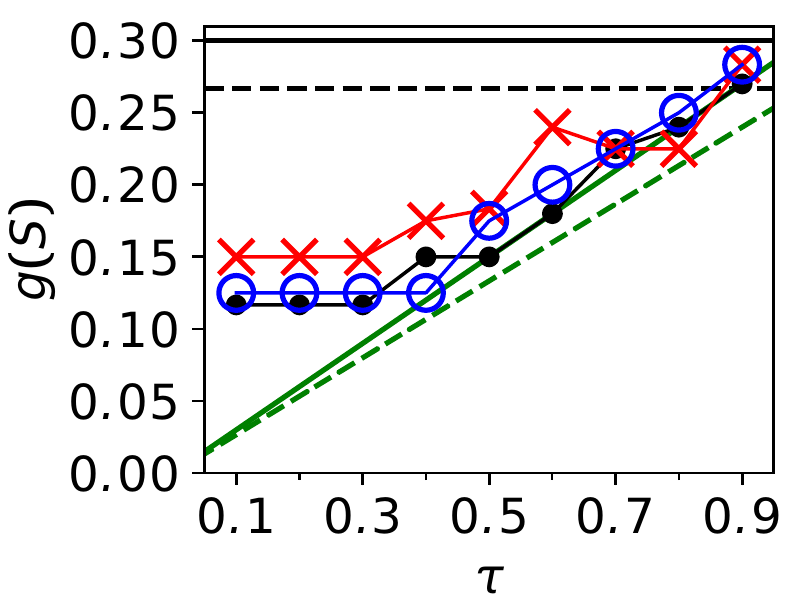}}
  \subcaptionbox{DBLP ($c = 5, k = 10$)\label{fig:mc:tau:3}}[.33\linewidth]{\includegraphics[width=.48\linewidth]{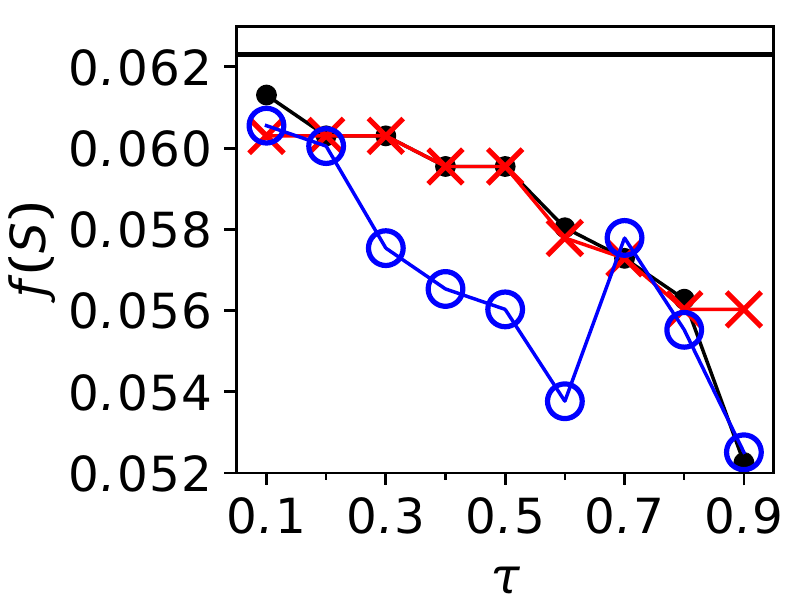} \includegraphics[width=.48\linewidth]{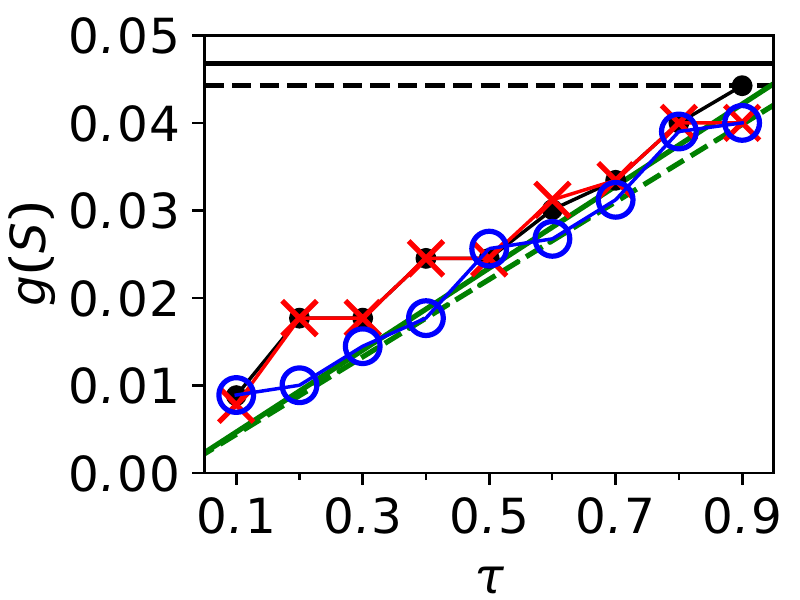}}
  \caption{Results of different algorithms for maximum coverage by varying the factor $\tau$ on two random graphs and DBLP. The green straight lines in the plots for $g(S)$ denote $\tau \cdot \mathtt{OPT}_g$ (solid) and $\tau \cdot \mathtt{OPT}'_g$ (dashed), where $\mathtt{OPT}_g$ and $\mathtt{OPT}'_g$ are computed by ILP and \textsc{Saturate}, to show whether the constraints $g(S) \geq \tau \cdot \mathtt{OPT}_g$ and $g(S) \geq \tau \cdot \mathtt{OPT}'_g$ are satisfied.}
  \label{fig:mc:tau}
\end{figure*}

\subsection{Maximum Coverage}

\parax{Setup}
We first apply the BSM framework to the maximum coverage (MC) problem to maximize the overall coverage while ensuring group-level fairness. Specifically, for a set $U$ of $m$ users and a collection $V$ of $n$ sets defined on $U$, we suppose that the utility $f_u(S)$ of a set $S \subseteq V$ to user $u$ is equal to $1$ if $u$ is contained by the union of the sets in $S$ and $0$ otherwise. Given the definition of $f_u$, the function $f(S)$ in Eq.~\ref{eq:average} captures the average coverage of $S$ over $U$ and the function $g(S)$ in Eq.~\ref{eq:maximin} denotes the minimum of average coverages of $S$ among different groups on $U$.

\begin{figure*}[t]
  \centering
  \includegraphics[width=.8\linewidth]{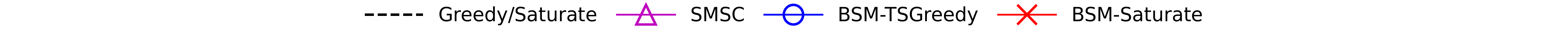}\\
  \smallskip
  \subcaptionbox{Facebook (Age, $c = 2, \tau = 0.8$)\label{fig:mc:k:1}}[.49\linewidth]{\includegraphics[width=.32\linewidth]{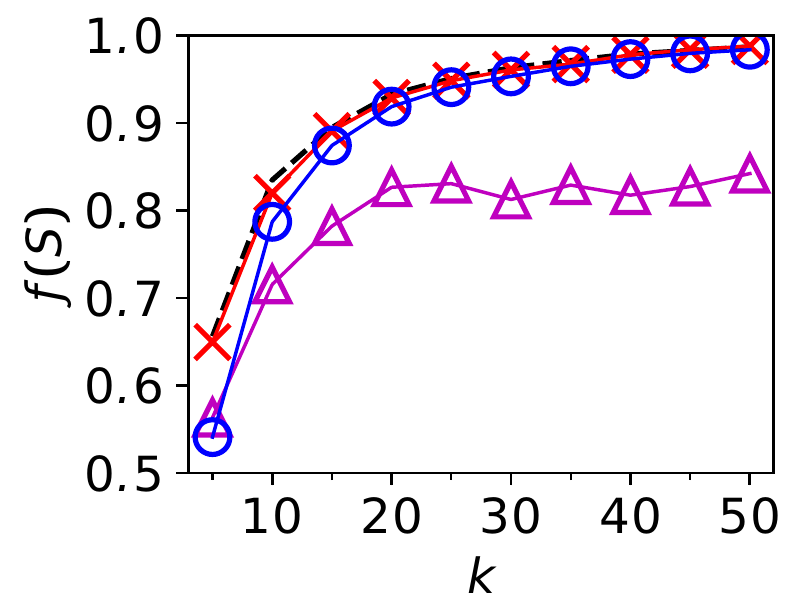} \includegraphics[width=.32\linewidth]{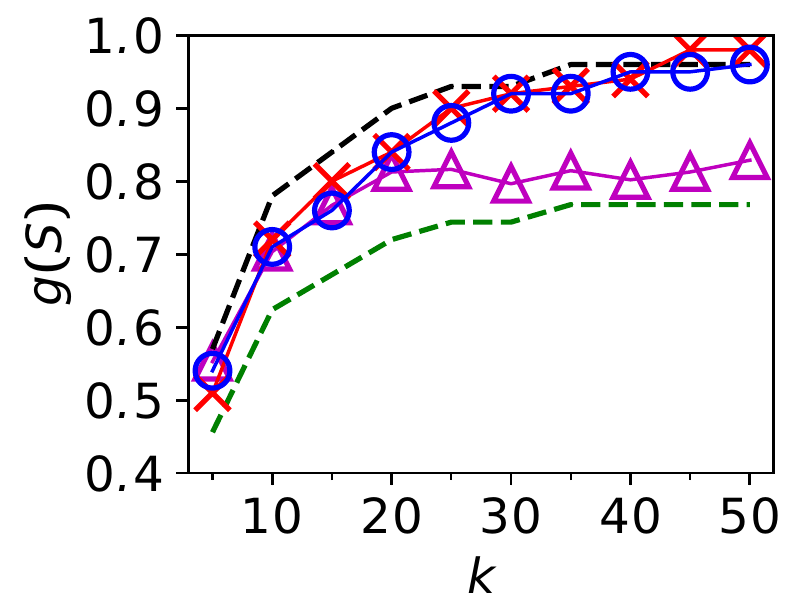} \includegraphics[width=.32\linewidth]{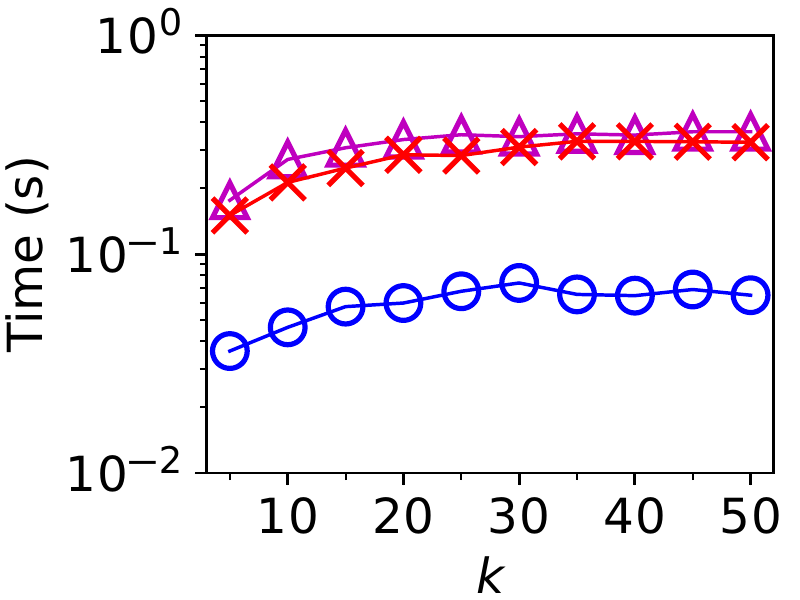}}
  \subcaptionbox{Facebook (Age, $c = 4, \tau = 0.8$)\label{fig:mc:k:2}}[.49\linewidth]{\includegraphics[width=.32\linewidth]{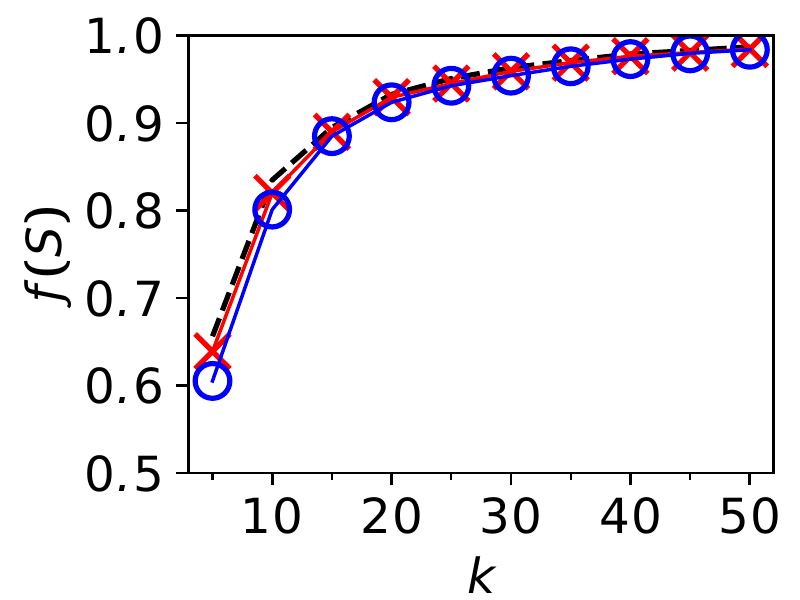} \includegraphics[width=.32\linewidth]{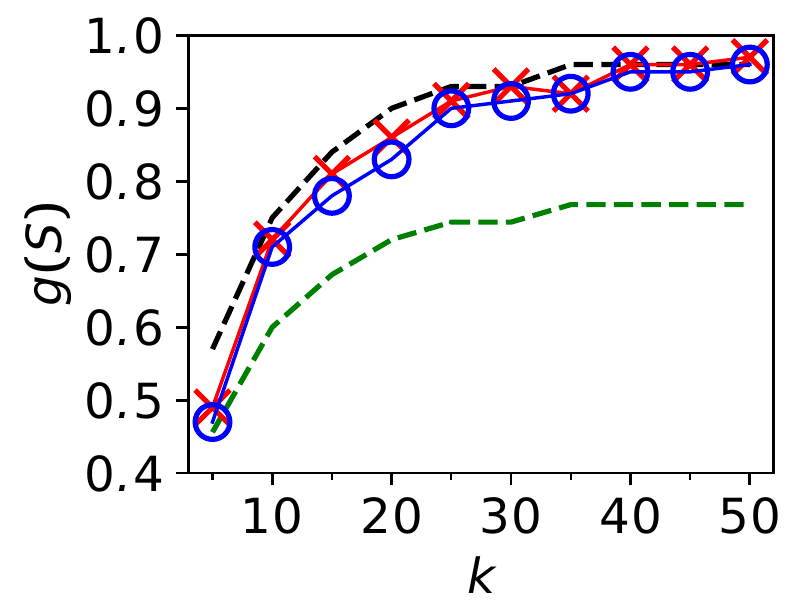} \includegraphics[width=.32\linewidth]{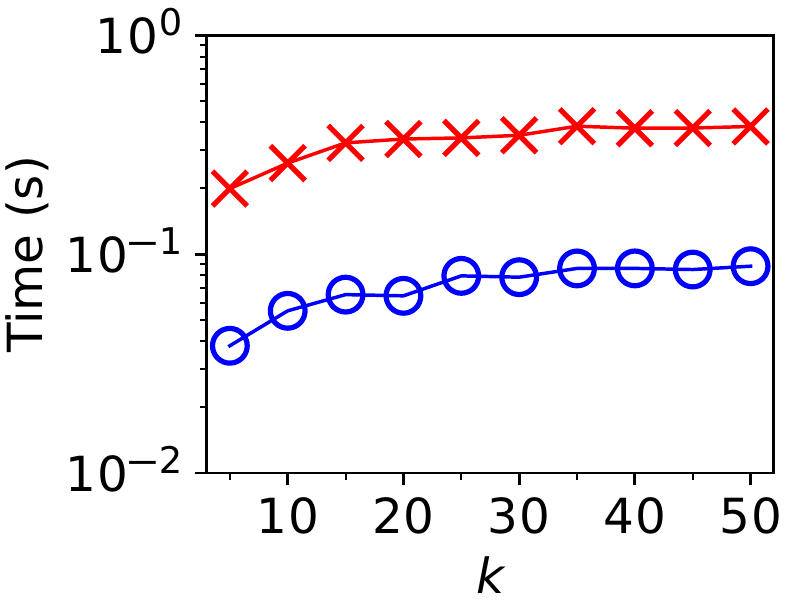}}\\
  \smallskip
  \subcaptionbox{Pokec (Gender, $c = 2, \tau = 0.8$)\label{fig:mc:k:3}}[.49\linewidth]{\includegraphics[width=.32\linewidth]{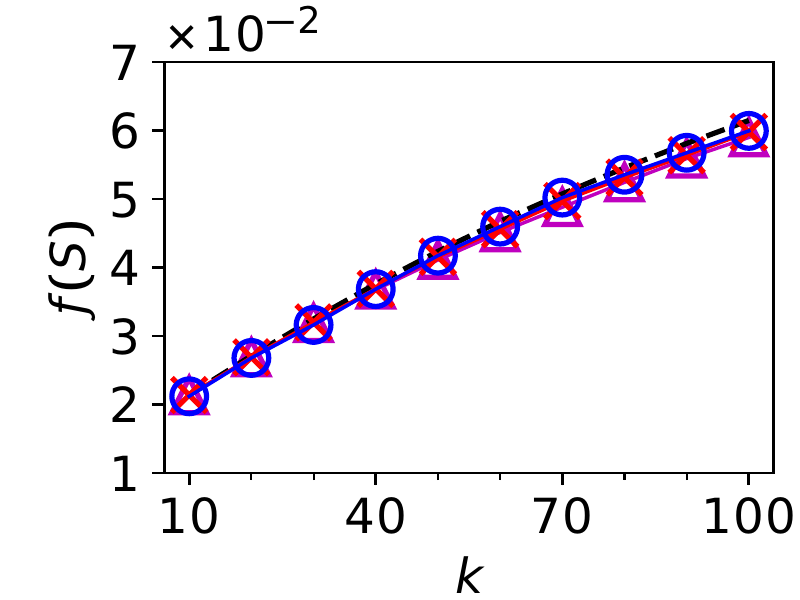} \includegraphics[width=.32\linewidth]{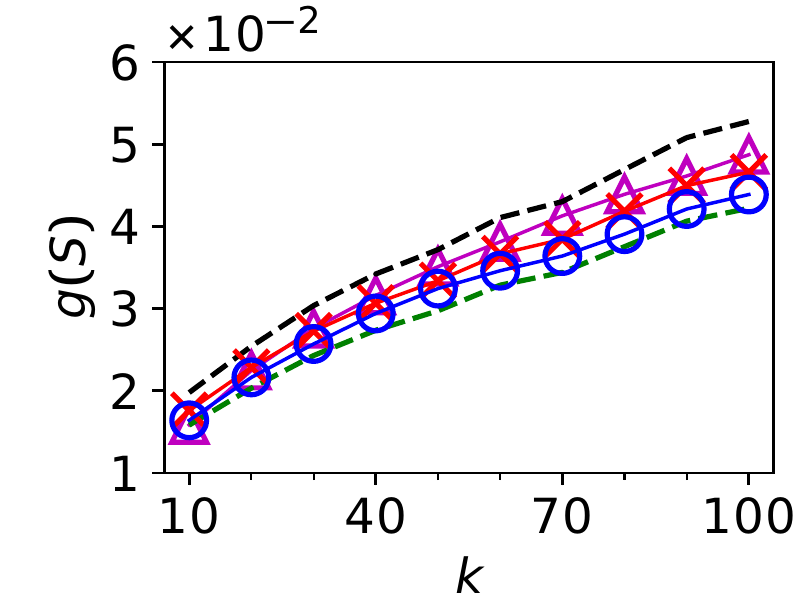} \includegraphics[width=.32\linewidth]{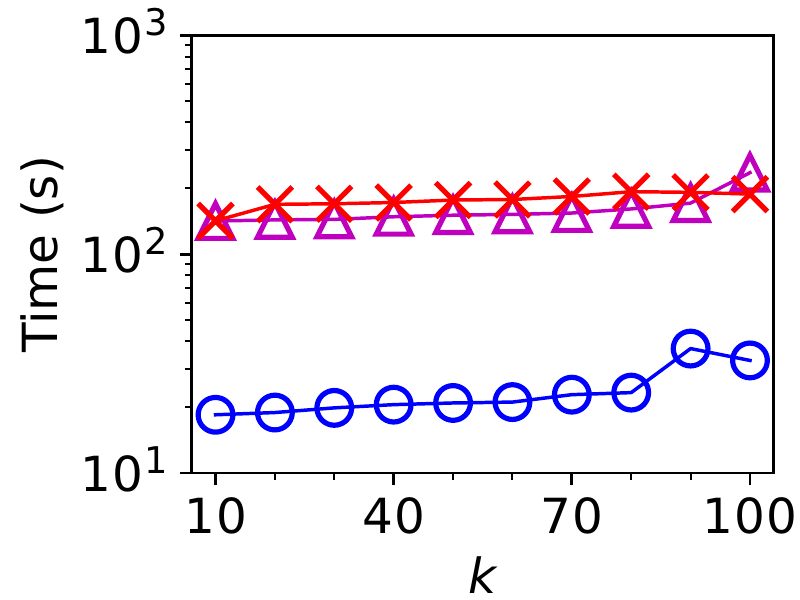}}
  \subcaptionbox{Pokec (Age, $c = 6, \tau = 0.8$)\label{fig:mc:k:4}}[.49\linewidth]{\includegraphics[width=.32\linewidth]{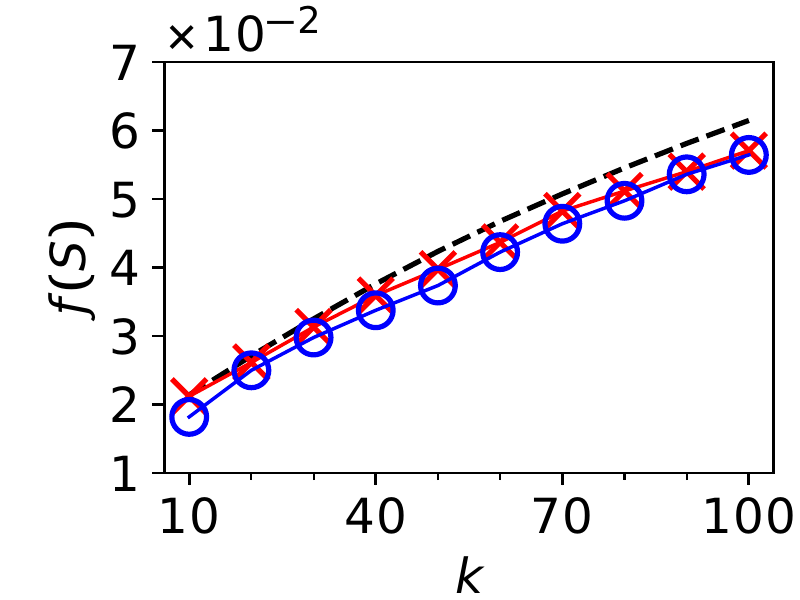} \includegraphics[width=.32\linewidth]{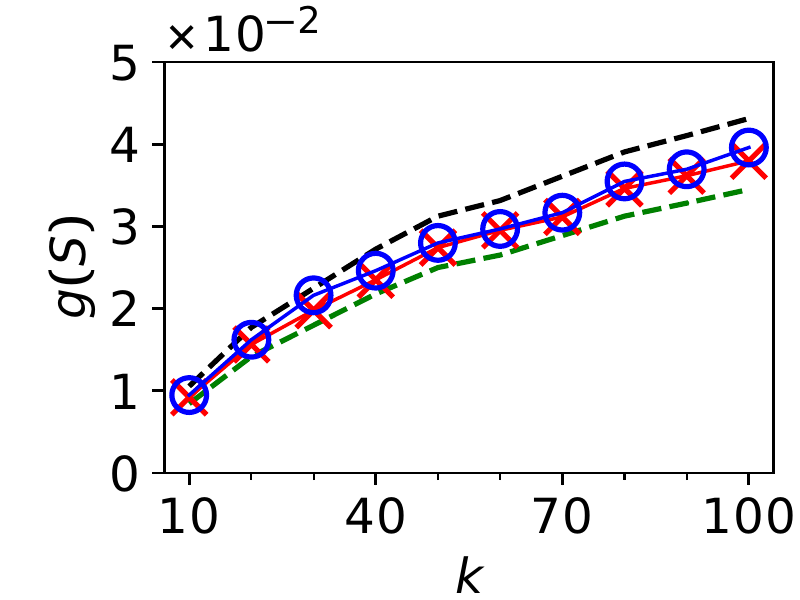} \includegraphics[width=.32\linewidth]{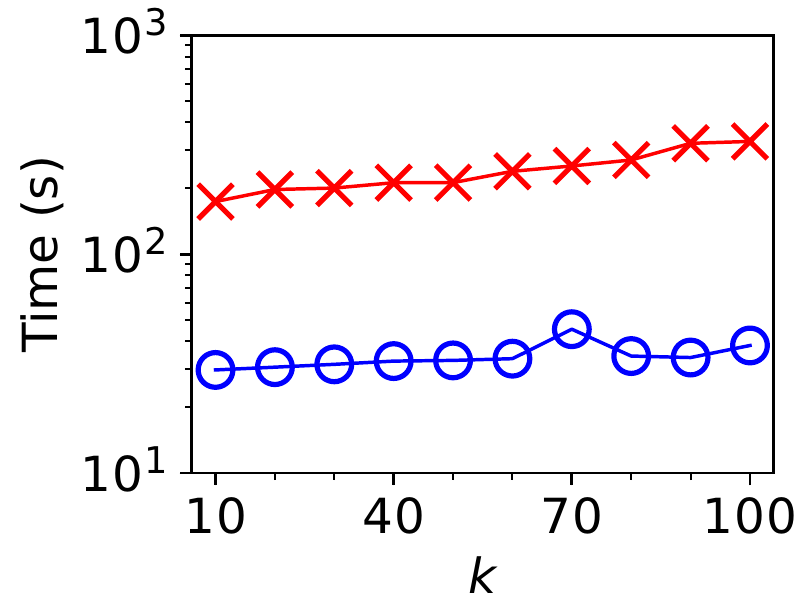}}
  \caption{Results of different algorithms for maximum coverage by varying solution size $k$ on Facebook and Pokec. The green dashed lines in the plots for $g(S)$ denote $\tau \cdot \mathtt{OPT}'_g$, where $\mathtt{OPT}'_g$ is computed by \textsc{Saturate}, to indicate whether $g(S) \geq \tau \cdot \mathtt{OPT}'_g$.}
  \label{fig:mc:k}
\end{figure*}

\begin{figure*}[t]
  \centering
  \includegraphics[width=.8\linewidth]{figs/legend-2.pdf}\\
  \smallskip
  \subcaptionbox{RAND ($c = 2, k = 5$)\label{fig:im:tau:1}}[.33\linewidth]{\includegraphics[width=.48\linewidth]{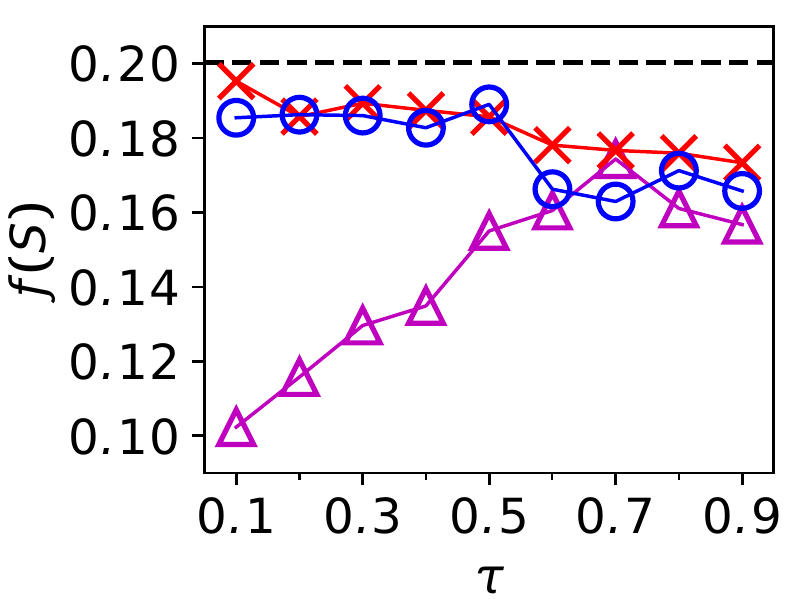} \includegraphics[width=.48\linewidth]{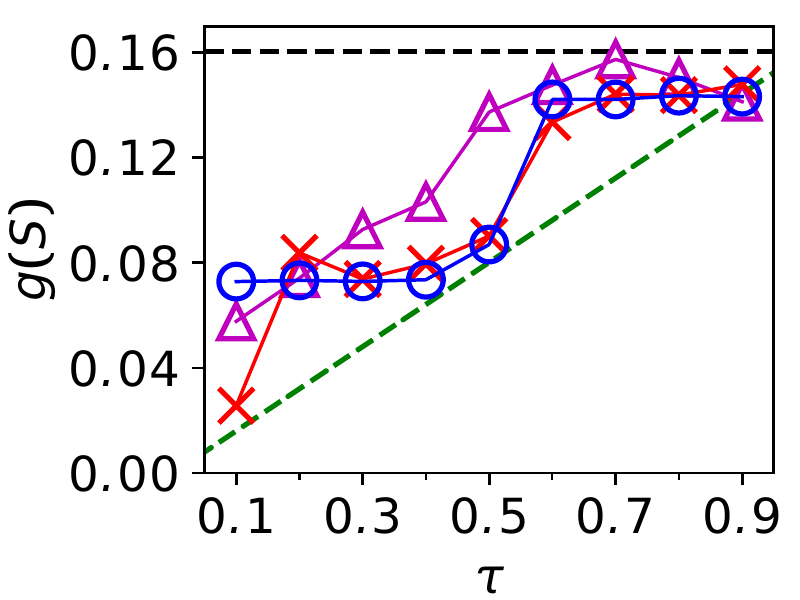}}
  \subcaptionbox{RAND ($c = 4, k = 5$)\label{fig:im:tau:2}}[.33\linewidth]{\includegraphics[width=.48\linewidth]{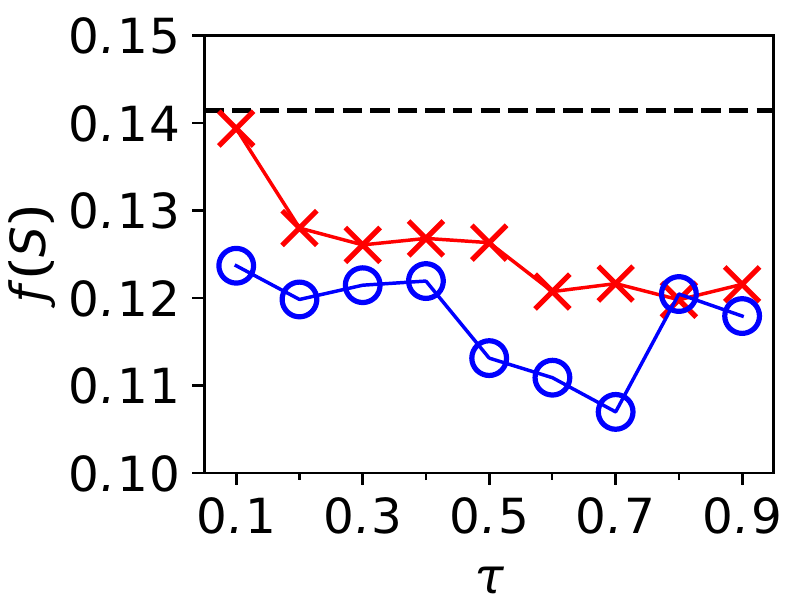} \includegraphics[width=.48\linewidth]{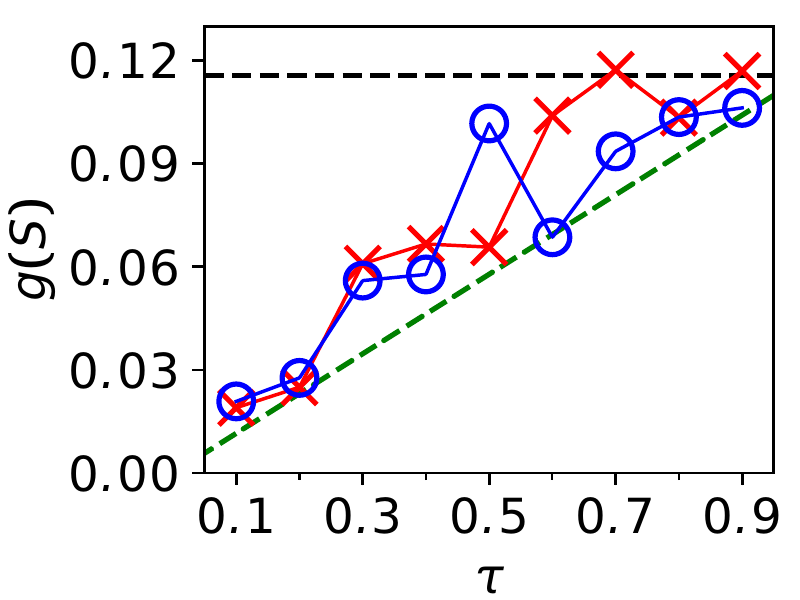}}
  \subcaptionbox{DBLP ($c = 5, k = 10$)\label{fig:im:tau:3}}[.33\linewidth]{\includegraphics[width=.48\linewidth]{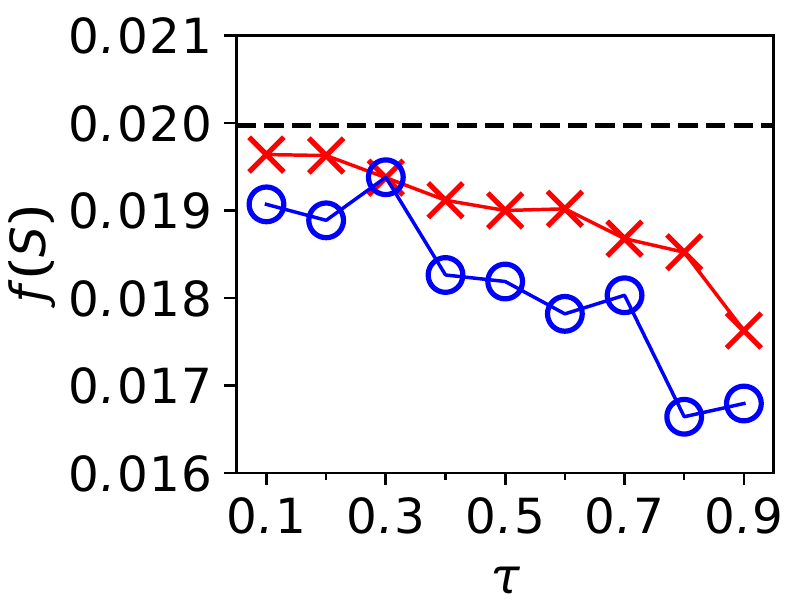} \includegraphics[width=.48\linewidth]{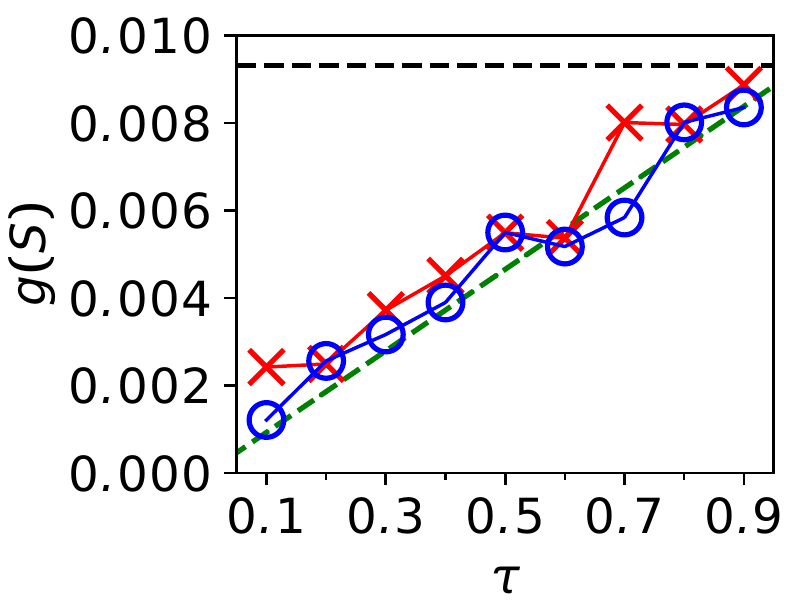}}
  \caption{Results of different algorithms for influence maximization by varying the factor $\tau$ on two random graphs and DBLP. The green, dashed straight lines denote $\tau \cdot \mathtt{OPT}'_g$, where $\mathtt{OPT}'_g$ is computed by \textsc{Saturate}, to show whether $g(S) \geq \tau \cdot \mathtt{OPT}'_g$.}
  \label{fig:im:tau}
\end{figure*}

We use two synthetic and three real-world datasets in the MC experiments.
The two synthetic datasets are undirected random graphs generated by the stochastic block model \cite{Holland83} (SBM) with 500 nodes consisting of two and four groups, respectively.
We set the intra-group and inter-group connection probabilities in the generation procedure to $0.1$ and $0.02$.
We also use three publicly available real-world datasets, namely \emph{Facebook} \cite{MisloveVGD10}, \emph{DBLP} \cite{10097603}, and \emph{Pokec}\footnote{\url{https://snap.stanford.edu/data/soc-Pokec.html}}, in the experiments.
The \emph{Facebook} dataset is an undirected graph representing the friendships between students at Rice University on Facebook.
Profile data contains the \emph{age} attribute to divide students into two (i.e., $\text{age}<20$ and $\text{age} \geq 20$) or four (i.e., $\text{age}= 19, 20, 21, 22$) groups.
The \emph{DBLP} dataset is an undirected graph denoting the co-authorships between researchers.
We divide them into five groups based on which continent (i.e., `Asia', `Europe', `North America', `Oceania', `South America') their affiliations are located in.
The \emph{Pokec} dataset is a directed graph representing the follower-followee relationships of users in a social network in Slovakia.
We use the \emph{gender} and \emph{age} information in profile data to divide users into two and six groups, respectively.
Table~\ref{tbl-stats} shows the statistics of all the above datasets, where $n$ (and $m$) is the number of nodes (and users) and $|E|$ is the number of edges.
We also present the percentage of users from each group in the population.
Following a common dominating set \cite{dominating-set} formulation, we construct a set system on each graph as follows. First, the ground set $U$ is equal to the set of nodes. Then, for each node $v$, we create a set $S(v)$ containing all its out-neighbors $N_{out}(v)$ plus itself. Finally, the set collection $V$ consists of the sets created for all the nodes in the graph. Based on the construction procedure, the task of BSM is to select a set of $k$ nodes that not only contain the largest number of users in their neighborhoods but also ensure at least a minimum average coverage for every group.

\smallskip\parax{Results on Effect of $\tau$}
Figure~\ref{fig:mc:tau} shows the values of $f(S), g(S)$, where $S$ is the solution returned by each algorithm, for each value of factor $\tau = \{0.1, 0.2, \ldots, 0.9\}$ on two random graphs when $k = 5$ and on \emph{DBLP} when $k = 10$.
The results for SMSC are ignored on the datasets with more than two groups because it does not provide any valid solution when $c>2$.
The optimums $\mathtt{OPT}_f$ and $\mathtt{OPT}_g$ for solely maximizing $f$ and $g$ as well as their approximations $\mathtt{OPT}'_f$ and $\mathtt{OPT}'_g$ returned by \textsc{Greedy} and \textsc{Saturate} are plotted as black horizontal lines to illustrate the trade-offs between \emph{utility} ($f$) and \emph{fairness} ($g$) of different algorithms w.r.t.~$\tau$.
In general, our BSM framework achieves a good utility-fairness trade-off. On the one hand, when the value of $\tau$ is close to $0$, $f(S)$ approaches or even reaches $\mathtt{OPT}_f$; on the other hand, when the value of $\tau$ increases, $f(S)$ decreases but $g(S)$ increases accordingly.
In contrast, the SMSC framework fails to balance two objectives well.
Compared with the optimal solutions returned by \textsc{BSM-Optimal}, the approximate solutions returned by \textsc{BSM-Saturate} and \textsc{BSM-TSGreedy} have up to $9\%$ and $26\%$ losses in $f(S)$.
\textsc{BSM-TSGreedy} provides lower-quality solutions than \textsc{BSM-Saturate} in terms of $f(S)$ for almost all $\tau$ values on all the three graphs.
This is mostly because the solution of \textsc{BSM-TSGreedy} has contained nearly $k$ items after the first stage, and the number $k'$ of items added in the second stage is thus close to $0$.
Nevertheless, \textsc{BSM-Saturate} achieves better trade-offs between $f(S)$ and $g(S)$ by combining them into a single objective function.
We find that \textsc{BSM-Saturate} and \textsc{BSM-TSGreedy} provide solutions higher in $f(S)$ than \textsc{BSM-Optimal} in very few cases.
Such results do not challenge the optimality of \textsc{BSM-Optimal} because \textsc{BSM-Saturate} and \textsc{BSM-TSGreedy} do not provide valid BSM solutions satisfying the constraint of $g(S) \geq \tau \cdot \mathtt{OPT}_g$ (i.e., above the solid green line) in these cases.
In fact, the solutions of both \textsc{BSM-Saturate} and \textsc{BSM-TSGreedy} violate the constraint $g(S) \geq \tau \cdot \mathtt{OPT}_g$ of BSM (i.e., below the solid green line) in some cases, because $\mathtt{OPT}_g$ is unknown to them, but always satisfy a ``weaker'' constraint $g(S) \geq \tau \cdot \mathtt{OPT}'_g$ (i.e., above the dashed green line), since $\mathtt{OPT}'_g$ is used as their input.

\begin{figure*}[t]
  \centering
  \includegraphics[width=.8\linewidth]{figs/legend-2.pdf}\\
  \smallskip
  \subcaptionbox{Facebook (Age, $c = 2, \tau = 0.8$)\label{fig:im:k:1}}[.49\linewidth]{\includegraphics[width=.32\linewidth]{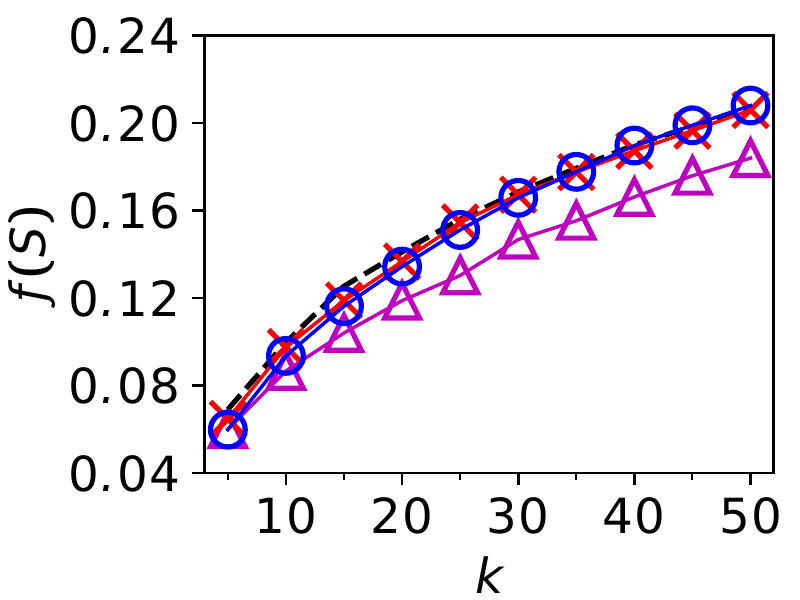} \includegraphics[width=.32\linewidth]{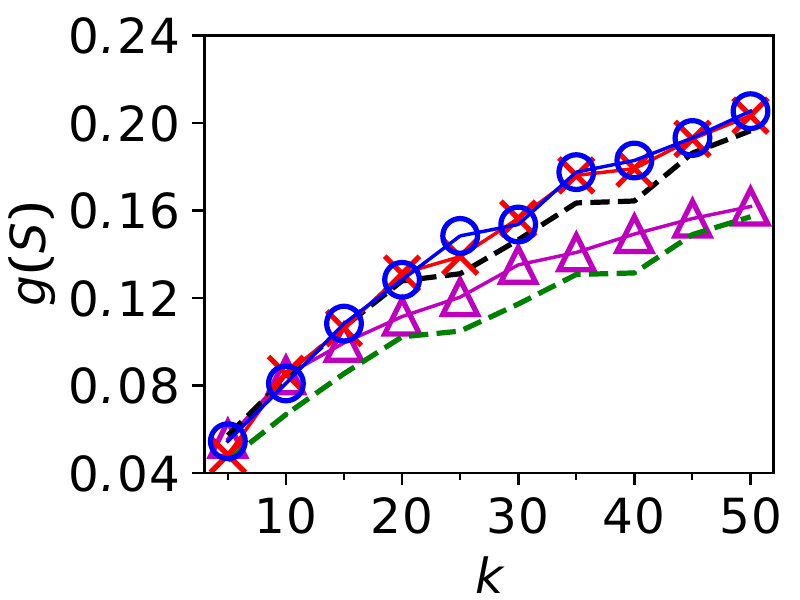} \includegraphics[width=.32\linewidth]{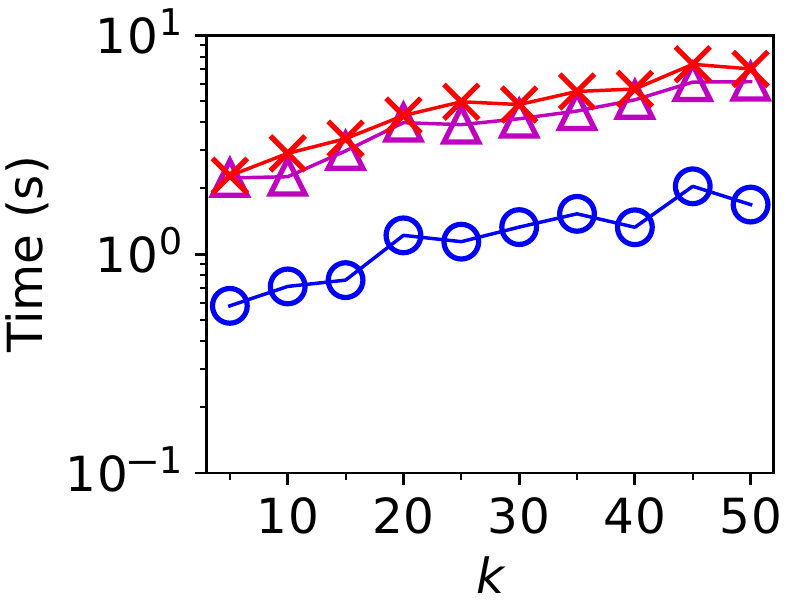}}
  \subcaptionbox{Facebook (Age, $c = 4, \tau = 0.8$)\label{fig:im:k:2}}[.49\linewidth]{\includegraphics[width=.32\linewidth]{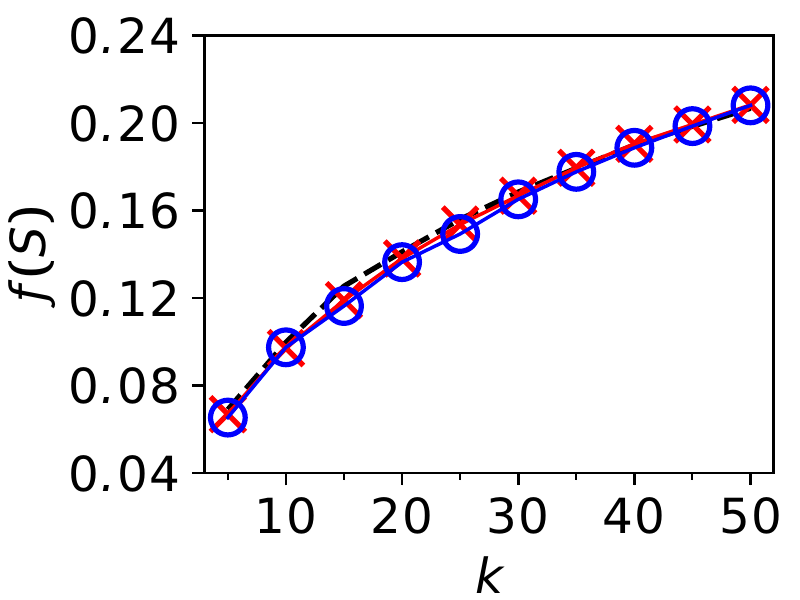} \includegraphics[width=.32\linewidth]{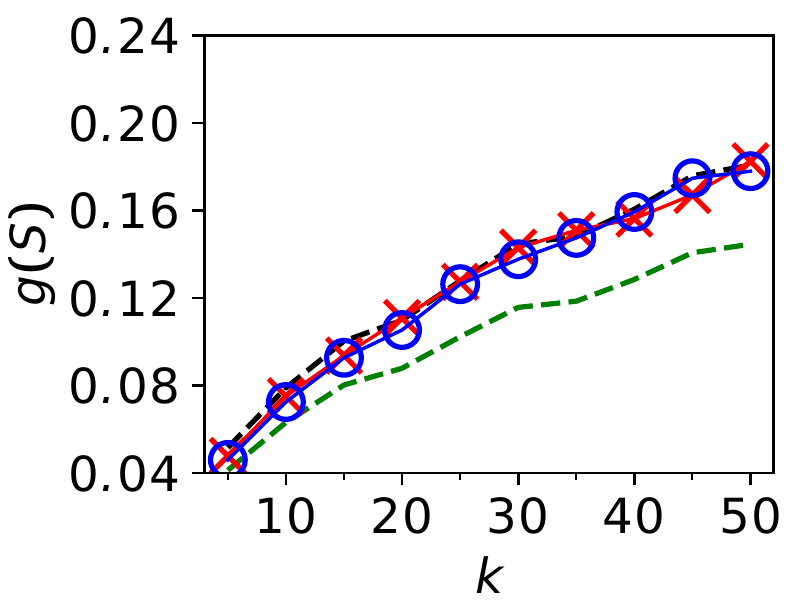} \includegraphics[width=.32\linewidth]{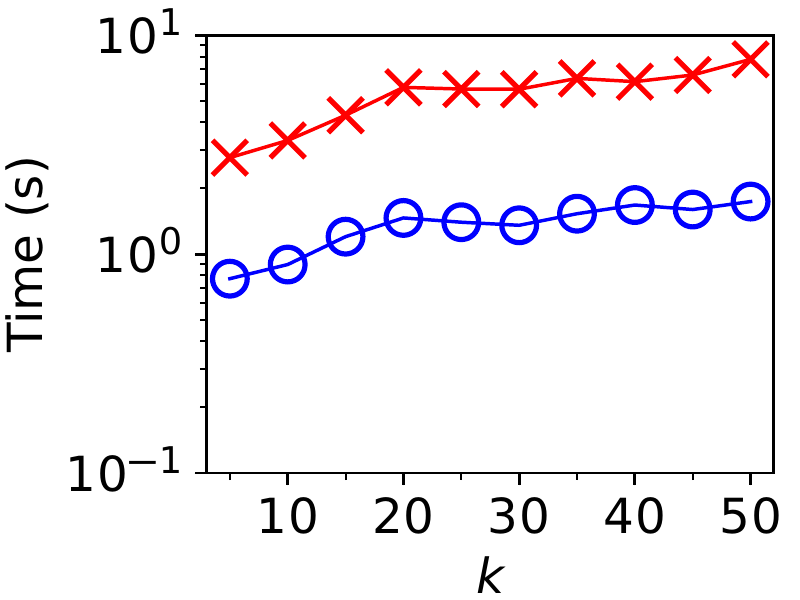}}
  \\ \smallskip
  \subcaptionbox{Pokec (Gender, $c = 2, \tau = 0.8$)\label{fig:im:k:3}}[.49\linewidth]{\includegraphics[width=.32\linewidth]{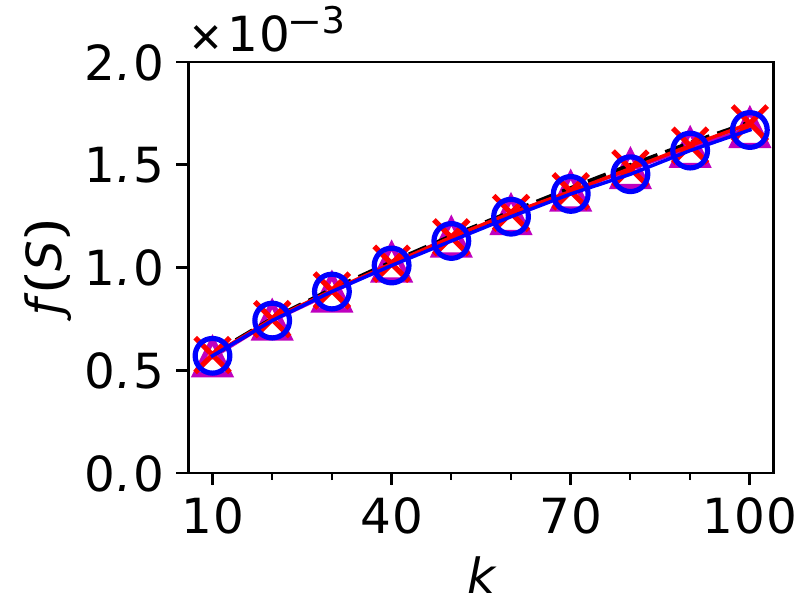} \includegraphics[width=.32\linewidth]{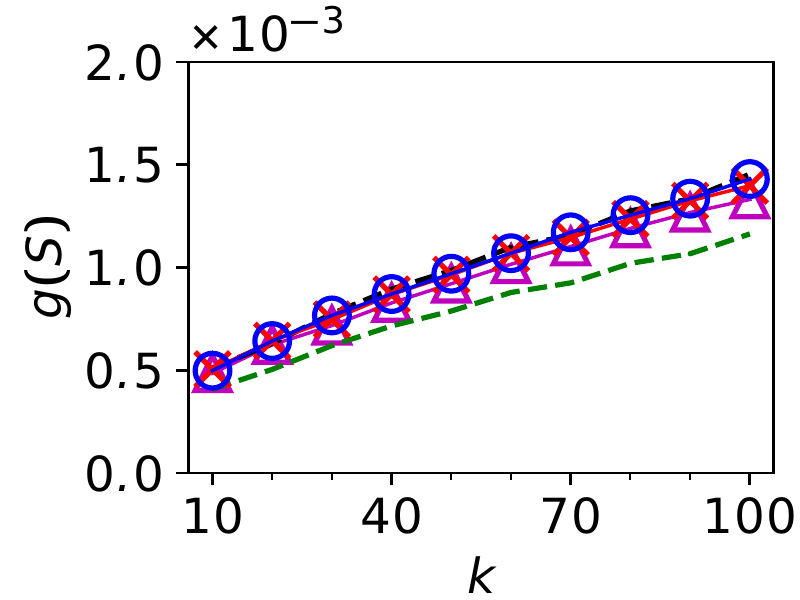} \includegraphics[width=.32\linewidth]{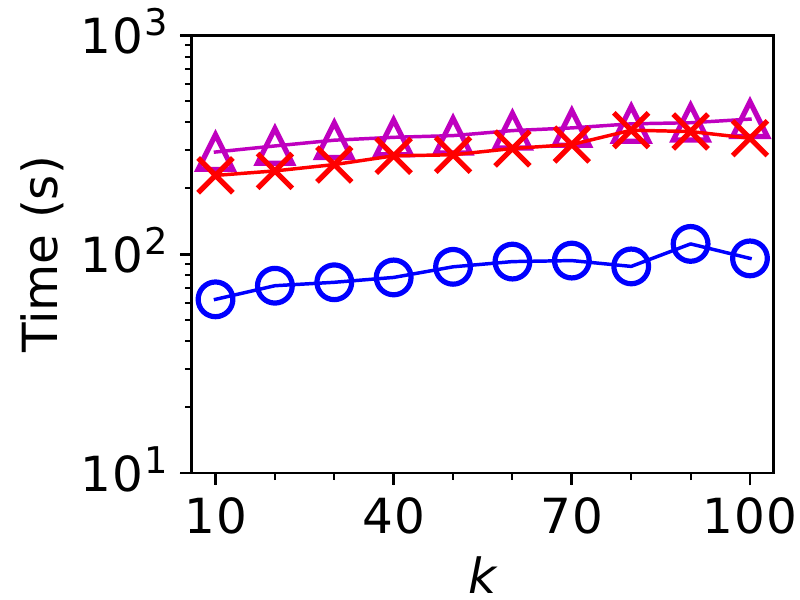}}
  \subcaptionbox{Pokec (Age, $c = 6, \tau = 0.8$)\label{fig:im:k:4}}[.49\linewidth]{\includegraphics[width=.32\linewidth]{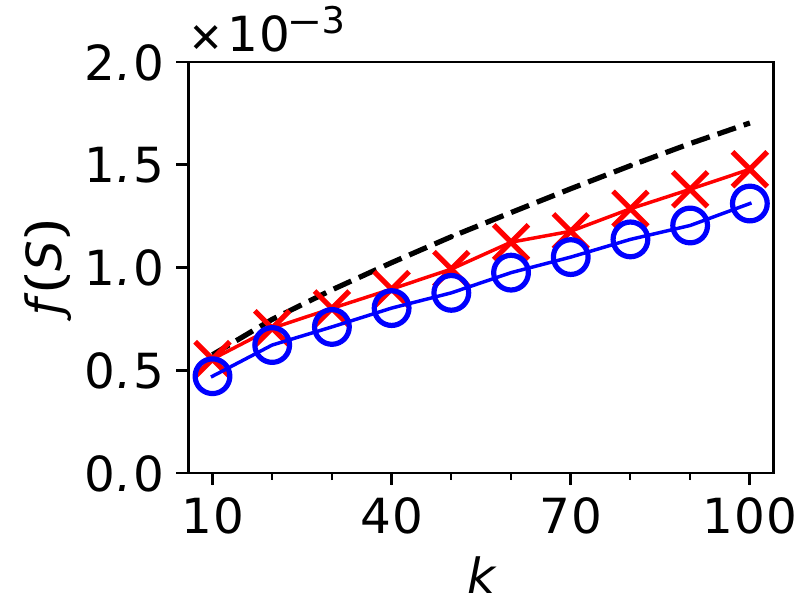} \includegraphics[width=.32\linewidth]{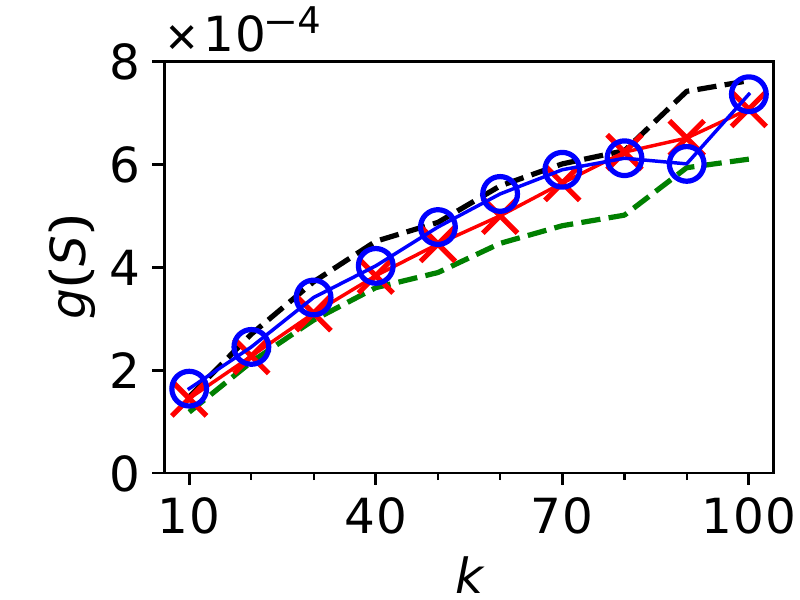} \includegraphics[width=.32\linewidth]{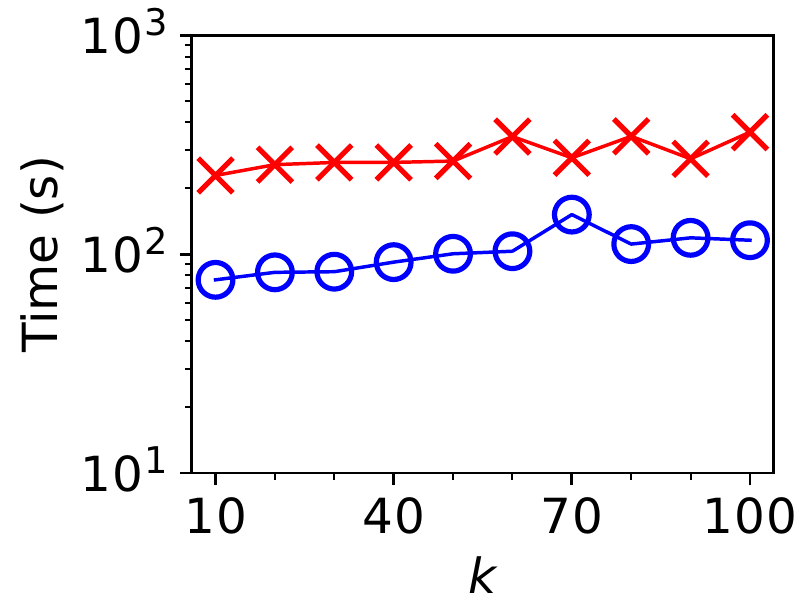}}
  \caption{Results of different algorithms for influence maximization by varying the solution size $k$ on Facebook and Pokec. The green dashed lines denote $\tau \cdot \mathtt{OPT}'_g$, where $\mathtt{OPT}'_g$ is computed by \textsc{Saturate}, to show whether $g(S) \geq \tau \cdot \mathtt{OPT}'_g$.}
  \label{fig:im:k}
\end{figure*}

\begin{table*}[t]
  \small
  \centering
  \caption{Statistics of datasets in the FL experiments.}
  \label{tbl-stats-fl}
  \begin{tabular}{lllll}
  \toprule
  \textbf{Dataset} & $n$ & $m$ & $d$ & \textbf{Percentage of users from each group} \\
  \midrule
  RAND ($c = 2$) & 100 & 100 & 5 & [`$U_0$': 15\%, `$U_1$': 85\%] \\
  RAND ($c = 3$) & 100 & 100 & 5 & [`$U_0$': 5\%, `$U_1$': 20\%, `$U_2$': 75\%] \\
  Adult-Small (Race, $c = 5$) & 100 & 100 & 6 & [`Amer-Indian-Eskimo': 1\%, `Asian-Pac-Islander': 2\%, `Black': 14\%, `White': 82\%, `Others': 1\%] \\
  Adult (Gender, $c = 2$) & 1,000 & 1,000 & 6 & [`Female': 34\%, `Male': 66\%] \\
  Adult (Race, $c = 5$)   & 1,000 & 1,000 & 6 & [`Amer-Indian-Eskimo': 1\%, `Asian-Pac-Islander': 3\%, `Black': 10\%, `White': 85\%, `Others': 1\%] \\
  FourSquare-NYC ($c = 1,000$) & 882 & 1,000  & 2 & [`$u_0$': 0.1\%, \ldots, `$u_{999}$': 0.1\%] \\
  FourSquare-TKY ($c = 1,000$) & 1,132 & 1,000 & 2 & [`$u_0$': 0.1\%, \ldots, `$u_{999}$': 0.1\%] \\
  \bottomrule
  \end{tabular}
\end{table*}

\smallskip\parax{Results on Effect of $k$}
Figure~\ref{fig:mc:k} present the values of $f(S), g(S)$, where $S$ is the solution returned by each algorithm, and the runtime of each algorithm for the cardinality constraint $k = \{5, 10, \ldots, 50\}$ on \emph{Facebook} and $k = \{10, 20, \ldots, 100\}$ on \emph{Pokec} with different group partitions.
According to Figure~\ref{fig:mc:tau}, we observe that the solutions of each algorithm are close to those for solely maximizing $f$ (i.e., SM) when $\tau \leq 0.5$ and to those for solely maximizing $g$ (i.e., RSM) when $\tau \geq 0.9$. By following the $80\%$ rule, a common practice in algorithmic fairness, and distinguishing BSM from SM and RSM, we set the value of $\tau$ to $0.8$ in the experiments to better evaluate the effect of $k$.
Since computing the optimal solutions becomes infeasible on large datasets, $\mathtt{OPT}_f$, $\mathtt{OPT}_g$, and \textsc{BSM-Optimal} are omitted.
In addition, the runtime of \textsc{Greedy} and \textsc{Saturate} are not presented separately in Figure~\ref{fig:mc:k} since they are used as subroutines in \textsc{BSM-TSGreedy} and \textsc{BSM-Saturate}.
Generally, the values of $f(S)$ and $g(S)$ increase with $k$ for all algorithms.
Meanwhile, the runtime of each algorithm only slightly grows with $k$.
This is because a great number of function evaluations are reduced by applying the \emph{lazy forward} strategy, and thus the total number of function evaluations increases marginally with $k$.
\textsc{BSM-Saturate} provides higher-quality solutions than \textsc{BSM-TSGreedy} at the expense of lower efficiency for different $k$'s.
The solutions of both algorithms satisfy $g(S) \geq \tau \cdot \mathtt{OPT}'_g$ in all cases.
Finally, their time efficiencies on \emph{Pokec} confirm that they are scalable to large graphs with over one million nodes and thirty million edges.
We note that the values of $f(S)$ and $g(S)$ are much smaller on \emph{Pokec} than on \emph{Facebook} because \emph{Pokec} is larger and sparser, on which only a very small portion of users ($<7\%$) are covered by at most $100$ nodes.

\subsection{Influence Maximization}

\parax{Setup}
In the BSM framework, we propose a new problem that integrates the classic IM \cite{KempeKT03} aimed to maximize the overall influence spread over all users and the group fairness-aware IM \cite{TsangWRTZ19, BeckerCDG20} aimed to guarantee a balanced influence distribution among groups.
In particular, we are given a graph $G=(V, E, p)$, where $V$ is a set of nodes, $E$ is a set of edges, and $p: E \rightarrow \mathbb{R}^+$ is a function that assigns a propagation probability to each edge.
We follow the \emph{independent cascade} (IC) model \cite{KempeKT03} to describe the diffusion process\footnote{Note that all the algorithms we compare can be trivially extended to any diffusion model, e.g., \emph{linear threshold} and \emph{triggering} models \cite{KempeKT03}, and parameter settings where the influence spread function is monotone and submodular.} and set $p(e) = 0.1$ or $0.01$ for each $e \in E$.
We define $f_u(S) = \mathbb{P}_u(S)$, where $\mathbb{P}_u(S)$ is the probability that user $u$ is influenced by a set $S \subseteq V$ of seeds under the IC model.
Accordingly, the two functions $f$ and $g$ in Eqs.~\ref{eq:average} and~\ref{eq:maximin} are equivalent to the influence spread function in \cite{KempeKT03} divided by the number of users and the maximin welfare function in \cite{TsangWRTZ19}, respectively.
We use a \emph{reverse influence set} \cite{BorgsBCL14} (RIS) based algorithm called IMM \cite{TangSX15} for influence spread estimation.
And for any solution $S$, we run Monte-Carlo simulations 10,000 times to estimate $f(S)$ and $g(S)$.
The datasets we use in the IM experiments are almost identical to those in the MC experiments, except that the number of nodes in random graphs is reduced to $100$.

\smallskip\parax{Results}
Figures~\ref{fig:im:tau}--\ref{fig:im:k} show the results by varying the factor $\tau = \{0.1, 0.2,$ $\ldots, 0.9\}$ on random graphs when $k = 5$ and \emph{DBLP} when $k =10$, as well as the results by varying the solution size $k = \{5, 10, \ldots, 50\}$ on \emph{Facebook} and $k = \{10, 20, \ldots, 100\}$ on \emph{Pokec} when $\tau = 0.8$.
For the IM problem, the optimal solutions are infeasible even on very small graphs because computing the influence spread under the IC model is $\#$P-hard \cite{ChenWW10}.
Thus, the results of $\mathtt{OPT}_f$, $\mathtt{OPT}_g$, and \textsc{BSM-Optimal} are all ignored.
Generally, the results for IM exhibit similar trends to those for MC. Specifically, \textsc{BSM-Saturate} and \textsc{BSM-TSGreedy} strike better balances between $f$ and $g$ than SMSC.
But the solution quality of \textsc{BSM-TSGreedy} is mostly close to and sometimes higher than that of \textsc{BSM-Saturate} for IM. Meanwhile, it still runs 1.5--4\texttimes\ faster than \textsc{BSM-Saturate}.
Nevertheless, \textsc{BSM-TSGreedy} breaks the ``weak'' constraint of $g(S) \geq \tau \cdot \mathtt{OPT}'_g$ in a few cases due to the errors in influence estimations by IMM.
Finally, although our algorithms take much longer for IM than MC, they are still scalable to large graphs such as \emph{Pokec}, where the computation is always finished within 1,000 seconds.

\begin{figure*}[t]
  \centering
  \includegraphics[width=.8\linewidth]{figs/legend-1.pdf}\\
  \smallskip
  \subcaptionbox{RAND ($c = 2, k = 5$)\label{fig:fl:tau:1}}[.33\linewidth]{\includegraphics[width=.48\linewidth]{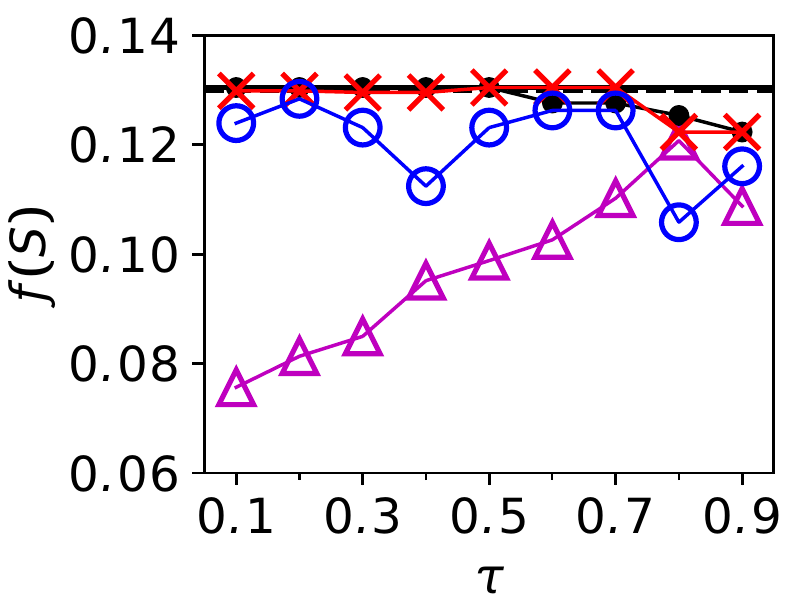} \includegraphics[width=.48\linewidth]{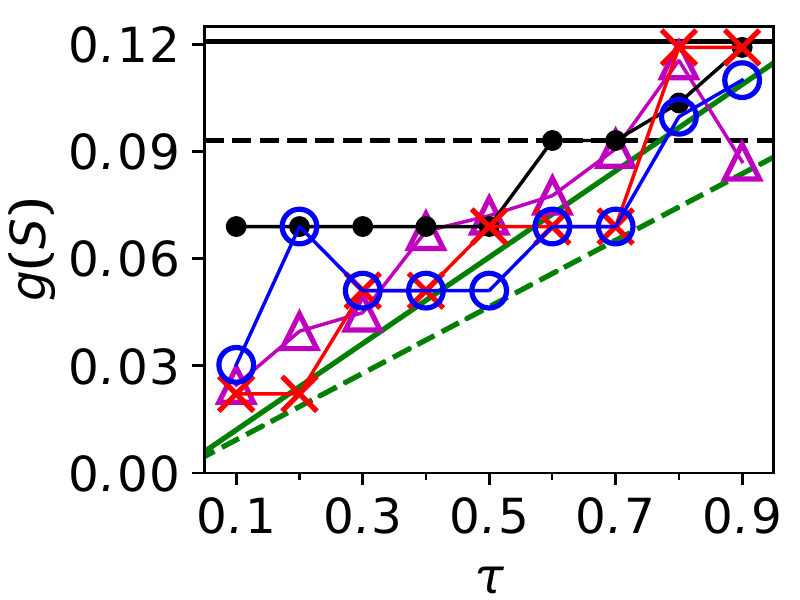}}
  \subcaptionbox{RAND ($c = 3, k = 5$)\label{fig:fl:tau:2}}[.33\linewidth]{\includegraphics[width=.48\linewidth]{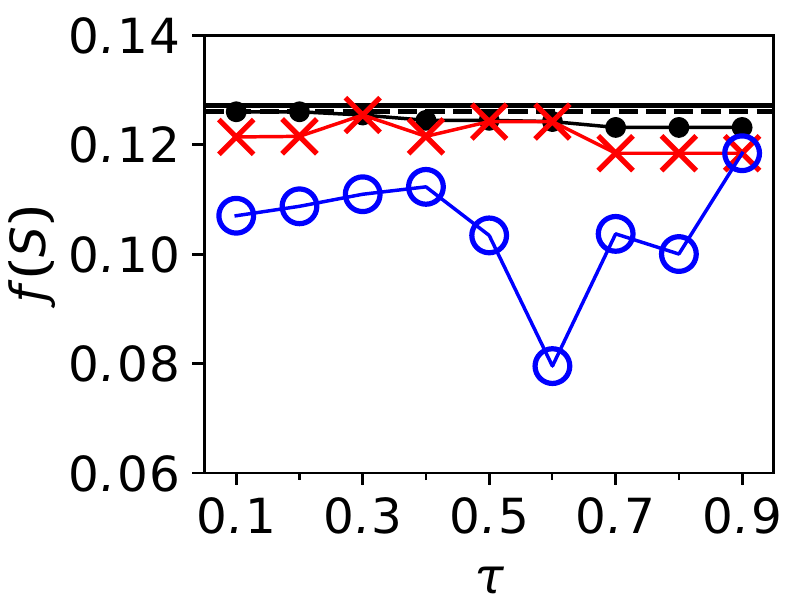} \includegraphics[width=.48\linewidth]{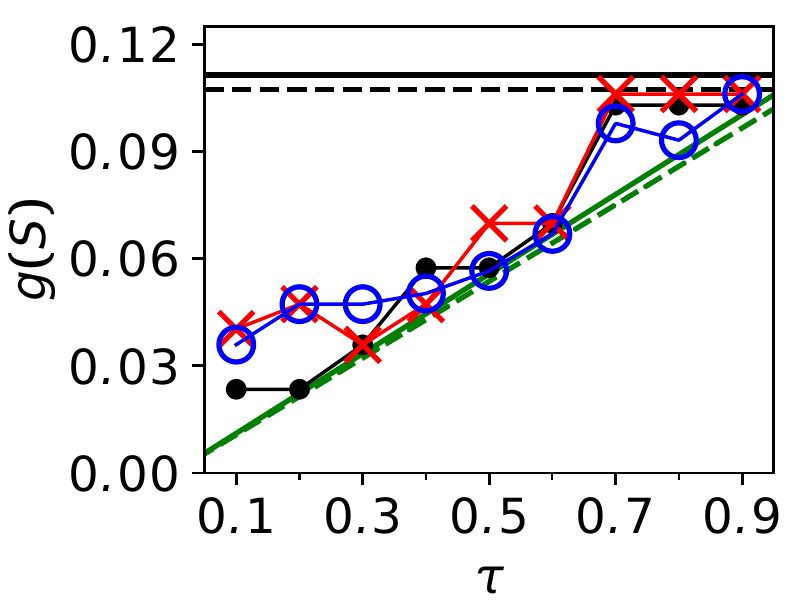}}
  \subcaptionbox{Adult-Small ($c = 5, k = 5$)\label{fig:fl:tau:3}}[.33\linewidth]{\includegraphics[width=.48\linewidth]{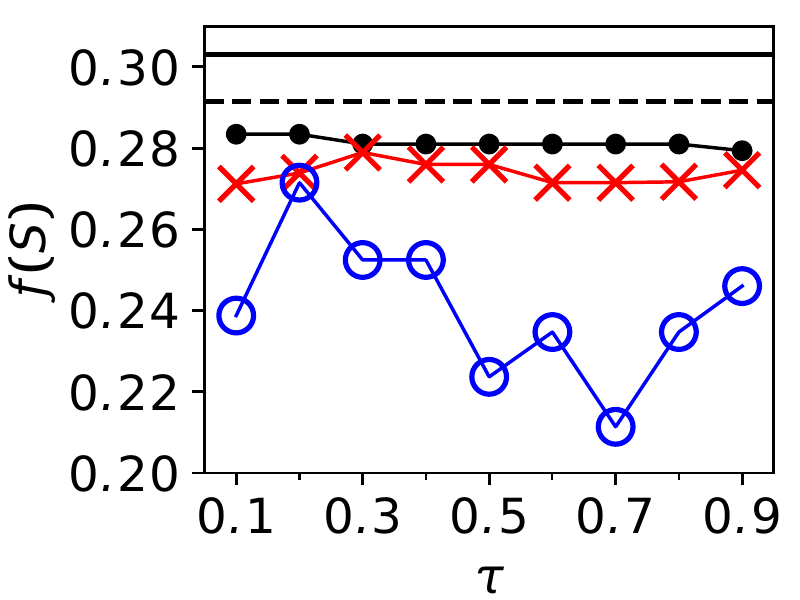} \includegraphics[width=.48\linewidth]{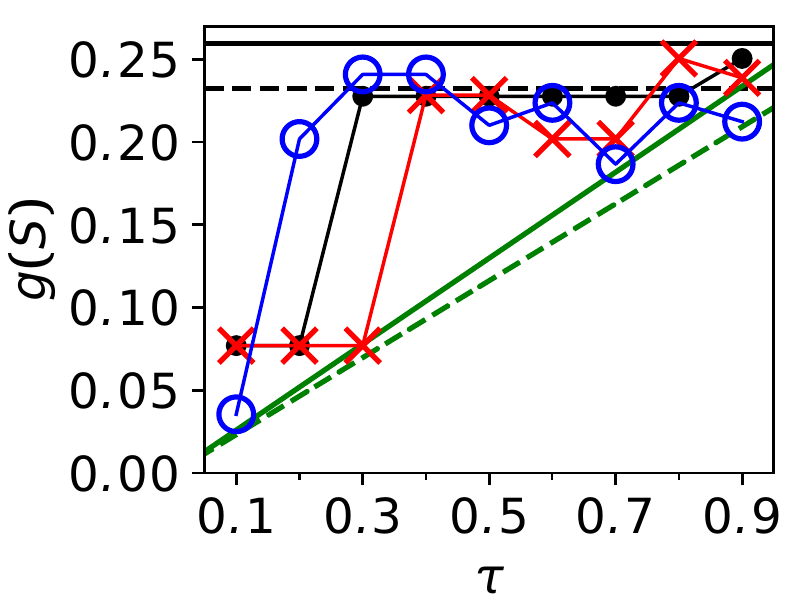}}
  \caption{Results of different algorithms for facility location by varying the factor $\tau$ on two random datasets and Adult-Small. The green straight lines are drawn in the same manner as those in Figure~\ref{fig:mc:tau}.}
  \label{fig:fl:tau}
\end{figure*}
\begin{figure*}[t]
  \centering
  \includegraphics[width=.8\linewidth]{figs/legend-2.pdf}\\
  \smallskip
  \subcaptionbox{Adult (Gender, $c = 2, \tau = 0.8$)\label{fig:fl:k:1}}[.49\linewidth]{\includegraphics[width=.32\linewidth]{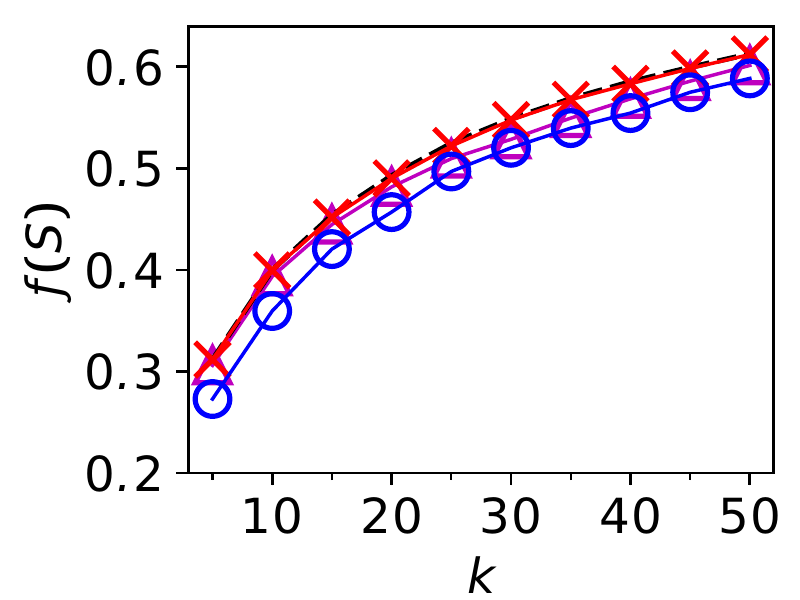} \includegraphics[width=.32\linewidth]{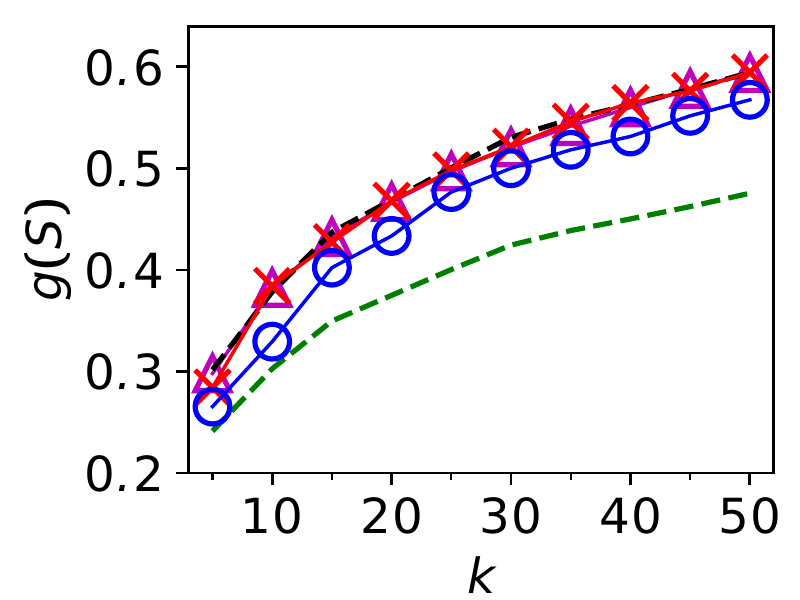} \includegraphics[width=.32\linewidth]{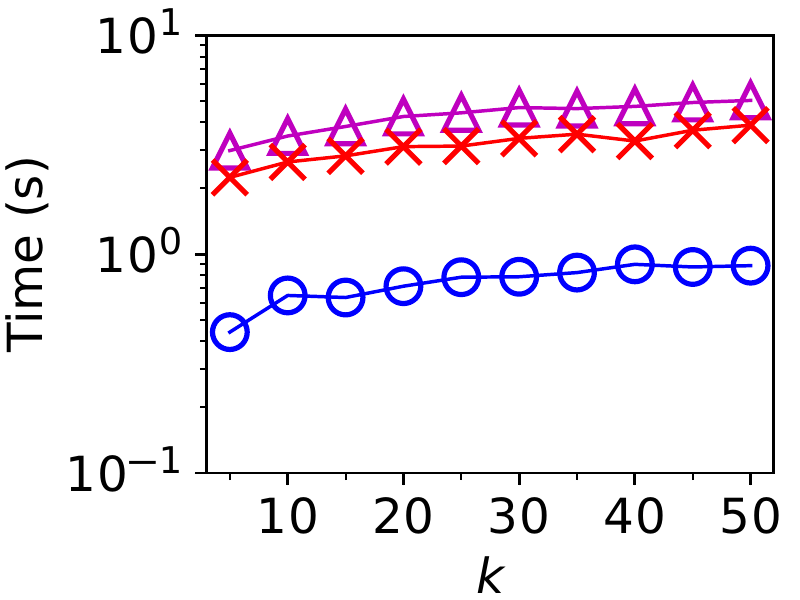}}
  \subcaptionbox{Adult (Race, $c = 5, \tau = 0.8$)\label{fig:fl:k:2}}[.49\linewidth]{\includegraphics[width=.32\linewidth]{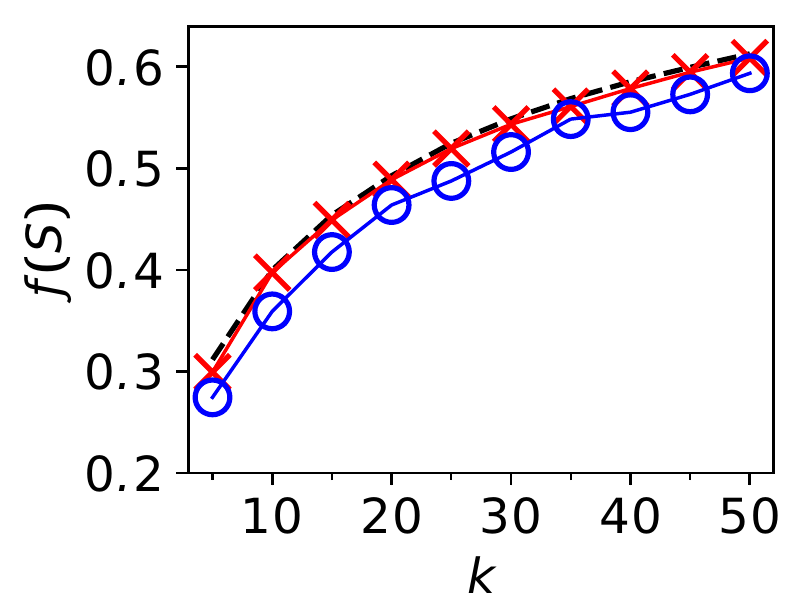} \includegraphics[width=.32\linewidth]{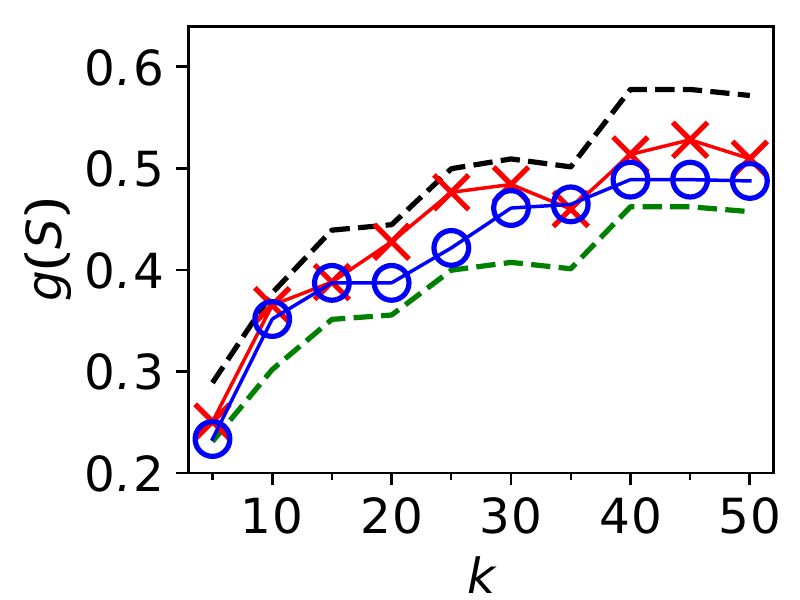} \includegraphics[width=.32\linewidth]{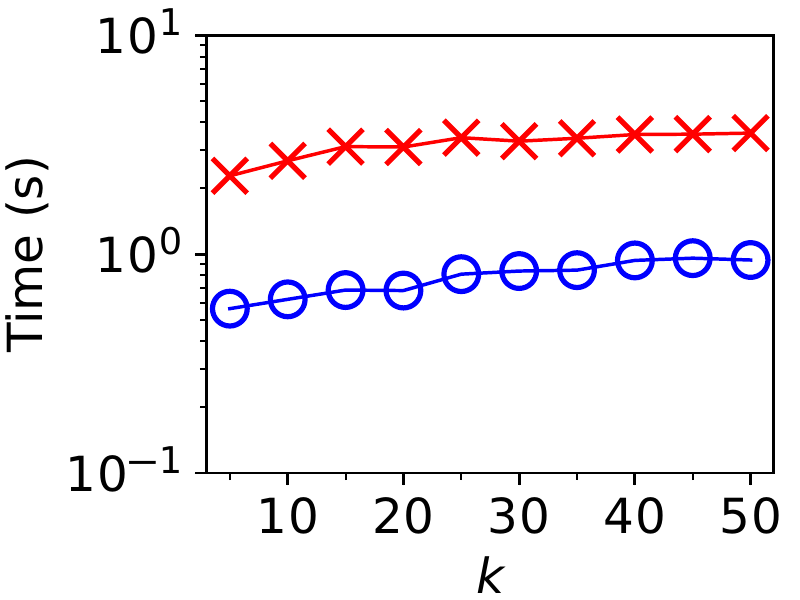}}
  \\ \smallskip
  \subcaptionbox{FourSquare-NYC ($c = 1000, \tau = 0.8$)\label{fig:fl:k:3}}[.49\linewidth]{\includegraphics[width=.32\linewidth]{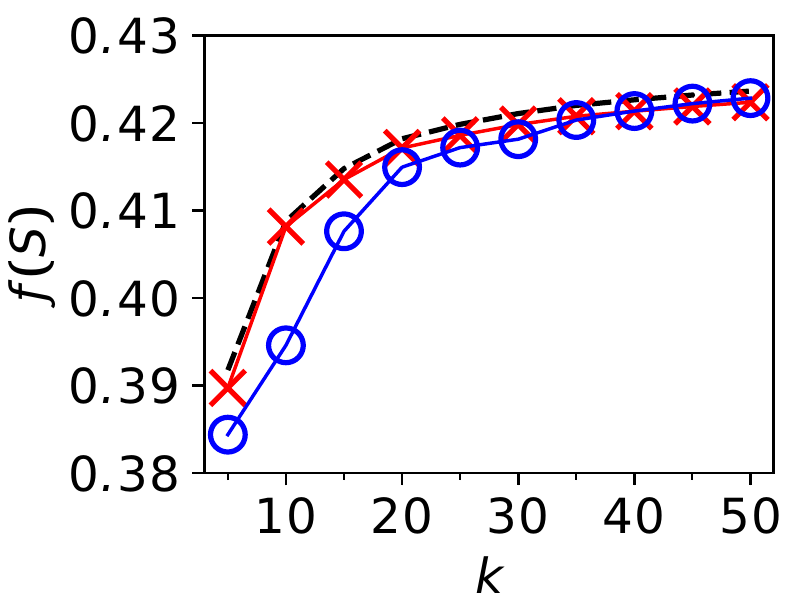} \includegraphics[width=.32\linewidth]{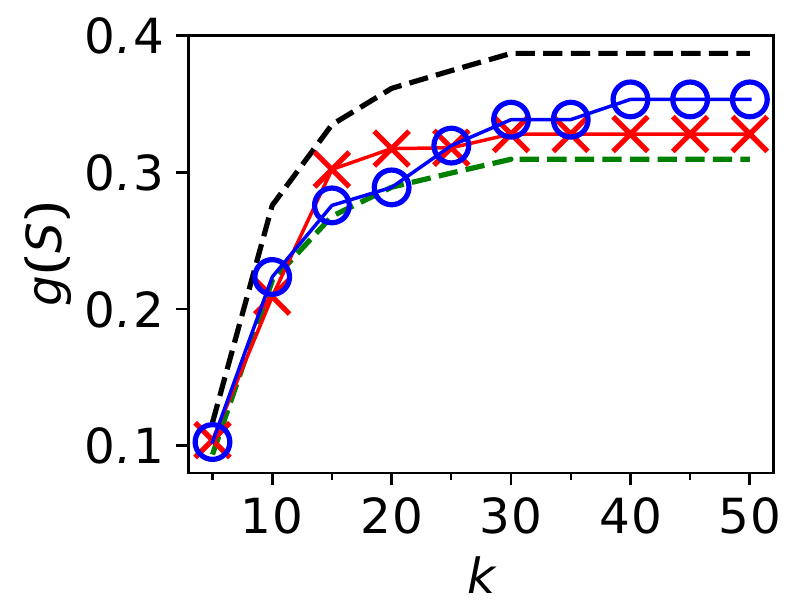} \includegraphics[width=.32\linewidth]{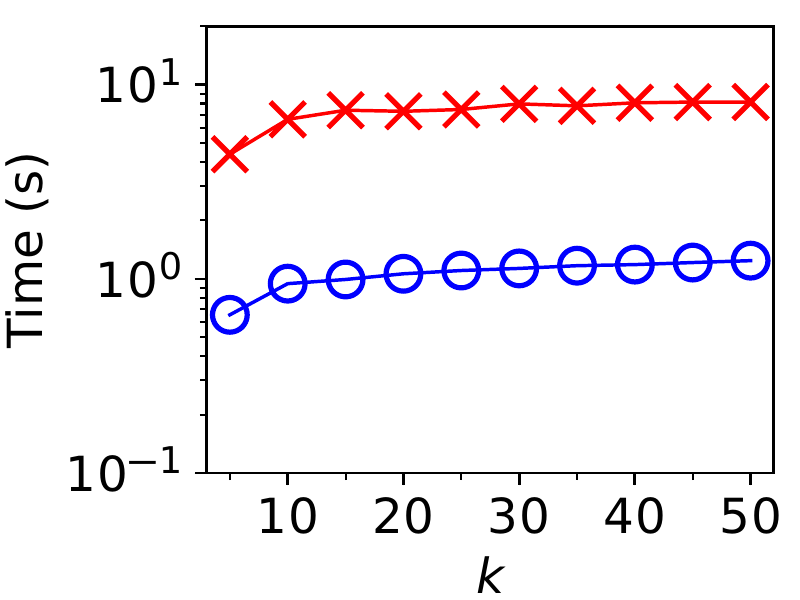}}
  \subcaptionbox{FourSquare-TKY ($c = 1000, \tau = 0.8$)\label{fig:fl:k:4}}[.49\linewidth]{\includegraphics[width=.32\linewidth]{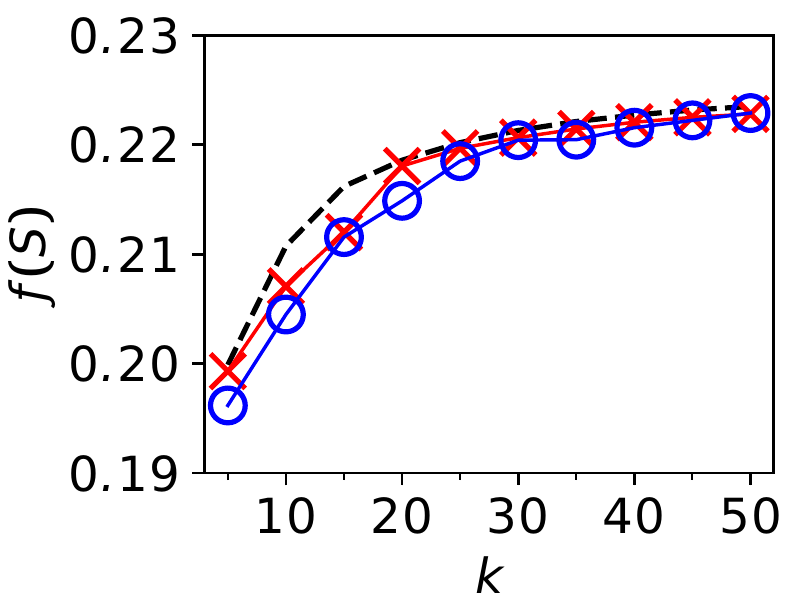} \includegraphics[width=.32\linewidth]{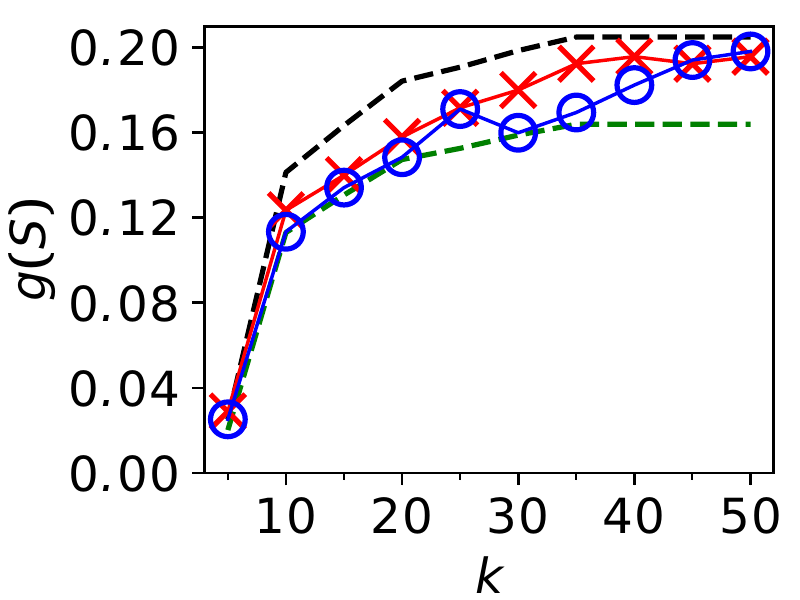} \includegraphics[width=.32\linewidth]{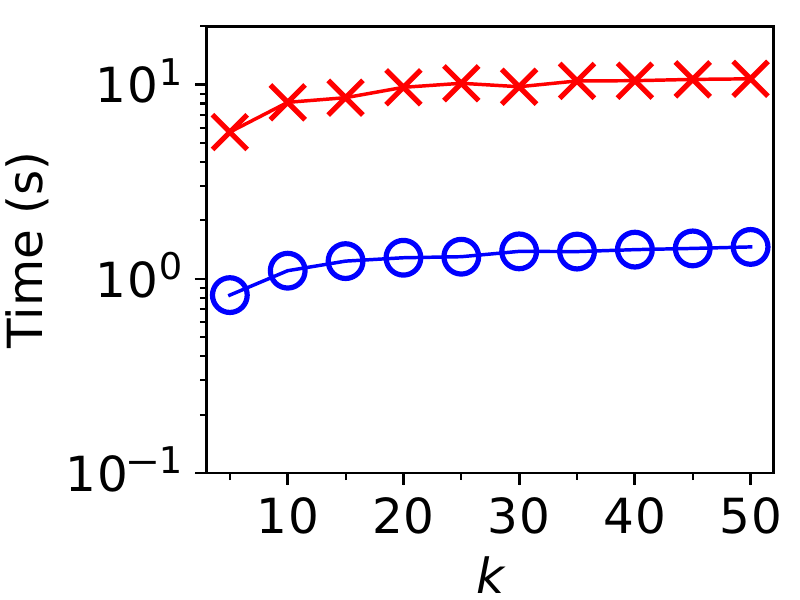}}
  \caption{Results of different algorithms for facility location by varying the solution size $k$ on Adult and FourSquare. The green dashed lines are drawn in the same manner as those in Figure~\ref{fig:mc:k}.}
  \label{fig:fl:k}
\end{figure*}

\subsection{Facility Location}

\parax{Setup}
Facility location (FL) is a general model for different real-world problems such as \emph{exemplar clustering} \cite{BadanidiyuruMKK14} and \emph{data summarization} \cite{LindgrenWD16}.
For a set $U$ of $m$ users, a set $V$ of $n$ items (facilities), and a nonnegative \emph{benefit matrix} $B \in \mathbb{R}^{m \times n}$, where $b_{uv} \in B$ denotes the benefit of item $v$ on user $u$, we define a function $f_u(S) = \max_{v \in S} b_{uv}$ to measure the benefit of set $S$ on user $u$.
In this way, the two functions $f$ and $g$ in Eqs.~\ref{eq:average} and~\ref{eq:maximin} measures the average utility for all users and the minimum of the average utilities for all groups, respectively.
And the goal of our BSM framework is to balance both objectives.
Suppose that each user $u$ or item $v$ is represented as a vector $\mathbf{p}_u$ or $\mathbf{p}_v$ in $\mathbb{R}^d$.
We use two standard methods to compute the benefits in the literature: one is based on the \emph{k-median clustering} \cite{BadanidiyuruMKK14}, i.e., $b_{uv} = \max\{0, \overline{d} - dist(\mathbf{p}_u, \mathbf{p}_v)\}$, where $dist(\mathbf{p}_u, \mathbf{p}_v)$ is the Euclidean distance between $\mathbf{p}_u$ and $\mathbf{p}_v$ and $\overline{d}$ is the distance for normalization; the other is based on the \emph{RBF kernel} \cite{LindgrenWD16}, i.e., $b_{uv} = e^{- dist(\mathbf{p}_u, \mathbf{p}_v)}$.

In the experiments, we use two public real-world datasets, namely \emph{Adult}\footnote{\url{https://archive.ics.uci.edu/ml/datasets/adult}} and \emph{FourSquare} \cite{YangZZY15}.
The \emph{Adult} dataset contains socioeconomic records of individuals, where the \emph{gender} and \emph{age} attributes are used for group partitioning. We randomly sample 100 or 1,000 records as both facilities and users and adopt RBF for benefit computation.
The \emph{FourSquare} dataset includes check-in data in New York City and Tokyo. We extract all locations of \emph{medical centers} as facilities, randomly sample 1,000 distinct check-in locations as representatives of users, and adopt the k-median function for benefit computation. Since user profiles are not available in \emph{FourSquare}, we treat each user as a single group with $c=$ 1,000 groups in total.
Moreover, we generate two random 5$d$ datasets of size 100 with 2 and 3 groups, where each group corresponds to an isotropic Gaussian blob, and RBF is used for benefit computation.
Table~\ref{tbl-stats-fl} presents the statistics of all the above datasets, where $n$ is the number of facilities, $m$ is the number of users, and $d$ is the dimension of feature vectors. We also report the percentage of users from each group in the population.

\smallskip\parax{Results}
Figures~\ref{fig:fl:tau}--\ref{fig:fl:k} show the results by varying the factor $\tau = \{0.1, 0.2,$ $\ldots, 0.9\}$ on two random datasets and \emph{Adult-Small} when $k = 5$ and the solution size $k = \{5, 10, \ldots, 50\}$ on \emph{Adult} and \emph{FourSquare} when $\tau = 0.8$.
The results for FL are also generally similar to those for MC and IM.
\textsc{BSM-Saturate} has considerable advantages over \textsc{BSM-TSGreedy} in terms of solution quality, whereas \textsc{BSM-TSGreedy} runs much faster than \textsc{BSM-Saturate}.
And they both strike better balances between $f$ and $g$ than SMSC.
Finally, their performance on \emph{FourSquare} demonstrates that they are capable of providing high-quality solutions to BSM efficiently when the number $c$ of groups is large (i.e., up to $c=1,000$).

\section{Conclusions and Future Work}
\label{sec:conclusion}

In this paper, we studied the problem of balancing utility and fairness in submodular maximization.
We formulated the problem as a bicriteria optimization problem called BSM.
Since BSM generally could not be approximated within any constant factor, we proposed two instance-dependent approximation algorithms called \textsc{BSM-TSGreedy} and \textsc{BSM-Saturate} for BSM.
We showed the effectiveness, efficiency, and scalability of our proposed algorithms by performing extensive experiments on real-world and synthetic data in three problems: maximum coverage, influence maximization, and facility location.

In future work, despite the inapproximability of BSM, we will explore how to further improve the approximation factors for BSM by performing problem-specific analyses and exploiting the correlations between group-specific utility functions. It would also be interesting to generalize BSM to non-monotone or weakly submodular functions.

\appendix

\section{ILP Formulation}
\label{app:ilp}

Although BSM is generally inapproximable, for specific classes of BSM problems, it is possible to find the optimal solution of a small instance by solving it as an integer linear programming~\cite{ip-book} (ILP) problem using any ILP solver.
This approach is referred to as the \textsc{BSM-Optimal} algorithm in the experiments (Section \ref{sec:exp}).
Here, we define the ILP formulations of \emph{maximum coverage} and \emph{facility location} in the context of BSM.
Note that these formulations are problem-specific and cannot be extended to other submodular maximization problems, such as \emph{influence maximization}.

\begin{figure*}[t]
  \centering
  \subcaptionbox{RAND (MC, $c = 2, k = 5$)}[.245\linewidth]{\includegraphics[width=.8\linewidth]{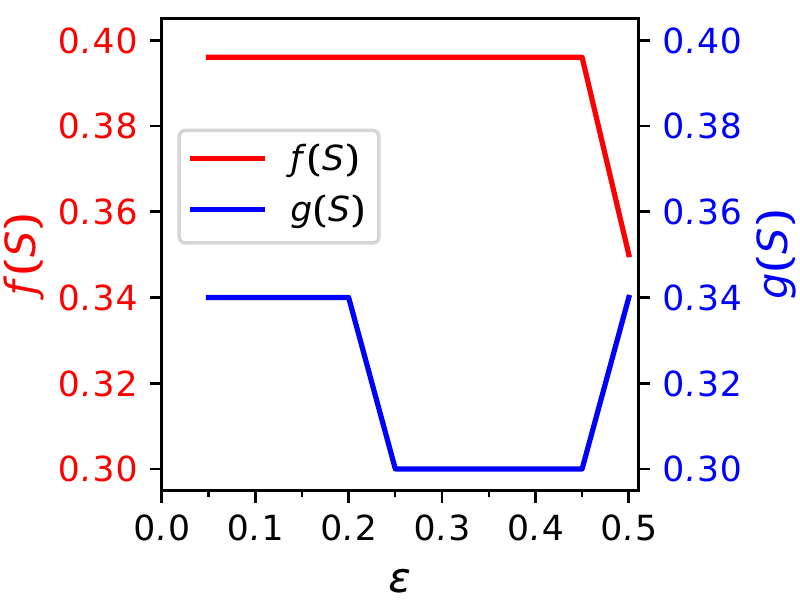}}
  \subcaptionbox{RAND (MC, $c = 4, k = 5$)}[.245\linewidth]{\includegraphics[width=.8\linewidth]{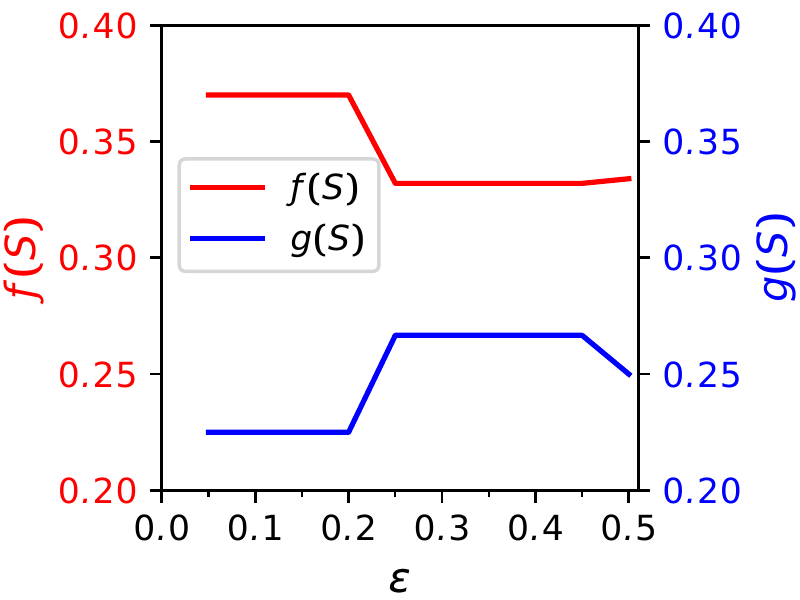}}
  \subcaptionbox{RAND (IM, $c = 2, k = 5$)}[.245\linewidth]{\includegraphics[width=.8\linewidth]{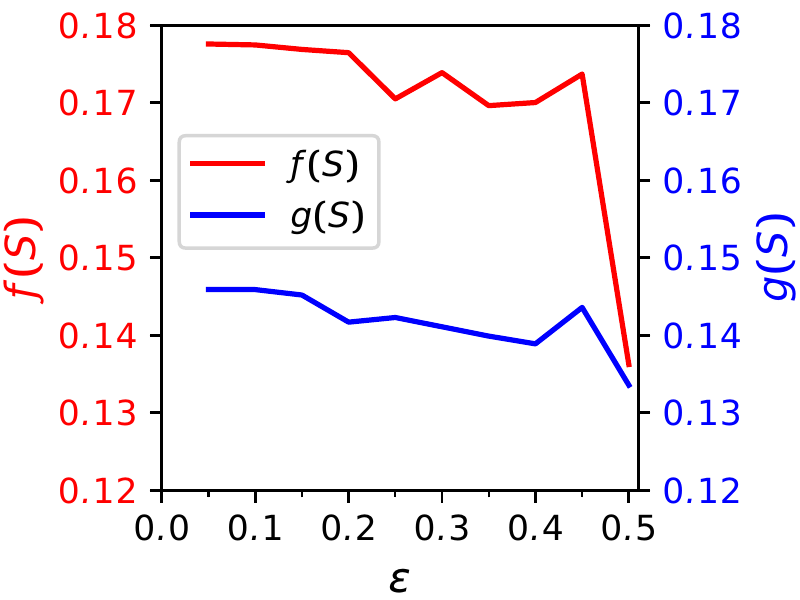}}
  \subcaptionbox{RAND (FL, $c = 2, k = 5$)}[.245\linewidth]{\includegraphics[width=.8\linewidth]{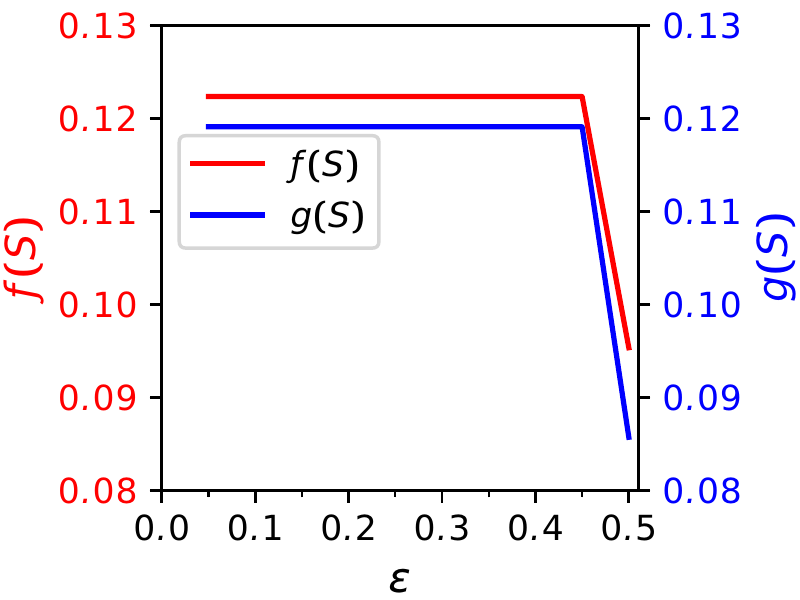}}
  \caption{Results of \textsc{BSM-Saturate} by varying the error parameter $\varepsilon$ on random datasets ($\tau = 0.8$).}
  \label{fig:eps}
\end{figure*}

By adapting the standard ILP formulation of maximum coverage\footnote{\url{https://en.wikipedia.org/wiki/Maximum_coverage_problem}}, we present an ILP to maximize the \emph{average coverage} (i.e., the function $f$ in BSM) on a universe $U = \{ u_1, \ldots, u_m \}$ of $m$ elements (users) and a collection $ V = \{ S_1, \ldots, S_n \}$ of $n$ sets (items) with a cardinality constraint $k$ as follows:
\begin{align}\label{eq:ilp:mc}
  \max        \quad & \quad \sum_{j \in [m]} \frac{y_j}{m} \\
  \text{subject to} \quad & \quad \sum_{l \in [n]} x_l \leq k \nonumber \\
                    & \quad \sum_{u_j \in S_l} x_l \geq y_j \nonumber \\
                    & \quad y_j \in \{0, 1\}, \forall j \in [m] \nonumber \\
                    & \quad x_l \in \{0, 1\}, \forall l \in [n] \nonumber
\end{align}
where $x_l$ is an indicator of whether set $S_l \in V$ is included in $S$ and $y_j$ is an indicator of whether user $u_j \in U$ is covered by $S$.

Then, we generalize the ILP in Eq.~\ref{eq:ilp:mc} to robust maximum coverage that aims to maximize the minimum of the average coverage among $c$ groups $U_1, \ldots, U_c$ (i.e., the function $g$ in BSM).
We introduce a new variable $w$ to denote the value of $g(S)$.
As such, the objective of the generalized ILP is to maximize $w$.
Meanwhile, we should incorporate new constraints on the average coverage of every group so that it does not exceed $w$ (thus, $w$ is the minimum).
In particular, the ILP formulation of \emph{robust maximum coverage} is as follows:
\begin{align}\label{eq:ilp:robust:mc}
  \max        \quad & \quad w \\
  \text{subject to} \quad & \quad \sum_{l \in [n]} x_l \leq k \nonumber \\
                    & \quad \sum_{u_j \in S_l} x_l \geq y_j \nonumber \\
                    & \quad \sum_{u_j \in U_i} \frac{y_j}{m_i} \geq w, \forall i \in [c] \nonumber \\
                    & \quad y_j \in \{0, 1\}, \forall j \in [m] \nonumber \\
                    & \quad x_l \in \{0, 1\}, \forall l \in [n] \nonumber \\
                    & \quad w \geq 0 \nonumber
\end{align}

Finally, given the optimal objective value of robust maximum coverage in Eq.~\ref{eq:ilp:robust:mc} as input $\mathtt{OPT}_g$, we can define the BSM version of maximum coverage for a balance parameter $\tau \in [0, 1]$ by adding new constraints $\sum_{u_j \in U_i} \frac{y_j}{m_i} \geq \tau \cdot \mathtt{OPT}_g$, $\forall i \in [c]$ to Eq.~\ref{eq:ilp:mc} to ensure that $g(S) \geq \tau \cdot \mathtt{OPT}_g$.

Then, for our facility location problem with a benefit matrix $B = \{b_{j l} : j \in [m], l \in [n] \} \in \mathbb{R}^{m \times n}$, we extend the ILP formulation for capacitated facility location\footnote{\url{https://en.wikipedia.org/wiki/Facility_location_problem}} as follows:
\begin{align}\label{eq:ilp:fl}
  \max        \quad & \quad \sum_{j \in [m]}\sum_{l \in [n]} \frac{b_{j l} y_{j l}}{m} \\
  \text{subject to} \quad & \quad \sum_{l \in [n]} x_l \leq k \nonumber \\
                    & \quad \sum_{l \in [n]} y_{j l} \leq 1, \forall j \in [m] \nonumber \\
                    & \quad y_{j l} \leq x_l, \forall j \in [m], l \in [n] \nonumber \\
                    & \quad y_{j l} \in \{0, 1\}, \forall j \in [m], l \in [n] \nonumber \\
                    & \quad x_l \in \{0, 1\}, \forall l \in [n] \nonumber
\end{align}
Similarly, we can also generalize Eq.~\ref{eq:ilp:fl} to the robust and BSM versions of facility location by
(\emph{i}) changing the objective to $\max w$ and adding the new constraints $\sum_{u_j \in U_i}\sum_{l \in [n]} \frac{b_{j l} y_{j l}}{m_i} \geq w, \forall i \in [c]$ and (\emph{ii}) adding the new constraints $\sum_{u_j \in U_i} \sum_{l \in [n]} \frac{b_{j l} y_{j l}}{m_i}$ $\geq \tau \cdot \mathtt{OPT}_g, \forall i \in [c]$, respectively.


\begin{figure*}[t]
  \centering
  \includegraphics[width=.8\linewidth]{figs/legend-1.pdf}\\
  \smallskip
  \subcaptionbox{Facebook (MC, Age, $c = 2, k = 5$)}[.33\linewidth]{\includegraphics[width=.48\linewidth]{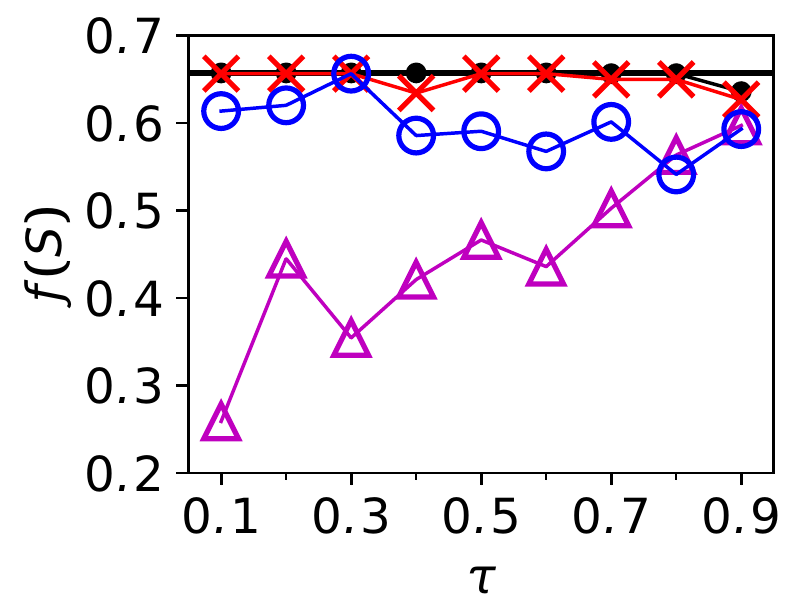} \includegraphics[width=.48\linewidth]{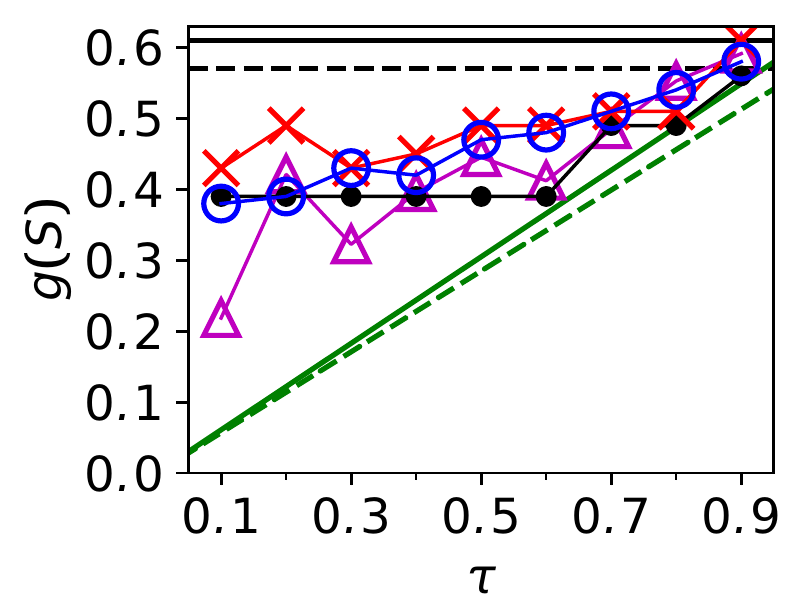}}
  \subcaptionbox{Facebook (MC, Age, $c = 4, k = 5$)}[.33\linewidth]{\includegraphics[width=.48\linewidth]{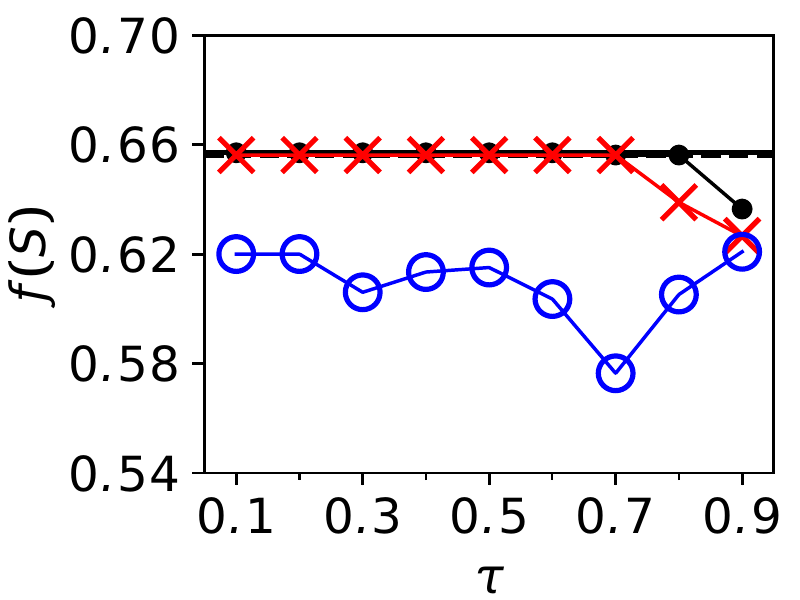} \includegraphics[width=.48\linewidth]{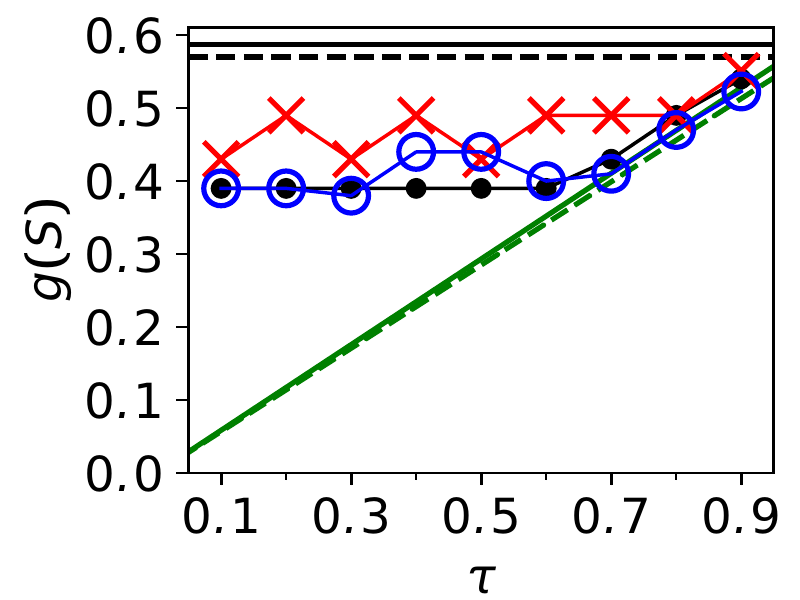}}
  \\
  \smallskip
  \subcaptionbox{Facebook (IM, Age, $c = 2, k = 5$)}[.33\linewidth]{\includegraphics[width=.48\linewidth]{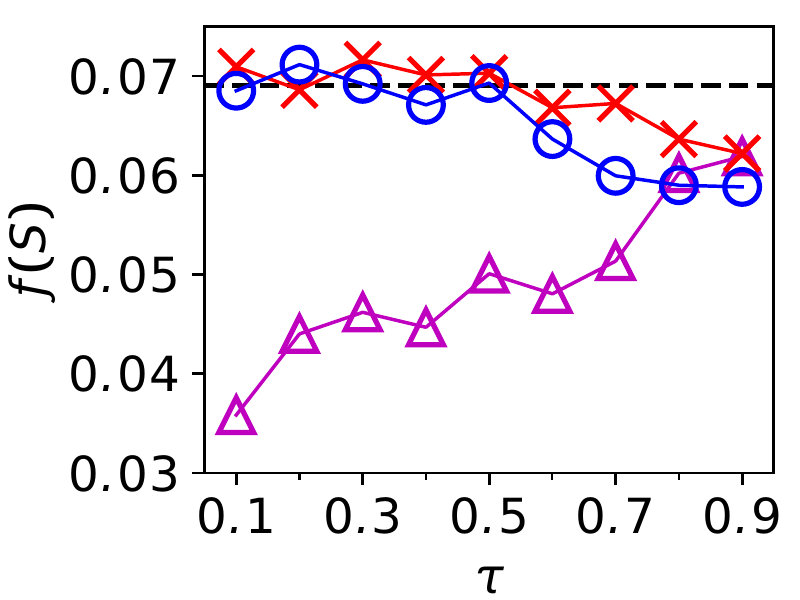} \includegraphics[width=.48\linewidth]{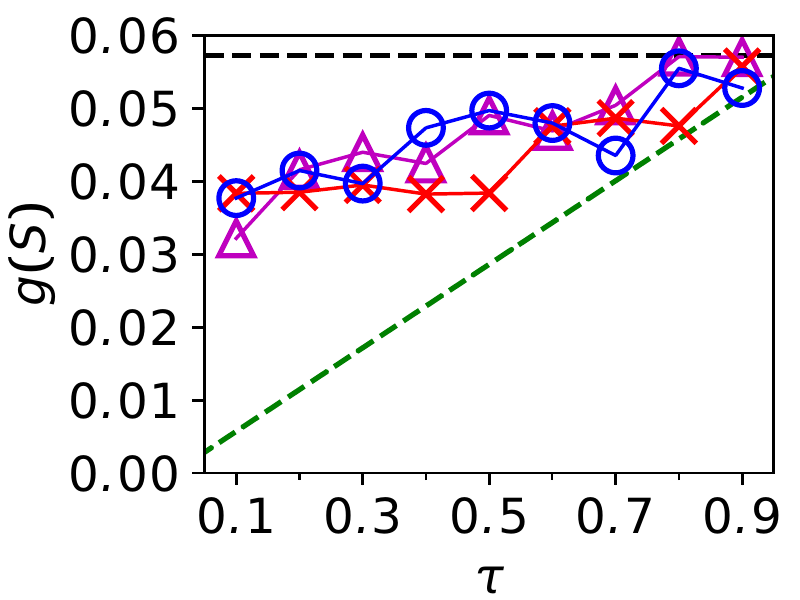}}
  \subcaptionbox{Facebook (IM, Age, $c = 4, k = 5$)}[.33\linewidth]{\includegraphics[width=.48\linewidth]{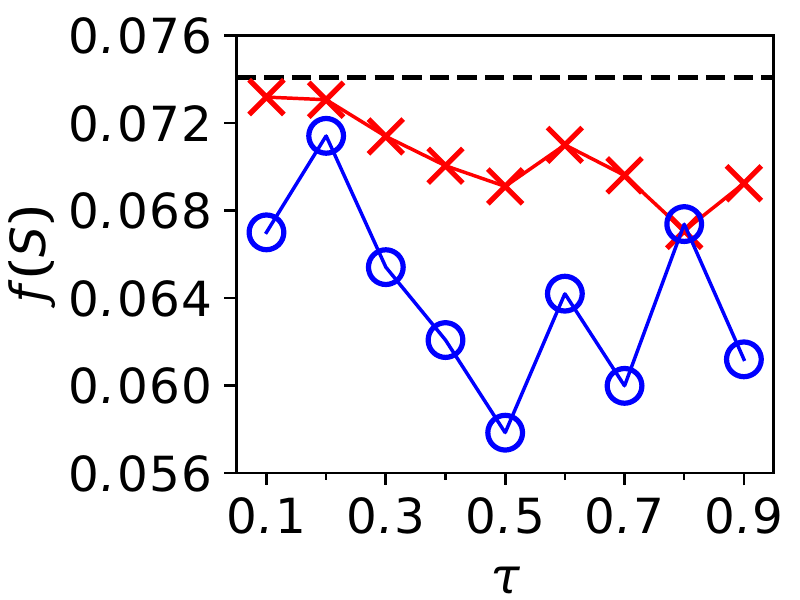} \includegraphics[width=.48\linewidth]{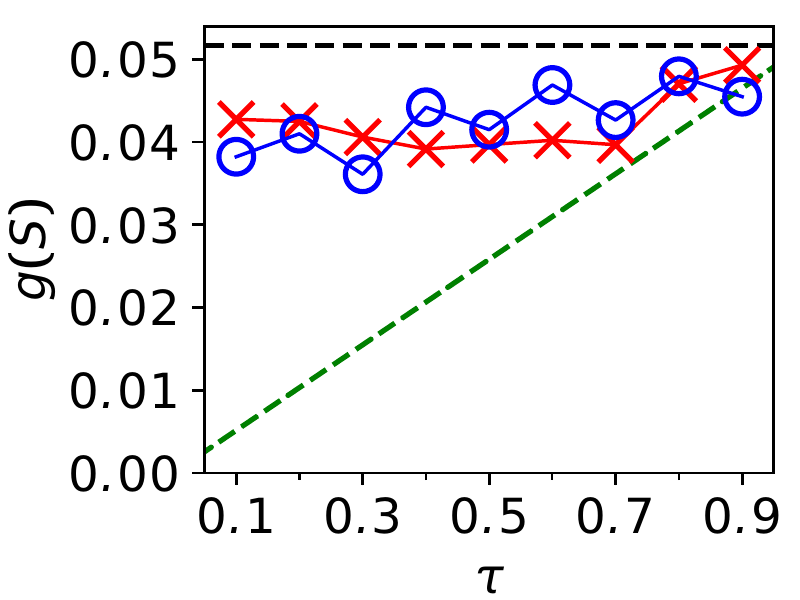}}
  \caption{Results for maximum coverage and influence maximization by varying the factor $\tau$ on Facebook.}
  \label{fig:tau-fb}
\end{figure*}
\begin{figure*}[t]
  \centering
  \includegraphics[width=.8\linewidth]{figs/legend-2.pdf}\\
  \smallskip
  \subcaptionbox{DBLP (MC, $c = 5, \tau = 0.8$)}[.49\linewidth]{\includegraphics[width=.32\linewidth]{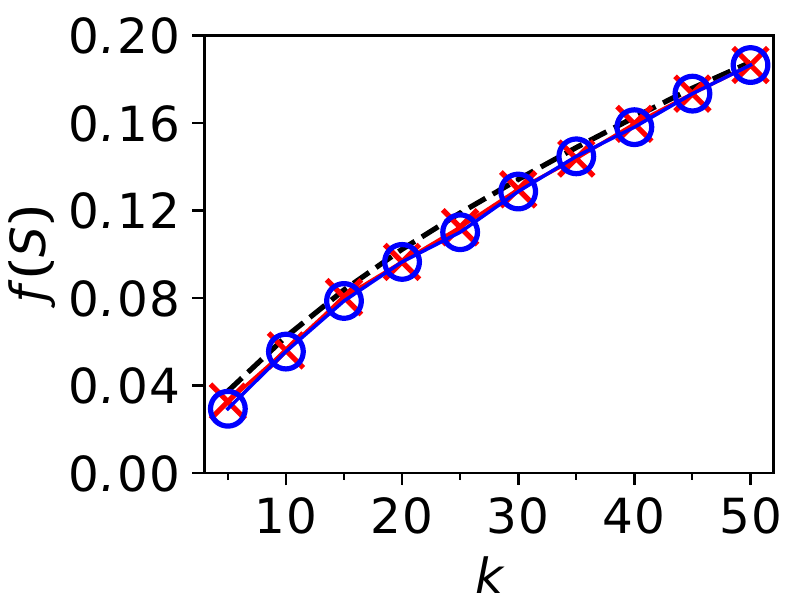} \includegraphics[width=.32\linewidth]{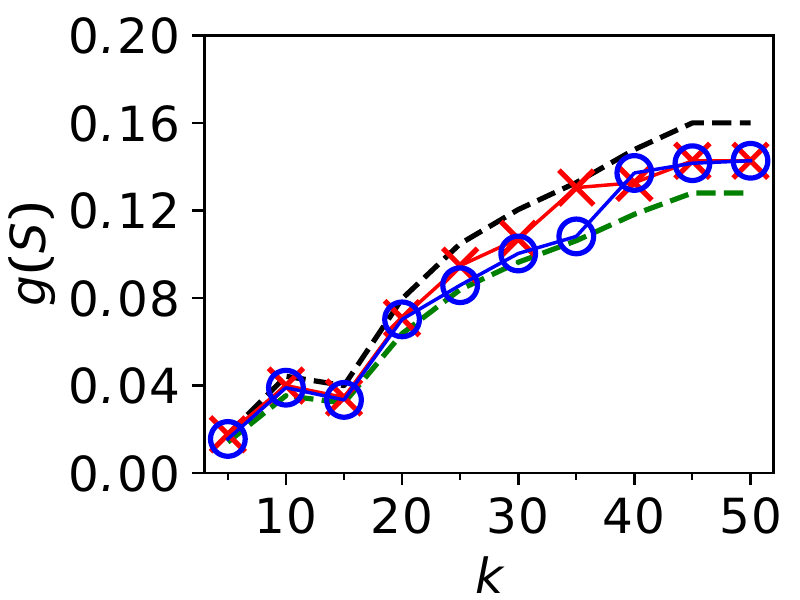} \includegraphics[width=.32\linewidth]{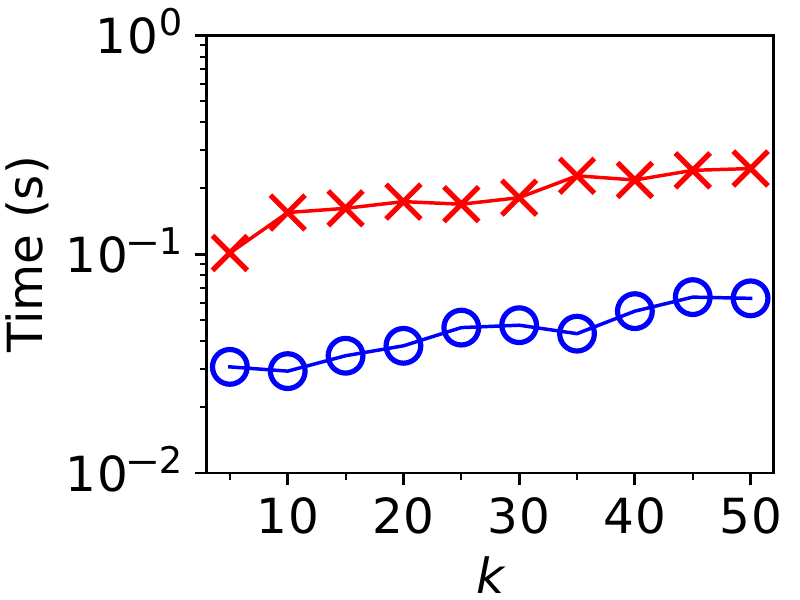}}
  \subcaptionbox{DBLP (IM, $c = 5, \tau = 0.8$)}[.49\linewidth]{\includegraphics[width=.32\linewidth]{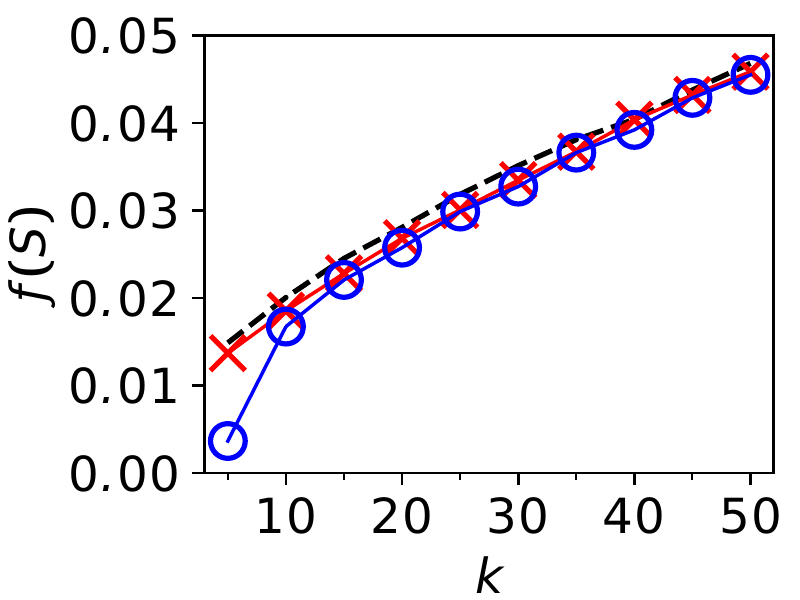} \includegraphics[width=.32\linewidth]{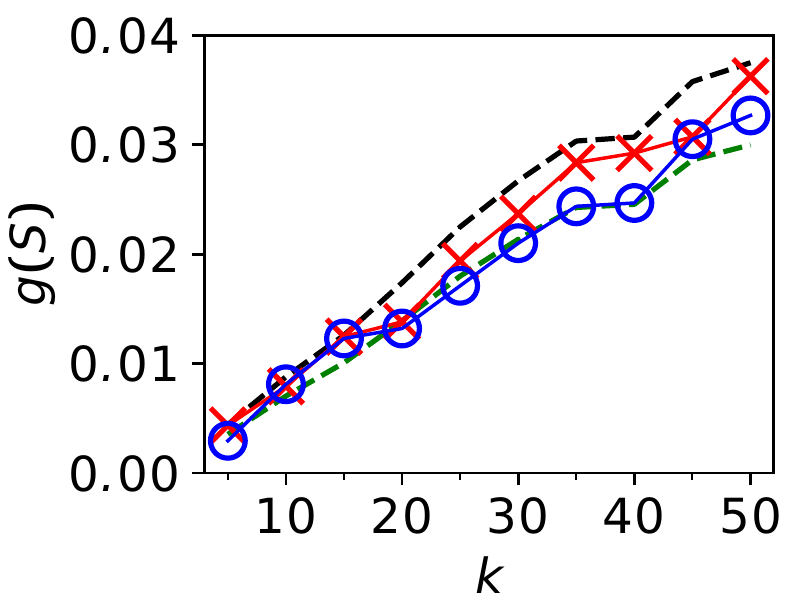} \includegraphics[width=.32\linewidth]{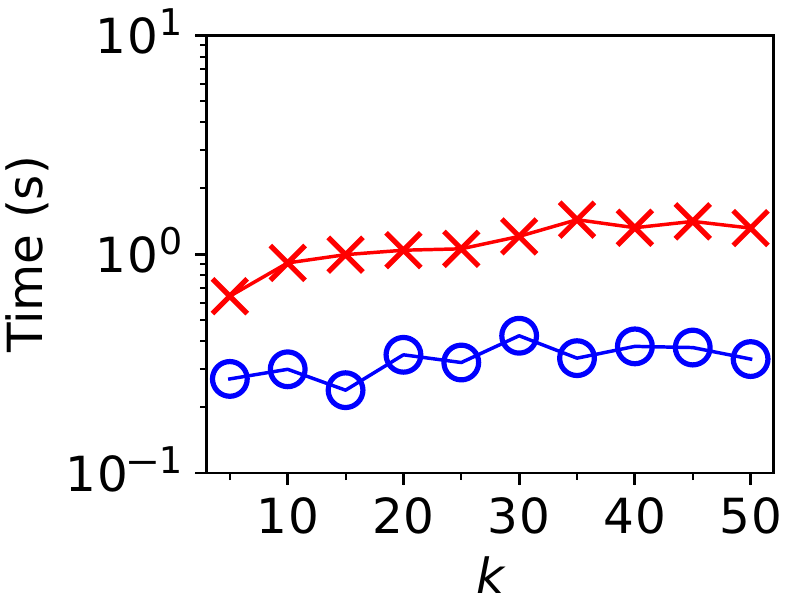}}
  \caption{Results for maximum coverage and influence maximization by varying the solution size $k$ on DBLP.}
  \label{fig:k-dblp}
\end{figure*}

\section{Additional Experiments}
\label{app:sec:exp}

\noindent\textbf{Effect of Error Parameter $\varepsilon$:} We present how the values of $f(S)$ and $g(S)$ of the solution $S$ returned by \textsc{BSM-Saturate} are affected by the error parameter $\varepsilon \in \{0.05, 0.1, \ldots, 0.5\}$ in Figure~\ref{fig:eps}.
In opposition to our theoretical analysis, we find that the value of $\varepsilon$ hardly affects the performance of \textsc{BSM-Saturate} until $\varepsilon \geq 0.5$ in practice.
Since the ratio between $\alpha_{max}$ and $\alpha_{min}$ (see Line~\ref{ln-stop} of Algorithm~\ref{alg1}) is used as the stop condition, \textsc{BSM-Saturate} will never terminate until $\alpha_{min} > 0$.
We observe that the absolute differences between $\alpha_{min}$'s at different iterations are often small since they are close to $0$. Thus, their corresponding solutions are close to each other in terms of $f(S)$ and $g(S)$.
Therefore, the values of $\alpha_{min}$, as well as $f(S)$ and $g(S)$, do not change so much for different $\varepsilon$'s smaller than $0.5$.
Nevertheless, to guarantee the good performance of \textsc{BSM-Saturate} in extreme cases, we use $\varepsilon = 0.05$ in all the experiments of Section~\ref{sec:exp}.

\smallskip
\noindent\textbf{More Results on Effects of $\tau$ and $k$:}
We provide more results on the effects of $\tau$ and $k$ in Figures~\ref{fig:tau-fb}--\ref{fig:k-dblp}, which further confirms the superiority of our proposed algorithms.

\begin{acks}
  This work was supported by the National Natural Science Foundation of China under grant number 62202169.
\end{acks}

\balance
\bibliographystyle{ACM-Reference-Format}
\bibliography{references}


\begin{thebibliography}{69}


\ifx \showCODEN    \undefined \def \showCODEN     #1{\unskip}     \fi
\ifx \showDOI      \undefined \def \showDOI       #1{#1}\fi
\ifx \showISBNx    \undefined \def \showISBNx     #1{\unskip}     \fi
\ifx \showISBNxiii \undefined \def \showISBNxiii  #1{\unskip}     \fi
\ifx \showISSN     \undefined \def \showISSN      #1{\unskip}     \fi
\ifx \showLCCN     \undefined \def \showLCCN      #1{\unskip}     \fi
\ifx \shownote     \undefined \def \shownote      #1{#1}          \fi
\ifx \showarticletitle \undefined \def \showarticletitle #1{#1}   \fi
\ifx \showURL      \undefined \def \showURL       {\relax}        \fi
\providecommand\bibfield[2]{#2}
\providecommand\bibinfo[2]{#2}
\providecommand\natexlab[1]{#1}
\providecommand\showeprint[2][]{arXiv:#2}

\bibitem[\protect\citeauthoryear{Asudeh, Berger-Wolf, DasGupta, and Sidiropoulos}{Asudeh et~al\mbox{.}}{2023}]%
        {AsudehBD20}
\bibfield{author}{\bibinfo{person}{Abolfazl Asudeh}, \bibinfo{person}{Tanya Berger-Wolf}, \bibinfo{person}{Bhaskar DasGupta}, {and} \bibinfo{person}{Anastasios Sidiropoulos}.} \bibinfo{year}{2023}\natexlab{}.
\newblock \showarticletitle{Maximizing Coverage While Ensuring Fairness: A Tale of Conflicting Objectives}.
\newblock \bibinfo{journal}{\emph{Algorithmica}}  \bibinfo{volume}{85} (\bibinfo{year}{2023}), \bibinfo{pages}{1287--1331}.
\newblock


\bibitem[\protect\citeauthoryear{Asudeh, Jagadish, Stoyanovich, and Das}{Asudeh et~al\mbox{.}}{2019}]%
        {AsudehJS019}
\bibfield{author}{\bibinfo{person}{Abolfazl Asudeh}, \bibinfo{person}{H.~V. Jagadish}, \bibinfo{person}{Julia Stoyanovich}, {and} \bibinfo{person}{Gautam Das}.} \bibinfo{year}{2019}\natexlab{}.
\newblock \showarticletitle{Designing Fair Ranking Schemes}. In \bibinfo{booktitle}{\emph{Proceedings of the 2019 International Conference on Management of Data}} \emph{(\bibinfo{series}{SIGMOD '19})}. \bibinfo{publisher}{Association for Computing Machinery}, \bibinfo{address}{New York, NY, USA}, \bibinfo{pages}{1259--1276}.
\newblock


\bibitem[\protect\citeauthoryear{Badanidiyuru, Mirzasoleiman, Karbasi, and Krause}{Badanidiyuru et~al\mbox{.}}{2014}]%
        {BadanidiyuruMKK14}
\bibfield{author}{\bibinfo{person}{Ashwinkumar Badanidiyuru}, \bibinfo{person}{Baharan Mirzasoleiman}, \bibinfo{person}{Amin Karbasi}, {and} \bibinfo{person}{Andreas Krause}.} \bibinfo{year}{2014}\natexlab{}.
\newblock \showarticletitle{Streaming Submodular Maximization: Massive Data Summarization on the Fly}. In \bibinfo{booktitle}{\emph{Proceedings of the 20th ACM SIGKDD International Conference on Knowledge Discovery and Data Mining}} \emph{(\bibinfo{series}{KDD '14})}. \bibinfo{publisher}{Association for Computing Machinery}, \bibinfo{address}{New York, NY, USA}, \bibinfo{pages}{671--680}.
\newblock


\bibitem[\protect\citeauthoryear{Bai, Iyer, Wei, and Bilmes}{Bai et~al\mbox{.}}{2016}]%
        {BaiIWB16}
\bibfield{author}{\bibinfo{person}{Wenruo Bai}, \bibinfo{person}{Rishabh Iyer}, \bibinfo{person}{Kai Wei}, {and} \bibinfo{person}{Jeff Bilmes}.} \bibinfo{year}{2016}\natexlab{}.
\newblock \showarticletitle{Algorithms for Optimizing the Ratio of Submodular Functions}. In \bibinfo{booktitle}{\emph{Proceedings of The 33rd International Conference on Machine Learning}}. \bibinfo{publisher}{PMLR}, \bibinfo{pages}{2751--2759}.
\newblock


\bibitem[\protect\citeauthoryear{Bandyapadhyay, Banik, and Bhore}{Bandyapadhyay et~al\mbox{.}}{2021}]%
        {BandyapadhyayBB21}
\bibfield{author}{\bibinfo{person}{Sayan Bandyapadhyay}, \bibinfo{person}{Aritra Banik}, {and} \bibinfo{person}{Sujoy Bhore}.} \bibinfo{year}{2021}\natexlab{}.
\newblock \showarticletitle{On Fair Covering and Hitting Problems}. In \bibinfo{booktitle}{\emph{Graph-Theoretic Concepts in Computer Science -- 47th International Workshop, WG 2021, Warsaw, Poland, June 23–25, 2021, Revised Selected Papers}}. \bibinfo{publisher}{Springer}, \bibinfo{address}{Cham}, \bibinfo{pages}{39--51}.
\newblock


\bibitem[\protect\citeauthoryear{Becker, Cor{\`{o}}, D'Angelo, and Gilbert}{Becker et~al\mbox{.}}{2020}]%
        {BeckerCDG20}
\bibfield{author}{\bibinfo{person}{Ruben Becker}, \bibinfo{person}{Federico Cor{\`{o}}}, \bibinfo{person}{Gianlorenzo D'Angelo}, {and} \bibinfo{person}{Hugo Gilbert}.} \bibinfo{year}{2020}\natexlab{}.
\newblock \showarticletitle{Balancing Spreads of Influence in a Social Network}.
\newblock \bibinfo{journal}{\emph{Proceedings of the AAAI Conference on Artificial Intelligence}} \bibinfo{volume}{34}, \bibinfo{number}{01} (\bibinfo{year}{2020}), \bibinfo{pages}{3--10}.
\newblock


\bibitem[\protect\citeauthoryear{Borgs, Brautbar, Chayes, and Lucier}{Borgs et~al\mbox{.}}{2014}]%
        {BorgsBCL14}
\bibfield{author}{\bibinfo{person}{Christian Borgs}, \bibinfo{person}{Michael Brautbar}, \bibinfo{person}{Jennifer Chayes}, {and} \bibinfo{person}{Brendan Lucier}.} \bibinfo{year}{2014}\natexlab{}.
\newblock \showarticletitle{Maximizing Social Influence in Nearly Optimal Time}. In \bibinfo{booktitle}{\emph{Proceedings of the Twenty-Fifth Annual ACM-SIAM Symposium on Discrete Algorithms}} \emph{(\bibinfo{series}{SODA '14})}. \bibinfo{publisher}{Society for Industrial and Applied Mathematics}, \bibinfo{pages}{946--957}.
\newblock


\bibitem[\protect\citeauthoryear{Bray and Weisstein}{Bray and Weisstein}{2023}]%
        {dominating-set}
\bibfield{author}{\bibinfo{person}{Nicolas Bray} {and} \bibinfo{person}{Eric~W. Weisstein}.} \bibinfo{year}{2023}\natexlab{}.
\newblock \bibinfo{title}{Dominating Set}.
\newblock
\newblock
\urldef\tempurl%
\url{https://mathworld.wolfram.com/DominatingSet.html}
\showURL{%
\tempurl}


\bibitem[\protect\citeauthoryear{C{\u{a}}linescu, Chekuri, P{\'{a}}l, and Vondr{\'{a}}k}{C{\u{a}}linescu et~al\mbox{.}}{2011}]%
        {CalinescuCPV11}
\bibfield{author}{\bibinfo{person}{Gruia C{\u{a}}linescu}, \bibinfo{person}{Chandra Chekuri}, \bibinfo{person}{Martin P{\'{a}}l}, {and} \bibinfo{person}{Jan Vondr{\'{a}}k}.} \bibinfo{year}{2011}\natexlab{}.
\newblock \showarticletitle{Maximizing a Monotone Submodular Function Subject to a Matroid Constraint}.
\newblock \bibinfo{journal}{\emph{{SIAM} J. Comput.}} \bibinfo{volume}{40}, \bibinfo{number}{6} (\bibinfo{year}{2011}), \bibinfo{pages}{1740--1766}.
\newblock


\bibitem[\protect\citeauthoryear{Calmon, Wei, Vinzamuri, Ramamurthy, and Varshney}{Calmon et~al\mbox{.}}{2017}]%
        {CalmonWVRV17}
\bibfield{author}{\bibinfo{person}{Fl{\'{a}}vio~P. Calmon}, \bibinfo{person}{Dennis Wei}, \bibinfo{person}{Bhanukiran Vinzamuri}, \bibinfo{person}{Karthikeyan~Natesan Ramamurthy}, {and} \bibinfo{person}{Kush~R. Varshney}.} \bibinfo{year}{2017}\natexlab{}.
\newblock \showarticletitle{Optimized Pre-Processing for Discrimination Prevention}.
\newblock \bibinfo{journal}{\emph{Advances in Neural Information Processing Systems}}  \bibinfo{volume}{30} (\bibinfo{year}{2017}), \bibinfo{pages}{3992--4001}.
\newblock


\bibitem[\protect\citeauthoryear{Chen, Wang, and Wang}{Chen et~al\mbox{.}}{2010}]%
        {ChenWW10}
\bibfield{author}{\bibinfo{person}{Wei Chen}, \bibinfo{person}{Chi Wang}, {and} \bibinfo{person}{Yajun Wang}.} \bibinfo{year}{2010}\natexlab{}.
\newblock \showarticletitle{Scalable Influence Maximization for Prevalent Viral Marketing in Large-Scale Social Networks}. In \bibinfo{booktitle}{\emph{Proceedings of the 16th ACM SIGKDD International Conference on Knowledge Discovery and Data Mining}} \emph{(\bibinfo{series}{KDD '10})}. \bibinfo{publisher}{Association for Computing Machinery}, \bibinfo{address}{New York, NY, USA}, \bibinfo{pages}{1029--1038}.
\newblock


\bibitem[\protect\citeauthoryear{Cui, Han, Zhu, Tang, Wu, and Huang}{Cui et~al\mbox{.}}{2021}]%
        {CuiHZZWH21}
\bibfield{author}{\bibinfo{person}{Shuang Cui}, \bibinfo{person}{Kai Han}, \bibinfo{person}{Tianshuai Zhu}, \bibinfo{person}{Jing Tang}, \bibinfo{person}{Benwei Wu}, {and} \bibinfo{person}{He Huang}.} \bibinfo{year}{2021}\natexlab{}.
\newblock \showarticletitle{Randomized Algorithms for Submodular Function Maximization with a $k$-System Constraint}. In \bibinfo{booktitle}{\emph{Proceedings of the 38th International Conference on Machine Learning}}. \bibinfo{publisher}{PMLR}, \bibinfo{pages}{2222--2232}.
\newblock


\bibitem[\protect\citeauthoryear{Dong, Ma, Wang, Chen, and Li}{Dong et~al\mbox{.}}{2023}]%
        {10097603}
\bibfield{author}{\bibinfo{person}{Yushun Dong}, \bibinfo{person}{Jing Ma}, \bibinfo{person}{Song Wang}, \bibinfo{person}{Chen Chen}, {and} \bibinfo{person}{Jundong Li}.} \bibinfo{year}{2023}\natexlab{}.
\newblock \showarticletitle{Fairness in Graph Mining: A Survey}.
\newblock \bibinfo{journal}{\emph{IEEE Trans. Knowl. Data Eng.}}  \bibinfo{volume}{PrePrints} (\bibinfo{year}{2023}), \bibinfo{pages}{1--22}.
\newblock
\urldef\tempurl%
\url{https://doi.org/10.1109/TKDE.2023.3265598}
\showDOI{\tempurl}


\bibitem[\protect\citeauthoryear{Dwork, Hardt, Pitassi, Reingold, and Zemel}{Dwork et~al\mbox{.}}{2012}]%
        {Dwork12}
\bibfield{author}{\bibinfo{person}{Cynthia Dwork}, \bibinfo{person}{Moritz Hardt}, \bibinfo{person}{Toniann Pitassi}, \bibinfo{person}{Omer Reingold}, {and} \bibinfo{person}{Richard Zemel}.} \bibinfo{year}{2012}\natexlab{}.
\newblock \showarticletitle{Fairness through Awareness}. In \bibinfo{booktitle}{\emph{Proceedings of the 3rd Innovations in Theoretical Computer Science Conference}} \emph{(\bibinfo{series}{ITCS '12})}. \bibinfo{publisher}{Association for Computing Machinery}, \bibinfo{address}{New York, NY, USA}, \bibinfo{pages}{214--226}.
\newblock


\bibitem[\protect\citeauthoryear{El~Halabi, Mitrovi\'{c}, Norouzi-Fard, Tardos, and Tarnawski}{El~Halabi et~al\mbox{.}}{2020}]%
        {NEURIPS2020_9d752cb0}
\bibfield{author}{\bibinfo{person}{Marwa El~Halabi}, \bibinfo{person}{Slobodan Mitrovi\'{c}}, \bibinfo{person}{Ashkan Norouzi-Fard}, \bibinfo{person}{Jakab Tardos}, {and} \bibinfo{person}{Jakub~M. Tarnawski}.} \bibinfo{year}{2020}\natexlab{}.
\newblock \showarticletitle{Fairness in Streaming Submodular Maximization: Algorithms and Hardness}.
\newblock \bibinfo{journal}{\emph{Advances in Neural Information Processing Systems}}  \bibinfo{volume}{33} (\bibinfo{year}{2020}), \bibinfo{pages}{13609--13622}.
\newblock


\bibitem[\protect\citeauthoryear{Emelianov, Gast, Gummadi, and Loiseau}{Emelianov et~al\mbox{.}}{2022}]%
        {Gummadi_selection}
\bibfield{author}{\bibinfo{person}{Vitalii Emelianov}, \bibinfo{person}{Nicolas Gast}, \bibinfo{person}{Krishna~P. Gummadi}, {and} \bibinfo{person}{Patrick Loiseau}.} \bibinfo{year}{2022}\natexlab{}.
\newblock \showarticletitle{On fair selection in the presence of implicit and differential variance}.
\newblock \bibinfo{journal}{\emph{Artif. Intell.}}  \bibinfo{volume}{302} (\bibinfo{year}{2022}), \bibinfo{pages}{103609}.
\newblock


\bibitem[\protect\citeauthoryear{Epasto, Lattanzi, Vassilvitskii, and Zadimoghaddam}{Epasto et~al\mbox{.}}{2017}]%
        {EpastoLVZ17}
\bibfield{author}{\bibinfo{person}{Alessandro Epasto}, \bibinfo{person}{Silvio Lattanzi}, \bibinfo{person}{Sergei Vassilvitskii}, {and} \bibinfo{person}{Morteza Zadimoghaddam}.} \bibinfo{year}{2017}\natexlab{}.
\newblock \showarticletitle{Submodular Optimization Over Sliding Windows}. In \bibinfo{booktitle}{\emph{Proceedings of the 26th International Conference on World Wide Web}} \emph{(\bibinfo{series}{WWW '17})}. \bibinfo{publisher}{Association for Computing Machinery}, \bibinfo{address}{New York, NY, USA}, \bibinfo{pages}{421--430}.
\newblock


\bibitem[\protect\citeauthoryear{Esmaeili, Duppala, Nanda, Srinivasan, and Dickerson}{Esmaeili et~al\mbox{.}}{2022}]%
        {EsmaeiliDNSD22}
\bibfield{author}{\bibinfo{person}{Seyed~A. Esmaeili}, \bibinfo{person}{Sharmila Duppala}, \bibinfo{person}{Vedant Nanda}, \bibinfo{person}{Aravind Srinivasan}, {and} \bibinfo{person}{John~P. Dickerson}.} \bibinfo{year}{2022}\natexlab{}.
\newblock \showarticletitle{Rawlsian Fairness in Online Bipartite Matching: Two-Sided, Group, and Individual}. In \bibinfo{booktitle}{\emph{Proceedings of the 21st International Conference on Autonomous Agents and Multiagent Systems}} \emph{(\bibinfo{series}{AAMAS '22})}. \bibinfo{publisher}{Association for Computing Machinery}, \bibinfo{address}{New York, NY, USA}, \bibinfo{pages}{1583--1585}.
\newblock


\bibitem[\protect\citeauthoryear{Feige}{Feige}{1998}]%
        {Feige98}
\bibfield{author}{\bibinfo{person}{Uriel Feige}.} \bibinfo{year}{1998}\natexlab{}.
\newblock \showarticletitle{A Threshold of ln \emph{n} for Approximating Set Cover}.
\newblock \bibinfo{journal}{\emph{J. {ACM}}} \bibinfo{volume}{45}, \bibinfo{number}{4} (\bibinfo{year}{1998}), \bibinfo{pages}{634--652}.
\newblock


\bibitem[\protect\citeauthoryear{Fu, Bhatt, Basu, and Pavan}{Fu et~al\mbox{.}}{2021}]%
        {FuBBP21}
\bibfield{author}{\bibinfo{person}{Xiaoyun Fu}, \bibinfo{person}{Rishabh~Rajendra Bhatt}, \bibinfo{person}{Samik Basu}, {and} \bibinfo{person}{Aduri Pavan}.} \bibinfo{year}{2021}\natexlab{}.
\newblock \showarticletitle{Multi-Objective Submodular Optimization with Approximate Oracles and Influence Maximization}. In \bibinfo{booktitle}{\emph{2021 IEEE International Conference on Big Data (Big Data)}}. \bibinfo{publisher}{IEEE}, \bibinfo{pages}{328--334}.
\newblock


\bibitem[\protect\citeauthoryear{Garc{\'{\i}}a{-}Soriano and Bonchi}{Garc{\'{\i}}a{-}Soriano and Bonchi}{2020}]%
        {dgs2020}
\bibfield{author}{\bibinfo{person}{David Garc{\'{\i}}a{-}Soriano} {and} \bibinfo{person}{Francesco Bonchi}.} \bibinfo{year}{2020}\natexlab{}.
\newblock \showarticletitle{Fair-by-design matching}.
\newblock \bibinfo{journal}{\emph{Data Min. Knowl. Discov.}} \bibinfo{volume}{34}, \bibinfo{number}{5} (\bibinfo{year}{2020}), \bibinfo{pages}{1291--1335}.
\newblock


\bibitem[\protect\citeauthoryear{Garc\'{\i}a-Soriano and Bonchi}{Garc\'{\i}a-Soriano and Bonchi}{2021}]%
        {dgs2021}
\bibfield{author}{\bibinfo{person}{David Garc\'{\i}a-Soriano} {and} \bibinfo{person}{Francesco Bonchi}.} \bibinfo{year}{2021}\natexlab{}.
\newblock \showarticletitle{Maxmin-Fair Ranking: Individual Fairness under Group-Fairness Constraints}. In \bibinfo{booktitle}{\emph{Proceedings of the 27th ACM SIGKDD Conference on Knowledge Discovery \& Data Mining}} \emph{(\bibinfo{series}{KDD '21})}. \bibinfo{publisher}{Association for Computing Machinery}, \bibinfo{address}{New York, NY, USA}, \bibinfo{pages}{436--446}.
\newblock


\bibitem[\protect\citeauthoryear{Gershtein, Milo, and Youngmann}{Gershtein et~al\mbox{.}}{2021}]%
        {GershteinMY21}
\bibfield{author}{\bibinfo{person}{Shay Gershtein}, \bibinfo{person}{Tova Milo}, {and} \bibinfo{person}{Brit Youngmann}.} \bibinfo{year}{2021}\natexlab{}.
\newblock \showarticletitle{Multi-Objective Influence Maximization}. In \bibinfo{booktitle}{\emph{Proceedings of the 24th International Conference on Extending Database Technology (EDBT '21)}}. \bibinfo{publisher}{OpenProceedings.org}, \bibinfo{pages}{145--156}.
\newblock


\bibitem[\protect\citeauthoryear{Ghadiri, Samadi, and Vempala}{Ghadiri et~al\mbox{.}}{2021}]%
        {GhadiriSV21}
\bibfield{author}{\bibinfo{person}{Mehrdad Ghadiri}, \bibinfo{person}{Samira Samadi}, {and} \bibinfo{person}{Santosh Vempala}.} \bibinfo{year}{2021}\natexlab{}.
\newblock \showarticletitle{Socially Fair k-Means Clustering}. In \bibinfo{booktitle}{\emph{Proceedings of the 2021 ACM Conference on Fairness, Accountability, and Transparency}} \emph{(\bibinfo{series}{FAccT '21})}. \bibinfo{publisher}{Association for Computing Machinery}, \bibinfo{address}{New York, NY, USA}, \bibinfo{pages}{438--448}.
\newblock


\bibitem[\protect\citeauthoryear{Hardt, Price, and Srebro}{Hardt et~al\mbox{.}}{2016}]%
        {HardtPNS16}
\bibfield{author}{\bibinfo{person}{Moritz Hardt}, \bibinfo{person}{Eric Price}, {and} \bibinfo{person}{Nati Srebro}.} \bibinfo{year}{2016}\natexlab{}.
\newblock \showarticletitle{Equality of Opportunity in Supervised Learning}.
\newblock \bibinfo{journal}{\emph{Advances in Neural Information Processing Systems}}  \bibinfo{volume}{29} (\bibinfo{year}{2016}), \bibinfo{pages}{3315--3323}.
\newblock


\bibitem[\protect\citeauthoryear{Holland, Laskey, and Leinhardt}{Holland et~al\mbox{.}}{1983}]%
        {Holland83}
\bibfield{author}{\bibinfo{person}{Paul~W. Holland}, \bibinfo{person}{Kathryn~Blackmond Laskey}, {and} \bibinfo{person}{Samuel Leinhardt}.} \bibinfo{year}{1983}\natexlab{}.
\newblock \showarticletitle{Stochastic blockmodels: First steps}.
\newblock \bibinfo{journal}{\emph{Soc. Netw.}} \bibinfo{volume}{5}, \bibinfo{number}{2} (\bibinfo{year}{1983}), \bibinfo{pages}{109--137}.
\newblock


\bibitem[\protect\citeauthoryear{Iyer and Bilmes}{Iyer and Bilmes}{2012}]%
        {IyerB12}
\bibfield{author}{\bibinfo{person}{Rishabh~K. Iyer} {and} \bibinfo{person}{Jeff~A. Bilmes}.} \bibinfo{year}{2012}\natexlab{}.
\newblock \showarticletitle{Algorithms for Approximate Minimization of the Difference Between Submodular Functions, with Applications}. In \bibinfo{booktitle}{\emph{Proceedings of the Twenty-Eighth Conference on Uncertainty in Artificial Intelligence}}. \bibinfo{publisher}{AUAI Press}, \bibinfo{pages}{407--417}.
\newblock


\bibitem[\protect\citeauthoryear{Iyer and Bilmes}{Iyer and Bilmes}{2013}]%
        {IyerB13}
\bibfield{author}{\bibinfo{person}{Rishabh~K. Iyer} {and} \bibinfo{person}{Jeff~A. Bilmes}.} \bibinfo{year}{2013}\natexlab{}.
\newblock \showarticletitle{Submodular Optimization with Submodular Cover and Submodular Knapsack Constraints}.
\newblock \bibinfo{journal}{\emph{Advances in Neural Information Processing Systems}}  \bibinfo{volume}{26} (\bibinfo{year}{2013}), \bibinfo{pages}{2436--2444}.
\newblock


\bibitem[\protect\citeauthoryear{Jagadish, Bonchi, Eliassi-Rad, Getoor, Gummadi, and Stoyanovich}{Jagadish et~al\mbox{.}}{2019}]%
        {JagadishBEGGS19}
\bibfield{author}{\bibinfo{person}{H.~V. Jagadish}, \bibinfo{person}{Francesco Bonchi}, \bibinfo{person}{Tina Eliassi-Rad}, \bibinfo{person}{Lise Getoor}, \bibinfo{person}{Krishna Gummadi}, {and} \bibinfo{person}{Julia Stoyanovich}.} \bibinfo{year}{2019}\natexlab{}.
\newblock \showarticletitle{The Responsibility Challenge for Data}. In \bibinfo{booktitle}{\emph{Proceedings of the 2019 International Conference on Management of Data}} \emph{(\bibinfo{series}{SIGMOD '19})}. \bibinfo{publisher}{Association for Computing Machinery}, \bibinfo{address}{New York, NY, USA}, \bibinfo{pages}{412--414}.
\newblock


\bibitem[\protect\citeauthoryear{Jin, Yang, Yang, Shi, Huang, and Xiao}{Jin et~al\mbox{.}}{2021}]%
        {JinYYSHX21}
\bibfield{author}{\bibinfo{person}{Tianyuan Jin}, \bibinfo{person}{Yu Yang}, \bibinfo{person}{Renchi Yang}, \bibinfo{person}{Jieming Shi}, \bibinfo{person}{Keke Huang}, {and} \bibinfo{person}{Xiaokui Xiao}.} \bibinfo{year}{2021}\natexlab{}.
\newblock \showarticletitle{Unconstrained Submodular Maximization with Modular Costs: Tight Approximation and Application to Profit Maximization}.
\newblock \bibinfo{journal}{\emph{Proc. {VLDB} Endow.}} \bibinfo{volume}{14}, \bibinfo{number}{10} (\bibinfo{year}{2021}), \bibinfo{pages}{1756--1768}.
\newblock


\bibitem[\protect\citeauthoryear{Kazemi, Zadimoghaddam, and Karbasi}{Kazemi et~al\mbox{.}}{2018}]%
        {KazemiZK18}
\bibfield{author}{\bibinfo{person}{Ehsan Kazemi}, \bibinfo{person}{Morteza Zadimoghaddam}, {and} \bibinfo{person}{Amin Karbasi}.} \bibinfo{year}{2018}\natexlab{}.
\newblock \showarticletitle{Scalable Deletion-Robust Submodular Maximization: Data Summarization with Privacy and Fairness Constraints}. In \bibinfo{booktitle}{\emph{Proceedings of the 35th International Conference on Machine Learning}}. \bibinfo{publisher}{PMLR}, \bibinfo{pages}{2544--2553}.
\newblock


\bibitem[\protect\citeauthoryear{Kearns, Roth, and Wu}{Kearns et~al\mbox{.}}{2017}]%
        {KearnsRW17}
\bibfield{author}{\bibinfo{person}{Michael~J. Kearns}, \bibinfo{person}{Aaron Roth}, {and} \bibinfo{person}{Zhiwei~Steven Wu}.} \bibinfo{year}{2017}\natexlab{}.
\newblock \showarticletitle{Meritocratic Fairness for Cross-Population Selection}. In \bibinfo{booktitle}{\emph{Proceedings of the 34th International Conference on Machine Learning}}. \bibinfo{publisher}{PMLR}, \bibinfo{pages}{1828--1836}.
\newblock


\bibitem[\protect\citeauthoryear{Kempe, Kleinberg, and Tardos}{Kempe et~al\mbox{.}}{2003}]%
        {KempeKT03}
\bibfield{author}{\bibinfo{person}{David Kempe}, \bibinfo{person}{Jon Kleinberg}, {and} \bibinfo{person}{\'{E}va Tardos}.} \bibinfo{year}{2003}\natexlab{}.
\newblock \showarticletitle{Maximizing the Spread of Influence through a Social Network}. In \bibinfo{booktitle}{\emph{Proceedings of the Ninth ACM SIGKDD International Conference on Knowledge Discovery and Data Mining}} \emph{(\bibinfo{series}{KDD '03})}. \bibinfo{publisher}{Association for Computing Machinery}, \bibinfo{address}{New York, NY, USA}, \bibinfo{pages}{137--146}.
\newblock


\bibitem[\protect\citeauthoryear{Kleinberg, Ludwig, Mullainathan, and Sunstein}{Kleinberg et~al\mbox{.}}{2019}]%
        {10.1093/jla/laz001}
\bibfield{author}{\bibinfo{person}{Jon Kleinberg}, \bibinfo{person}{Jens Ludwig}, \bibinfo{person}{Sendhil Mullainathan}, {and} \bibinfo{person}{Cass~R. Sunstein}.} \bibinfo{year}{2019}\natexlab{}.
\newblock \showarticletitle{Discrimination in the Age of Algorithms}.
\newblock \bibinfo{journal}{\emph{J. Leg. Anal.}}  \bibinfo{volume}{10} (\bibinfo{year}{2019}), \bibinfo{pages}{113--174}.
\newblock


\bibitem[\protect\citeauthoryear{Krause and Golovin}{Krause and Golovin}{2014}]%
        {KrauseG14}
\bibfield{author}{\bibinfo{person}{Andreas Krause} {and} \bibinfo{person}{Daniel Golovin}.} \bibinfo{year}{2014}\natexlab{}.
\newblock \showarticletitle{Submodular Function Maximization}.
\newblock In \bibinfo{booktitle}{\emph{Tractability: Practical Approaches to Hard Problems}}. \bibinfo{publisher}{Cambridge University Press}, \bibinfo{address}{Cambridge, UK}, \bibinfo{pages}{71--104}.
\newblock


\bibitem[\protect\citeauthoryear{Krause, McMahan, Guestrin, and Gupta}{Krause et~al\mbox{.}}{2008}]%
        {Krause08}
\bibfield{author}{\bibinfo{person}{Andreas Krause}, \bibinfo{person}{H.~Brendan McMahan}, \bibinfo{person}{Carlos Guestrin}, {and} \bibinfo{person}{Anupam Gupta}.} \bibinfo{year}{2008}\natexlab{}.
\newblock \showarticletitle{Robust Submodular Observation Selection}.
\newblock \bibinfo{journal}{\emph{J. Mach. Learn. Res.}} \bibinfo{volume}{9}, \bibinfo{number}{12} (\bibinfo{year}{2008}), \bibinfo{pages}{2761--2801}.
\newblock


\bibitem[\protect\citeauthoryear{Leskovec, Krause, Guestrin, Faloutsos, Van~Briesen, and Glance}{Leskovec et~al\mbox{.}}{2007}]%
        {LeskovecKGFVG07}
\bibfield{author}{\bibinfo{person}{Jure Leskovec}, \bibinfo{person}{Andreas Krause}, \bibinfo{person}{Carlos Guestrin}, \bibinfo{person}{Christos Faloutsos}, \bibinfo{person}{Jeanne Van~Briesen}, {and} \bibinfo{person}{Natalie Glance}.} \bibinfo{year}{2007}\natexlab{}.
\newblock \showarticletitle{Cost-Effective Outbreak Detection in Networks}. In \bibinfo{booktitle}{\emph{Proceedings of the 13th ACM SIGKDD International Conference on Knowledge Discovery and Data Mining}} \emph{(\bibinfo{series}{KDD '07})}. \bibinfo{publisher}{Association for Computing Machinery}, \bibinfo{address}{New York, NY, USA}, \bibinfo{pages}{420--429}.
\newblock


\bibitem[\protect\citeauthoryear{Li, Fan, Wang, and Tan}{Li et~al\mbox{.}}{2018}]%
        {LiFWT18}
\bibfield{author}{\bibinfo{person}{Yuchen Li}, \bibinfo{person}{Ju Fan}, \bibinfo{person}{Yanhao Wang}, {and} \bibinfo{person}{Kian{-}Lee Tan}.} \bibinfo{year}{2018}\natexlab{}.
\newblock \showarticletitle{Influence Maximization on Social Graphs: {A} Survey}.
\newblock \bibinfo{journal}{\emph{{IEEE} Trans. Knowl. Data Eng.}} \bibinfo{volume}{30}, \bibinfo{number}{10} (\bibinfo{year}{2018}), \bibinfo{pages}{1852--1872}.
\newblock


\bibitem[\protect\citeauthoryear{Lindgren, Wu, and Dimakis}{Lindgren et~al\mbox{.}}{2016}]%
        {LindgrenWD16}
\bibfield{author}{\bibinfo{person}{Erik~M. Lindgren}, \bibinfo{person}{Shanshan Wu}, {and} \bibinfo{person}{Alexandros~G. Dimakis}.} \bibinfo{year}{2016}\natexlab{}.
\newblock \showarticletitle{Leveraging Sparsity for Efficient Submodular Data Summarization}.
\newblock \bibinfo{journal}{\emph{Advances in Neural Information Processing Systems}}  \bibinfo{volume}{29} (\bibinfo{year}{2016}), \bibinfo{pages}{3414--3422}.
\newblock


\bibitem[\protect\citeauthoryear{Makarychev and Vakilian}{Makarychev and Vakilian}{2021}]%
        {MakarychevV21}
\bibfield{author}{\bibinfo{person}{Yury Makarychev} {and} \bibinfo{person}{Ali Vakilian}.} \bibinfo{year}{2021}\natexlab{}.
\newblock \showarticletitle{Approximation Algorithms for Socially Fair Clustering}. In \bibinfo{booktitle}{\emph{Proceedings of Thirty-Fourth Conference on Learning Theory}}. \bibinfo{publisher}{PMLR}, \bibinfo{pages}{3246--3264}.
\newblock


\bibitem[\protect\citeauthoryear{Mashiat, Gitiaux, Rangwala, Fowler, and Das}{Mashiat et~al\mbox{.}}{2022}]%
        {MashiatGRFD22}
\bibfield{author}{\bibinfo{person}{Tasfia Mashiat}, \bibinfo{person}{Xavier Gitiaux}, \bibinfo{person}{Huzefa Rangwala}, \bibinfo{person}{Patrick Fowler}, {and} \bibinfo{person}{Sanmay Das}.} \bibinfo{year}{2022}\natexlab{}.
\newblock \showarticletitle{Trade-Offs between Group Fairness Metrics in Societal Resource Allocation}. In \bibinfo{booktitle}{\emph{Proceedings of the 2022 ACM Conference on Fairness, Accountability, and Transparency}} \emph{(\bibinfo{series}{FAccT '22})}. \bibinfo{publisher}{Association for Computing Machinery}, \bibinfo{address}{New York, NY, USA}, \bibinfo{pages}{1095--1105}.
\newblock


\bibitem[\protect\citeauthoryear{Mehrabi, Morstatter, Saxena, Lerman, and Galstyan}{Mehrabi et~al\mbox{.}}{2022}]%
        {MehrabiMSLG21}
\bibfield{author}{\bibinfo{person}{Ninareh Mehrabi}, \bibinfo{person}{Fred Morstatter}, \bibinfo{person}{Nripsuta Saxena}, \bibinfo{person}{Kristina Lerman}, {and} \bibinfo{person}{Aram Galstyan}.} \bibinfo{year}{2022}\natexlab{}.
\newblock \showarticletitle{A Survey on Bias and Fairness in Machine Learning}.
\newblock \bibinfo{journal}{\emph{{ACM} Comput. Surv.}} \bibinfo{volume}{54}, \bibinfo{number}{6}, Article \bibinfo{articleno}{115} (\bibinfo{year}{2022}), \bibinfo{numpages}{35}~pages.
\newblock


\bibitem[\protect\citeauthoryear{Mirzasoleiman, Badanidiyuru, and Karbasi}{Mirzasoleiman et~al\mbox{.}}{2016a}]%
        {MirzasoleimanBK16}
\bibfield{author}{\bibinfo{person}{Baharan Mirzasoleiman}, \bibinfo{person}{Ashwinkumar Badanidiyuru}, {and} \bibinfo{person}{Amin Karbasi}.} \bibinfo{year}{2016}\natexlab{a}.
\newblock \showarticletitle{Fast Constrained Submodular Maximization: Personalized Data Summarization}. In \bibinfo{booktitle}{\emph{Proceedings of The 33rd International Conference on Machine Learning}}. \bibinfo{publisher}{PMLR}, \bibinfo{pages}{1358--1367}.
\newblock


\bibitem[\protect\citeauthoryear{Mirzasoleiman, Badanidiyuru, Karbasi, Vondr{\'{a}}k, and Krause}{Mirzasoleiman et~al\mbox{.}}{2015}]%
        {MirzasoleimanBK15}
\bibfield{author}{\bibinfo{person}{Baharan Mirzasoleiman}, \bibinfo{person}{Ashwinkumar Badanidiyuru}, \bibinfo{person}{Amin Karbasi}, \bibinfo{person}{Jan Vondr{\'{a}}k}, {and} \bibinfo{person}{Andreas Krause}.} \bibinfo{year}{2015}\natexlab{}.
\newblock \showarticletitle{Lazier Than Lazy Greedy}.
\newblock \bibinfo{journal}{\emph{Proceedings of the AAAI Conference on Artificial Intelligence}}  \bibinfo{volume}{29} (\bibinfo{year}{2015}), \bibinfo{pages}{1812--1818}.
\newblock


\bibitem[\protect\citeauthoryear{Mirzasoleiman, Karbasi, and Krause}{Mirzasoleiman et~al\mbox{.}}{2017}]%
        {MirzasoleimanKK17}
\bibfield{author}{\bibinfo{person}{Baharan Mirzasoleiman}, \bibinfo{person}{Amin Karbasi}, {and} \bibinfo{person}{Andreas Krause}.} \bibinfo{year}{2017}\natexlab{}.
\newblock \showarticletitle{Deletion-Robust Submodular Maximization: Data Summarization with ``the Right to be Forgotten''}. In \bibinfo{booktitle}{\emph{Proceedings of the 34th International Conference on Machine Learning}}. \bibinfo{publisher}{PMLR}, \bibinfo{pages}{2449--2458}.
\newblock


\bibitem[\protect\citeauthoryear{Mirzasoleiman, Karbasi, Sarkar, and Krause}{Mirzasoleiman et~al\mbox{.}}{2016b}]%
        {MirzasoleimanKS16}
\bibfield{author}{\bibinfo{person}{Baharan Mirzasoleiman}, \bibinfo{person}{Amin Karbasi}, \bibinfo{person}{Rik Sarkar}, {and} \bibinfo{person}{Andreas Krause}.} \bibinfo{year}{2016}\natexlab{b}.
\newblock \showarticletitle{Distributed Submodular Maximization}.
\newblock \bibinfo{journal}{\emph{J. Mach. Learn. Res.}} \bibinfo{volume}{17}, \bibinfo{number}{235} (\bibinfo{year}{2016}), \bibinfo{pages}{1--44}.
\newblock


\bibitem[\protect\citeauthoryear{Mislove, Viswanath, Gummadi, and Druschel}{Mislove et~al\mbox{.}}{2010}]%
        {MisloveVGD10}
\bibfield{author}{\bibinfo{person}{Alan Mislove}, \bibinfo{person}{Bimal Viswanath}, \bibinfo{person}{Krishna~P. Gummadi}, {and} \bibinfo{person}{Peter Druschel}.} \bibinfo{year}{2010}\natexlab{}.
\newblock \showarticletitle{You Are Who You Know: Inferring User Profiles in Online Social Networks}. In \bibinfo{booktitle}{\emph{Proceedings of the Third ACM International Conference on Web Search and Data Mining}} \emph{(\bibinfo{series}{WSDM '10})}. \bibinfo{publisher}{Association for Computing Machinery}, \bibinfo{address}{New York, NY, USA}, \bibinfo{pages}{251--260}.
\newblock


\bibitem[\protect\citeauthoryear{Monemizadeh}{Monemizadeh}{2020}]%
        {NEURIPS2020_6fbd841e}
\bibfield{author}{\bibinfo{person}{Morteza Monemizadeh}.} \bibinfo{year}{2020}\natexlab{}.
\newblock \showarticletitle{Dynamic Submodular Maximization}.
\newblock \bibinfo{journal}{\emph{Advances in Neural Information Processing Systems}}  \bibinfo{volume}{33} (\bibinfo{year}{2020}), \bibinfo{pages}{9806--9817}.
\newblock


\bibitem[\protect\citeauthoryear{Nemhauser, Wolsey, and Fisher}{Nemhauser et~al\mbox{.}}{1978}]%
        {NemhauserWF78}
\bibfield{author}{\bibinfo{person}{George~L. Nemhauser}, \bibinfo{person}{Laurence~A. Wolsey}, {and} \bibinfo{person}{Marshall~L. Fisher}.} \bibinfo{year}{1978}\natexlab{}.
\newblock \showarticletitle{An analysis of approximations for maximizing submodular set functions - {I}}.
\newblock \bibinfo{journal}{\emph{Math. Program.}} \bibinfo{volume}{14}, \bibinfo{number}{1} (\bibinfo{year}{1978}), \bibinfo{pages}{265--294}.
\newblock


\bibitem[\protect\citeauthoryear{Nguyen and Thai}{Nguyen and Thai}{2021}]%
        {NguyenT21}
\bibfield{author}{\bibinfo{person}{Lan~N. Nguyen} {and} \bibinfo{person}{My~T. Thai}.} \bibinfo{year}{2021}\natexlab{}.
\newblock \showarticletitle{Minimum Robust Multi-Submodular Cover for Fairness}.
\newblock \bibinfo{journal}{\emph{Proceedings of the AAAI Conference on Artificial Intelligence}} \bibinfo{volume}{35}, \bibinfo{number}{10} (\bibinfo{year}{2021}), \bibinfo{pages}{9109--9116}.
\newblock


\bibitem[\protect\citeauthoryear{Nikolakaki, Ene, and Terzi}{Nikolakaki et~al\mbox{.}}{2021}]%
        {NikolakakiET21}
\bibfield{author}{\bibinfo{person}{Sofia~Maria Nikolakaki}, \bibinfo{person}{Alina Ene}, {and} \bibinfo{person}{Evimaria Terzi}.} \bibinfo{year}{2021}\natexlab{}.
\newblock \showarticletitle{An Efficient Framework for Balancing Submodularity and Cost}. In \bibinfo{booktitle}{\emph{Proceedings of the 27th ACM SIGKDD Conference on Knowledge Discovery \& Data Mining}} \emph{(\bibinfo{series}{KDD '21})}. \bibinfo{publisher}{Association for Computing Machinery}, \bibinfo{address}{New York, NY, USA}, \bibinfo{pages}{1256--1266}.
\newblock


\bibitem[\protect\citeauthoryear{Ohsaka and Matsuoka}{Ohsaka and Matsuoka}{2021}]%
        {OhsakaM21}
\bibfield{author}{\bibinfo{person}{Naoto Ohsaka} {and} \bibinfo{person}{Tatsuya Matsuoka}.} \bibinfo{year}{2021}\natexlab{}.
\newblock \showarticletitle{Approximation algorithm for submodular maximization under submodular cover}. In \bibinfo{booktitle}{\emph{Proceedings of the Thirty-Seventh Conference on Uncertainty in Artificial Intelligence}}. \bibinfo{publisher}{PMLR}, \bibinfo{pages}{792--801}.
\newblock


\bibitem[\protect\citeauthoryear{Parambath, Vijayakumar, and Chawla}{Parambath et~al\mbox{.}}{2018}]%
        {ParambathVC18}
\bibfield{author}{\bibinfo{person}{Shameem~Puthiya Parambath}, \bibinfo{person}{Nishant Vijayakumar}, {and} \bibinfo{person}{Sanjay Chawla}.} \bibinfo{year}{2018}\natexlab{}.
\newblock \showarticletitle{SAGA: A Submodular Greedy Algorithm for Group Recommendation}.
\newblock \bibinfo{journal}{\emph{Proceedings of the AAAI Conference on Artificial Intelligence}}  \bibinfo{volume}{32} (\bibinfo{year}{2018}), \bibinfo{pages}{3900--3908}.
\newblock


\bibitem[\protect\citeauthoryear{Rawls}{Rawls}{1971}]%
        {Rawls71}
\bibfield{author}{\bibinfo{person}{John Rawls}.} \bibinfo{year}{1971}\natexlab{}.
\newblock \bibinfo{booktitle}{\emph{A Theory of Justice}}.
\newblock \bibinfo{publisher}{Harvard University Press}, \bibinfo{address}{Cambridge, MA, USA}.
\newblock


\bibitem[\protect\citeauthoryear{Saha and Getoor}{Saha and Getoor}{2009}]%
        {SahaG09}
\bibfield{author}{\bibinfo{person}{Barna Saha} {and} \bibinfo{person}{Lise Getoor}.} \bibinfo{year}{2009}\natexlab{}.
\newblock \showarticletitle{On Maximum Coverage in the Streaming Model {\&} Application to Multi-topic Blog-Watch}. In \bibinfo{booktitle}{\emph{Proceedings of the 2009 SIAM International Conference on Data Mining (SDM)}}. \bibinfo{publisher}{Society for Industrial and Applied Mathematics}, \bibinfo{pages}{697--708}.
\newblock


\bibitem[\protect\citeauthoryear{Serbos, Qi, Mamoulis, Pitoura, and Tsaparas}{Serbos et~al\mbox{.}}{2017}]%
        {SerbosQMPT17}
\bibfield{author}{\bibinfo{person}{Dimitris Serbos}, \bibinfo{person}{Shuyao Qi}, \bibinfo{person}{Nikos Mamoulis}, \bibinfo{person}{Evaggelia Pitoura}, {and} \bibinfo{person}{Panayiotis Tsaparas}.} \bibinfo{year}{2017}\natexlab{}.
\newblock \showarticletitle{Fairness in Package-to-Group Recommendations}. In \bibinfo{booktitle}{\emph{Proceedings of the 26th International Conference on World Wide Web}} \emph{(\bibinfo{series}{WWW '17})}. \bibinfo{publisher}{Association for Computing Machinery}, \bibinfo{address}{New York, NY, USA}, \bibinfo{pages}{371--379}.
\newblock


\bibitem[\protect\citeauthoryear{Tang, Tang, Lim, Han, Li, and Yuan}{Tang et~al\mbox{.}}{2021}]%
        {TangTLHLY21}
\bibfield{author}{\bibinfo{person}{Jing Tang}, \bibinfo{person}{Xueyan Tang}, \bibinfo{person}{Andrew Lim}, \bibinfo{person}{Kai Han}, \bibinfo{person}{Chongshou Li}, {and} \bibinfo{person}{Junsong Yuan}.} \bibinfo{year}{2021}\natexlab{}.
\newblock \showarticletitle{Revisiting Modified Greedy Algorithm for Monotone Submodular Maximization with a Knapsack Constraint}.
\newblock \bibinfo{journal}{\emph{Proc. {ACM} Meas. Anal. Comput. Syst.}} \bibinfo{volume}{5}, \bibinfo{number}{1}, Article \bibinfo{articleno}{08} (\bibinfo{year}{2021}), \bibinfo{numpages}{22}~pages.
\newblock


\bibitem[\protect\citeauthoryear{Tang, Shi, and Xiao}{Tang et~al\mbox{.}}{2015}]%
        {TangSX15}
\bibfield{author}{\bibinfo{person}{Youze Tang}, \bibinfo{person}{Yanchen Shi}, {and} \bibinfo{person}{Xiaokui Xiao}.} \bibinfo{year}{2015}\natexlab{}.
\newblock \showarticletitle{Influence Maximization in Near-Linear Time: A Martingale Approach}. In \bibinfo{booktitle}{\emph{Proceedings of the 2015 ACM SIGMOD International Conference on Management of Data}} \emph{(\bibinfo{series}{SIGMOD '15})}. \bibinfo{publisher}{Association for Computing Machinery}, \bibinfo{address}{New York, NY, USA}, \bibinfo{pages}{1539--1554}.
\newblock


\bibitem[\protect\citeauthoryear{Thejaswi, Ordozgoiti, and Gionis}{Thejaswi et~al\mbox{.}}{2021}]%
        {ThejaswiOG21}
\bibfield{author}{\bibinfo{person}{Suhas Thejaswi}, \bibinfo{person}{Bruno Ordozgoiti}, {and} \bibinfo{person}{Aristides Gionis}.} \bibinfo{year}{2021}\natexlab{}.
\newblock \showarticletitle{Diversity-Aware k-median: Clustering with Fair Center Representation}. In \bibinfo{booktitle}{\emph{Machine Learning and Knowledge Discovery in Databases. Research Track -- European Conference, ECML PKDD 2021, Bilbao, Spain, September 13–17, 2021, Proceedings, Part {II}}}. \bibinfo{publisher}{Springer}, \bibinfo{address}{Cham}, \bibinfo{pages}{765--780}.
\newblock


\bibitem[\protect\citeauthoryear{Torrico, Singh, Pokutta, Haghtalab, Naor, and Anari}{Torrico et~al\mbox{.}}{2021}]%
        {AnariHNPST19}
\bibfield{author}{\bibinfo{person}{Alfredo Torrico}, \bibinfo{person}{Mohit Singh}, \bibinfo{person}{Sebastian Pokutta}, \bibinfo{person}{Nika Haghtalab}, \bibinfo{person}{Joseph~(Seffi) Naor}, {and} \bibinfo{person}{Nima Anari}.} \bibinfo{year}{2021}\natexlab{}.
\newblock \showarticletitle{Structured Robust Submodular Maximization: Offline and Online Algorithms}.
\newblock \bibinfo{journal}{\emph{{INFORMS} J. Comput.}} \bibinfo{volume}{33}, \bibinfo{number}{4} (\bibinfo{year}{2021}), \bibinfo{pages}{1590--1607}.
\newblock


\bibitem[\protect\citeauthoryear{Tsang, Wilder, Rice, Tambe, and Zick}{Tsang et~al\mbox{.}}{2019}]%
        {TsangWRTZ19}
\bibfield{author}{\bibinfo{person}{Alan Tsang}, \bibinfo{person}{Bryan Wilder}, \bibinfo{person}{Eric Rice}, \bibinfo{person}{Milind Tambe}, {and} \bibinfo{person}{Yair Zick}.} \bibinfo{year}{2019}\natexlab{}.
\newblock \showarticletitle{Group-Fairness in Influence Maximization}. In \bibinfo{booktitle}{\emph{Proceedings of the Twenty-Eighth International Joint Conference on Artificial Intelligence, {IJCAI-19}}}. \bibinfo{publisher}{International Joint Conferences on Artificial Intelligence Organization}, \bibinfo{pages}{5997--6005}.
\newblock


\bibitem[\protect\citeauthoryear{Udwani}{Udwani}{2018}]%
        {Udwani18}
\bibfield{author}{\bibinfo{person}{Rajan Udwani}.} \bibinfo{year}{2018}\natexlab{}.
\newblock \showarticletitle{Multi-objective Maximization of Monotone Submodular Functions with Cardinality Constraint}.
\newblock \bibinfo{journal}{\emph{Advances in Neural Information Processing Systems}}  \bibinfo{volume}{31} (\bibinfo{year}{2018}), \bibinfo{pages}{9513--9524}.
\newblock


\bibitem[\protect\citeauthoryear{Wang, Fabbri, and Mathioudakis}{Wang et~al\mbox{.}}{2021}]%
        {WangFM21}
\bibfield{author}{\bibinfo{person}{Yanhao Wang}, \bibinfo{person}{Francesco Fabbri}, {and} \bibinfo{person}{Michael Mathioudakis}.} \bibinfo{year}{2021}\natexlab{}.
\newblock \showarticletitle{Fair and Representative Subset Selection from Data Streams}. In \bibinfo{booktitle}{\emph{Proceedings of the Web Conference 2021}} \emph{(\bibinfo{series}{WWW '21})}. \bibinfo{publisher}{Association for Computing Machinery}, \bibinfo{address}{New York, NY, USA}, \bibinfo{pages}{1340--1350}.
\newblock


\bibitem[\protect\citeauthoryear{Wang, Fan, Li, and Tan}{Wang et~al\mbox{.}}{2017}]%
        {WangFLT17}
\bibfield{author}{\bibinfo{person}{Yanhao Wang}, \bibinfo{person}{Qi Fan}, \bibinfo{person}{Yuchen Li}, {and} \bibinfo{person}{Kian{-}Lee Tan}.} \bibinfo{year}{2017}\natexlab{}.
\newblock \showarticletitle{Real-Time Influence Maximization on Dynamic Social Streams}.
\newblock \bibinfo{journal}{\emph{Proc. {VLDB} Endow.}} \bibinfo{volume}{10}, \bibinfo{number}{7} (\bibinfo{year}{2017}), \bibinfo{pages}{805--816}.
\newblock


\bibitem[\protect\citeauthoryear{Wang, Li, and Tan}{Wang et~al\mbox{.}}{2019}]%
        {WangLT19}
\bibfield{author}{\bibinfo{person}{Yanhao Wang}, \bibinfo{person}{Yuchen Li}, {and} \bibinfo{person}{Kian{-}Lee Tan}.} \bibinfo{year}{2019}\natexlab{}.
\newblock \showarticletitle{Efficient Representative Subset Selection over Sliding Windows}.
\newblock \bibinfo{journal}{\emph{{IEEE} Trans. Knowl. Data Eng.}} \bibinfo{volume}{31}, \bibinfo{number}{7} (\bibinfo{year}{2019}), \bibinfo{pages}{1327--1340}.
\newblock


\bibitem[\protect\citeauthoryear{Wei, Iyer, Wang, Bai, and Bilmes}{Wei et~al\mbox{.}}{2015}]%
        {WeiIWBB15}
\bibfield{author}{\bibinfo{person}{Kai Wei}, \bibinfo{person}{Rishabh~K. Iyer}, \bibinfo{person}{Shengjie Wang}, \bibinfo{person}{Wenruo Bai}, {and} \bibinfo{person}{Jeff~A. Bilmes}.} \bibinfo{year}{2015}\natexlab{}.
\newblock \showarticletitle{Mixed Robust/Average Submodular Partitioning: Fast Algorithms, Guarantees, and Applications}.
\newblock \bibinfo{journal}{\emph{Advances in Neural Information Processing Systems}}  \bibinfo{volume}{28} (\bibinfo{year}{2015}), \bibinfo{pages}{2233--2241}.
\newblock


\bibitem[\protect\citeauthoryear{Wolsey}{Wolsey}{2020}]%
        {ip-book}
\bibfield{author}{\bibinfo{person}{Laurence Wolsey}.} \bibinfo{year}{2020}\natexlab{}.
\newblock \bibinfo{booktitle}{\emph{Integer Programming, Second Edition}}.
\newblock \bibinfo{publisher}{John Wiley \& Sons, Inc.}, \bibinfo{address}{Hoboken, NJ, USA}.
\newblock


\bibitem[\protect\citeauthoryear{Wolsey}{Wolsey}{1982}]%
        {Wolsey82}
\bibfield{author}{\bibinfo{person}{Laurence~A. Wolsey}.} \bibinfo{year}{1982}\natexlab{}.
\newblock \showarticletitle{An analysis of the greedy algorithm for the submodular set covering problem}.
\newblock \bibinfo{journal}{\emph{Combinatorica}} \bibinfo{volume}{2}, \bibinfo{number}{4} (\bibinfo{year}{1982}), \bibinfo{pages}{385--393}.
\newblock


\bibitem[\protect\citeauthoryear{Yang, Zhang, Zheng, and Yu}{Yang et~al\mbox{.}}{2015}]%
        {YangZZY15}
\bibfield{author}{\bibinfo{person}{Dingqi Yang}, \bibinfo{person}{Daqing Zhang}, \bibinfo{person}{Vincent~W. Zheng}, {and} \bibinfo{person}{Zhiyong Yu}.} \bibinfo{year}{2015}\natexlab{}.
\newblock \showarticletitle{Modeling User Activity Preference by Leveraging User Spatial Temporal Characteristics in LBSNs}.
\newblock \bibinfo{journal}{\emph{{IEEE} Trans. Syst. Man Cybern. Syst.}} \bibinfo{volume}{45}, \bibinfo{number}{1} (\bibinfo{year}{2015}), \bibinfo{pages}{129--142}.
\newblock


\end{thebibliography}

\end{document}